\PassOptionsToPackage{hidelinks}{hyperref}
\documentclass[AMS,Times1COL]{WileyNJDv5_arx} %STIX1COL,STIX2COL,STIXSMALL

\usepackage{silence}
\articletype{Article Type}%

\usepackage{amsmath}
\usepackage{amsfonts}

\usepackage{newtxtext,newtxmath}
\usepackage{color}
\usepackage{textcomp}
\usepackage{graphicx}
\usepackage{epsfig}
\usepackage{color,soul}
\usepackage{mathrsfs}
\usepackage{xcolor}
\usepackage[breakable]{tcolorbox}
\usepackage{mwe}
\usepackage{subcaption}
\usepackage{tabularx}
\usepackage{placeins}

\newcommand{\T}{^{\mbox{\tiny T}}}
\newcommand{\Tcc}{^{\vphantom{\mbox{\tiny T}}}}
\usepackage{url}

\newcommand{\eps}{\varepsilon}

\DeclareMathOperator{\Ker}{Ker}
\DeclareMathOperator{\Span}{span}
\newcommand{\R}{\mathbb{R}}

\newcommand{\N}{\mathbb{N}}
\let\leq\leqslant
\let\geq\geqslant

\newenvironment{proof}[1][Proof]%
{\par\noindent\textit{#1:\ }}%
{\hspace*{\fill} \rule{6pt}{6pt}}

\newenvironment{proof*}[1][Proof]%
{\par\noindent\textit{#1:\ }}{}

\renewcommand{\theenumi}{\arabic{enumi}}

\DeclareMathOperator{\im}{Im}

\IfFileExists{wasysym.sty}%
{
  
  \usepackage{wasysym}%
  }%
{}
\newenvironment{system}[1]%
{\setlength{\arraycolsep}{0.5mm}\left\{ \; \begin{array}{#1}}%
    {\end{array} \right.}
\newenvironment{system*}[1]%
{\setlength{\arraycolsep}{0.5mm} \begin{array}{#1}}%
  {\end{array}}

\usepackage{tikz}
\usetikzlibrary{shapes,calc,arrows,patterns,decorations.pathmorphing
  ,decorations.markings}
\usetikzlibrary{arrows.meta}

\received{Date Month Year}
\revised{Date Month Year}
\accepted{Date Month Year}
\journal{Journal}
\volume{00}
\copyyear{2025}
\startpage{1}

\makeatletter
\providecommand{\oddfoot@titlepage@info}{}
\makeatother

\begin{document}

\title{Scale-Free \texorpdfstring{\boldmath $\delta$-Level}\ \
  Coherent Output Synchronization of Multi-Agent Systems with Adaptive
  Protocols and Bounded Disturbances}
 
\author[4]{Anton A. Stoorvogel}
\author[1]{Ali Saberi}
\author[1]{Donya Nojavanzadeh*}

\authormark{A.A. Stoorvogel et al.}

\address[4]{\orgdiv{Department of Electrical Engineering, Mathematics
    and Computer Science}, \orgname{University of Twente},
  \orgaddress{\state{Enschede}, \country{The Netherlands}}} 
\address[1]{\orgdiv{School of Electrical Engineering and Computer
    Science}, \orgname{Washington State University},
  \orgaddress{\state{Washington}, \country{USA}}}

\corres{*Donya Nojavanzadeh. \email{donya.nojavanzadeh@wsu.edu}}

\presentaddress{ATTN: Donya Nojavanzadeh, Franklin Ave, Tustin, CA, 92780, USA}

\abstract[Abstract]{In this paper, we investigate scale-free
  $\delta-$level coherent output synchronization for multi-agent
  systems (MAS) operating under bounded disturbances/noises. We
  introduce an adaptive scale-free framework designed solely based on
  the knowledge of agent models and completely agnostic to both the
  communication topology and the size of the network. We define the
  level of coherency for each agent as the norm of the weighted sum of
  the disagreement dynamics with its neighbors.
  
  We define each agent’s coherency level as the norm of a weighted sum
  of its disagreement dynamics relative to its neighbors. The goal is
  to ensure that the network’s coherency level remains below a
  prescribed threshold $\delta$, without requiring any \emph{a-priori}
  knowledge of the disturbance.}


\jnlcitation{\cname{\author{Stoorvogel A. A.},
    \author{Saberi A.}, and  \author{Nojavanzadeh D.}} (\cyear{2025}), 
  \ctitle{Scale-free $\delta$-level output synchronization of
    multi-agent systems with adaptive protocols and bounded
    disturbances} }  
\maketitle

\section{Introduction}

Synchronization and consensus problems in multi-agent systems (MAS)
have become a major research focus in recent years, driven by their
broad applications in cooperative control, including robot networks,
autonomous vehicles, distributed sensor networks, and power
systems. The goal of synchronization in MAS is to achieve asymptotic
agreement on a common state or output trajectory through local
interactions between agents; see
\cite{bai-arcak-wen,mesbahi-egerstedt,ren-book,wu-book,%
  liu-nojavanzedah-saberi-2022-book,saberi-stoorvogel-zhang-sannuti}
and the references therein.

In a network with full-state coupling, each agent has access to a
linear combination of its own state relative to the states of its
neighbors. If this linear combination involves only a subset of the
states relative to the corresponding states of the neighboring agents,
it is called partial-state coupling. For state synchronization under
diffusive full-state coupling, research has evolved from simple
single- and double-integrator models (\emph{e.g.},
\cite{saber-murray2,ren,ren-beard}) to more general agent dynamics
(\emph{e.g.}, \cite{scardovi-sepulchre,tuna1,wieland-kim-allgower}). State
synchronization under diffusive partial-state coupling has also been
widely investigated, including static controller design
\cite{liu-stoorvogel-saberi-nojavanzadeh-ijrnc19,liu-zhang-saberi-stoorvogel-auto},
dynamic controller design
\cite{kim-shim-back-seo,seo-back-kim-shim-iet,seo-shim-back,su-huang-tac,tuna3},
and approaches based on localized information exchange
\cite{chowdhury-khalil,scardovi-sepulchre}. More recently, we have
developed a new class of scale-free protocols for synchronization and
almost synchronization of homogeneous and heterogeneous MAS under
various practical challenges, including external disturbances, input
saturation, communication delays, and input delays; see, for example,
\cite{liu-nojavanzadeh-saberi-saberi-stoorvogel-scl,liu-donya-dmitri-saberi-stoorvogel-ccdc2020,liu-saberi-stoorvogel-nojavanzadeh-cdc19,donya-liu-saberi-stoorvogel-ACC2020}.

Synchronization and almost synchronization in the presence of external
disturbances have been widely studied in the literature. Three main
classes of disturbances are typically considered: \vspace{-3mm}
\begin{enumerate}
\item Disturbances and measurement noises with known frequencies
\item Deterministic disturbances with finite power
\item Stochastic disturbances with bounded variance
\end{enumerate}

For disturbances and measurement noise with known frequencies, it has
been shown in
\cite{zhang-saberi-stoorvogel-ACC2015,zhang-saberi-stoorvogel-CDC2016}
that exact synchronization can be achieved. Specifically,
\cite{zhang-saberi-stoorvogel-ACC2015} establishes this result for
heterogeneous MAS composed of minimum-phase, non-introspective agents
operating over time-varying directed communication graphs. The work in
\cite{zhang-saberi-stoorvogel-CDC2016} further extends these results
to non-minimum-phase agents by employing localized information
exchange.
For deterministic disturbances with finite power, the notion of
$H_\infty$ almost synchronization was introduced by Peymani \emph{et al.},
for homogeneous MAS with non-introspective agents, using additional
communication exchange \cite{peymani-grip-saberi}. The objective of
$H_\infty$ almost synchronization is to attenuate the effect of
disturbances on the synchronization error to an arbitrarily small
level, quantified in the $H_\infty$ norm. This line of work was later
extended in
\cite{peymani-grip-saberi-wang-fossen,zhang-saberi-grip-stoorvogel,zhang-saberi-stoorvogel-sannuti2}
to heterogeneous MAS with non-introspective agents, without requiring
additional communication and accommodating time-varying communication
graphs. $H_\infty$ almost synchronization via static protocols is
studied in \cite{stoorvogel-saberi-liu-nojavanzadeh-ijrnc19} for MAS
with passive and passifiable agents. More recently,
\cite{stoorvogel-saberi-zhang-liu-ejc} established necessary and
sufficient conditions for the solvability of $H_\infty$ almost
synchronization in homogeneous networks with non-introspective agents
and without extra communication exchange. Finally, we have developed a
scale-free framework for $H_\infty$ almost state synchronization in
homogeneous networks, leveraging suitably designed localized
information exchange
\cite{liu-saberi-stoorvogel-donya-almost-automatica}.

For stochastic disturbances with bounded variance, the concept of
stochastic almost synchronization was introduced in
\cite{zhang-saberi-stoorvogel-stochastic2}, where both stochastic
disturbances and disturbances with known frequencies may be
present. The goal of stochastic almost synchronization is to make the
stochastic RMS norm of the synchronization error arbitrarily small,
despite the presence of colored stochastic disturbances modeled as
outputs of linear time-invariant systems driven by white noise with
unit power spectral intensity. By augmenting the disturbance model
with the agent dynamics, one may equivalently treat the stochastic
disturbances as white noise with unit power spectral intensity. Under
linear protocols, the stochastic RMS norm of the synchronization error
then corresponds to the $H_2$ norm of the transfer function from the
disturbance to the synchronization error. Consequently, the stochastic
almost synchronization problem can be reformulated as a deterministic
$H_2$ almost synchronization problem, where the objective is to make
this $H_2$ norm arbitrarily small. These two formulations are
equivalent. Recent work on the $H_2$ almost synchronization problem
includes \cite{stoorvogel-saberi-zhang-liu-ejc}, which provides
necessary and sufficient conditions for solvability in homogeneous
networks with non-introspective agents and without additional
communication exchange. $H_2$ almost synchronization using static
protocols has also been studied in
\cite{stoorvogel-saberi-liu-nojavanzadeh-ijrnc19} for MAS composed of
passive and passifiable agents.

The above $H_\infty$ and $H_2$ almost synchronization approaches for
MAS, however, suffer from the following disadvantages: \vspace{-3mm}
\renewcommand\labelitemi{{\boldmath$\bullet$}}
\begin{itemize}
\item \textit{Tuning requirement:} The protocols developed for
  $H_\infty$ and $H_2$ almost synchronization are parameterized by a
  tuning parameter. Although the $H_\infty$ or $H_2$ norm of the
  transfer function from the external disturbances can be made
  arbitrarily small through appropriate tuning, this relationship
  depends heavily on the structure of the communication graph.
\item \textit{Dependency on disturbance size:} In $H_\infty$ and $H_2$
  almost synchronization, the size of the synchronization error scales
  linearly with the magnitude of the disturbance. Therefore, to
  guarantee a desired coherence level, prior knowledge of the
  disturbance magnitude is required to properly select the tuning
  parameter $\delta$.
\end{itemize}
On the other hand, in
\cite{donya-liu-saberi-stoorvogel-arxiv-delta} we have studied scale-free $\delta$-level coherent state synchronization for homogeneous MAS under bounded
disturbances/noises. This approach ensures
a guaranteed level of coherency—\emph{ i.e.}, a specified degree of
synchronization—without requiring prior knowledge of the network
topology or the magnitude of the disturbances. In this paper, we
consider scale-free $\delta$-level coherent \textbf{output} synchronization for
homogeneous MAS under bounded disturbances/noises. The contribution of
this work is threefold.  \vspace{-3mm}
\begin{enumerate}
\item The protocols are designed solely based on the knowledge of the
  agent models, using no information about the communication network
  such as bounds on the spectrum of the associated Laplacian matrix or
  the number of agents. That is to say, universal nonlinear protocols
  are scale-free and work for any communication network. Note that
  there is no need to impose restrictions on the connectivity of the
  graph.
\item We achieve scale-free $\delta-$level coherent output
  synchronization for MAS in the presence of bounded
  disturbances/noises such that, for any given $\delta$, one can
 restrict the level of coherency of the network to $\delta$. The only
  assumption is that the disturbances are bounded (which is a very
  reasonable condition). However, the protocol is independent of the bound
  and does not require any other knowledge about the disturbances.
\item We propose two types of adaptive protocols, \emph{i.e.,}
  collaborative and noncollaborative protocols to achieve
  $\delta$-level coherent output synchronization where in the case of
  collaborative protocols we significantly release the assumptions on
  the agent models.
\end{enumerate}

Note that we only consider disturbances to the different agents and
not measurement noises. 

\subsection*{Preliminaries on graph theory}

Given a matrix $M\in \mathbb{R}^{m\times n}$, $M\T$ denotes its
conjugate transpose. A square matrix $M$ is said to be Hurwitz stable
if all its eigenvalues are in the open left half complex plane.
$\im M$ denotes the image of matrix $M$. $M_a\otimes M_b$ depicts the
Kronecker product of $M_a$ and $M_b$. $I_n$ denotes the
$n$-dimensional identity matrix and $0_n$ denotes $n\times n$ zero
matrix; sometimes we drop the subscript if the dimension is clear from
the context. For a signal $u$, we denote the $L_2$ norm by $||u||$ or
$||u||_2$, and the $L_\infty$ norm by $||u||_\infty$.

To describe the information flow among the agents we associate a
\emph{weighted graph} $\mathcal{G}$ to the communication network. The
weighted graph $\mathcal{G}$ is defined by a triple
$(\mathcal{V}, \mathcal{E}, \mathcal{A})$ where
$\mathcal{V}=\{1,\ldots, N\}$ is a node set, $\mathcal{E}$ is a set of
pairs of nodes indicating connections among nodes, and
$\mathcal{A}=[a_{ij}]\in \mathbb{R}^{N\times N}$ is the weighted
adjacency matrix with non negative elements $a_{ij}$. Each pair in
$\mathcal{E}$ is called an \emph{edge}, where $a_{ij}>0$ denotes an
edge $(j,i)\in \mathcal{E}$ from node $j$ to node $i$ with weight
$a_{ij}$. Moreover, $a_{ij}=0$ if there is no edge from node $j$ to
node $i$. We assume there are no self-loops, \emph{i.e.} we have
$a_{ii}=0$. A \emph{path} from node $i_1$ to $i_k$ is a sequence of
nodes $\{i_1,\ldots, i_k\}$ such that $(i_j, i_{j+1})\in \mathcal{E}$
for $j=1,\ldots, k-1$. A \emph{directed tree} is a subgraph (a subset of
nodes and edges) in which every node has exactly one parent node
except for one node, called the \emph{root}, which has no parent
node. A \emph{directed spanning tree} is a subgraph which is a
directed tree containing all the nodes of the original graph. If a
directed spanning tree exists, the root has a directed path to every
other node in the tree \cite{royle-godsil}.

For a weighted graph $\mathcal{G}$, the matrix
$L=[\ell_{ij}]$ with
\[
\ell_{ij}=
\begin{system}{cl}
\sum_{k=1}^{N} a_{ik}, & i=j,\\
-a_{ij}, & i\neq j,
\end{system}
\]
is called the \emph{Laplacian matrix} associated with the graph
$\mathcal{G}$. The Laplacian matrix $L$ has all its eigenvalues in the
closed right half plane and at least one eigenvalue at zero associated
with right eigenvector $\textbf{1}$ \cite{royle-godsil}. Moreover, if
the graph contains a directed spanning tree, the Laplacian matrix $L$
has a single eigenvalue at the origin and all other eigenvalues are
located in the open right-half complex plane \cite{ren-book}.
In the absence of a directed spanning tree, the Laplacian matrix of
the graph has an eigenvalue at the origin with a multiplicity $k$
larger than $1$. This implies that it is a $k$-reducible matrix and
the graph has $k$ basic bi-components.  The book \cite[Definition
2.19]{wu-book} shows that, after a suitable permutation of the nodes,
a Laplacian matrix with $k$ basic bi-components can be written in the
following form:
\begin{equation}\label{Lstruc}
	L=\begin{pmatrix}
		L_0 & L_{01}     & L_{02} & \cdots  & L_{0k} \\
		0   & L_1        & 0      & \cdots  & 0 \\
		\vdots & \ddots  & \ddots & \ddots  & \vdots \\
		\vdots &         & \ddots & L_{k-1} & 0 \\
		0      & \cdots & \cdots  & 0       & L_k
	\end{pmatrix}
\end{equation}
where $L_1,\ldots, L_k$ are the Laplacian matrices associated to the
$k$ basic bi-components $\{ \mathcal{B}_1,\ldots, \mathcal{B}_k\}$ in
our network. These matrices have a simple eigenvalue in $0$ because
they are associated with a strongly connected component. On the other
hand, $L_0$ contains all non-basic bi-components and is a grounded
Laplacian with all eigenvalues in the open right-half plane. After
all, if $L_0$ would be singular then the network would have an
additional basic bi-component.

\section{Problem formulation}

Consider a MAS consisting of $N$ identical linear
agents
\begin{equation}\label{agent-g-noise}
  \begin{system*}{ccl}
    \dot{x}_i &=& Ax_i+Bu_i+Ew_i,\quad i=1,\ldots, N \\
    y_i &=& C x_i
  \end{system*}
\end{equation}
where $x_i\in\mathbb{R}^n$, $u_i\in\mathbb{R}^m$,
$y_i\in\mathbb{R}^p$ and
$w_i\in\mathbb{R}^w$ are state, input, output and external disturbance/noise, respectively.

The communication network is such that each agent observes a weighted
combination of its own output relative to the output of other agents,
\emph{i.e.}, for the protocol of agent $i$, the signal
\begin{equation}\label{zeta1noise}
  \zeta_i=\sum_{j=1}^{N}a_{ij}(y_i-y_j)
\end{equation}
is available where $a_{ij}\geq0$ and $a_{ii}=0$. As explained before,
the matrix $\mathcal{A}=[a_{ij}]$ is the weighted adjacency matrix of a directed
graph $\mathcal{G}$. This matrix describes the communication
topology of the network, where the nodes of network correspond to the
agents. We can also express the dynamics in terms of an associated
Laplacian matrix $L=[\ell_{ij}]_{N\times N}$, such that the signal
$\zeta_i$ in \eqref{zeta1noise} can be rewritten in the following form
\begin{equation}\label{zetanoise}
  \zeta_i= \sum_{j=1}^{N}\ell_{ij}y_j.
\end{equation}
The size of $\zeta_i(t)$ can be viewed as the level of coherency at agent $i$.

We define the set of communication graphs considered in this paper as follows.

\begin{definition}\label{def1} 
  $\mathbb{G}^N$ denotes the set of directed graphs of $N$ agents
\end{definition}

Next, in the following definition, we define the concept of
$\delta$-level-coherent output synchronization for the MAS with agents
\eqref{agent-g-noise} and communication information \eqref{zetanoise}.

\begin{definition}\label{delta-level}
  For any given $\delta>0$, the MAS \eqref{agent-g-noise} and
  \eqref{zetanoise} achieves $\delta$-level coherent output
  synchronization if there exists a $T>0$ such that
  \[
    \|\zeta_i(t)\| \le\delta,
  \]
  for all $t>T$, for all $i\in\{1,\ldots, N\}$, and for any bounded disturbance.
\end{definition}

\begin{problem}\label{prob_x}
  Consider a MAS \eqref{agent-g-noise} with associated network
  communication \eqref{zetanoise} and a given parameter
  $\delta>0$. The \textbf{scale-free $\delta$-level-coherent output
    synchronization in the presence of bounded external
    disturbances} is to find, if possible, a fully distributed
  nonlinear noncollaborative protocol using only knowledge of agent
  models, \emph{i.e.}, $(A, B, C)$, and $\delta$ of the form
  \begin{equation}\label{out_dyn}
    \begin{system}{cl}
      \dot{x}_{i,c}&=f(x_{i,c},\zeta_i),\\
      u_i&=g(x_{i,c},\zeta_i),
    \end{system}
  \end{equation}
  where $x_{i,c}\in\mathbb{R}^{n_c}$ is the state of protocol, such
  that the MAS with the above protocol achieves
  $\delta$-level-coherent output synchronization in the presence of
  disturbances with any size of the network $N$ and for any graph
  $\mathscr{G}\in\mathbb{G}^N$. In other words, the MAS achieves $\delta$-level
  coherent output synchronization as defined in Definition
  \ref{delta-level}.
\end{problem}

Next, we have Problem \ref{prob_xcol}, where for collaborative
protocols we assume there is an additional communication available
over the same network between the protocols of the different
agents. In particular, for the protocol of agent $i$ besides
\eqref{zetanoise}, they can also use
\begin{equation}\label{zetanoiseu}
  \tilde{\zeta}_i= \sum_{j=1}^{N}\ell_{ij} x_{i,c}
\end{equation}
which communicates the input signals of the different agents over the network.
This extra communication significantly reduces the number of
assumptions we need to make on our agent model.

\begin{definition}\label{delta-level2}
  For any given $\delta>0$, the MAS \eqref{agent-g-noise} and
  \eqref{zetanoise} achieves collaborative $\delta$-level-coherent output
  synchronization if there exist a $T>0$ such that
  \[
    \|\zeta_i(t)\| \le\delta,\qquad    \| C\tilde{\zeta}_i(t)\| \le\delta,
  \]
  for all $t>T$, for all $i\in\{1,\ldots, N\}$, and for any bounded disturbance.
\end{definition}

\begin{problem}\label{prob_xcol}
  Consider a MAS \eqref{agent-g-noise} with associated network
  communication \eqref{zetanoise}, \eqref{zetanoiseu} and a given
  parameter $\delta>0$. The \textbf{scale-free $\delta$-level-coherent
    output synchronization in the presence of bounded external
    disturbances} is to find, if possible, a fully distributed
  nonlinear collaborative protocol using only knowledge of agent
  models, \emph{i.e.}, $(A, B, C)$, and $\delta$ of the form
  \begin{equation}\label{out_dyncol}
    \begin{system}{cl}
      \dot{x}_{i,c}&=f(x_{i,c},\zeta_i,\tilde{\zeta}_i),\\
      u_i&=g(x_{i,c},\zeta_i,\tilde{\zeta}_i),
    \end{system}
  \end{equation}
  where $x_{i,c}\in\mathbb{R}^{n_c}$ is the state of protocol, such
  that the MAS with the above protocol achieves collaborative
  $\delta$-level-coherent output synchronization in the presence of
  disturbances with any size of the network $N$ and for any
  graph $\mathscr{G}\in\mathbb{G}^N$. In other words, the MAS achieves
  collaborative $\delta$-level coherent output synchronization as
  defined in Definition \ref{delta-level2}.
\end{problem}

\section{Noncollaborative Protocol design}

In this section, we design an adaptive protocol to achieve the
objectives of Problem \ref{prob_x}. We make the following assumption. 
\begin{assumption}\label{ass}\mbox{}
\vspace{-3mm}
  \begin{enumerate}
  \item\label{1.1} $(A,B)$ is stabilizable and $(C,A)$ is detectable.
  \item\label{1.2}  $\im E \subset \im B$.
  \item $(A,B,C)$ has relative degree $1$.
  \item $(A,B,C)$ is left-invertible.
  \item\label{1.5}  $(A,B,C)$ is minimum-phase.
  \item\label{1.6}  The disturbances $w_i$ are bounded for
    $i\in \{1,2,\ldots, N\}$. In other words, we have that
    $\| w \|_\infty < \infty$ for $i\in \{1,2,\ldots, N\}$.
  \end{enumerate}
\end{assumption}

\begin{remark}
  Assumptions \ref{ass}.\ref{1.1} is an obvious necessary condition.  As argued
  before, Assumption \ref{ass}.\ref{1.6} is basically valid for any disturbance in the
  real world. The key point is that we do not need to know the bound.
  The assumptions \ref{ass}.\ref{1.2}-\ref{ass}.\ref{1.5} are quite strong and we will see
  later that, using collaborative protocols, these assumptions can be
  relaxed quite significantly. 
\end{remark}

We design the nonlinear noncollaborative adaptive protocol as follows.

\begin{tcolorbox}[breakable,colback=white]
  We choose invertible matrices $S$ and $T$ such that
  \[
    \tilde{B}=SB = \begin{pmatrix} 0 \\ B_2 \end{pmatrix},\qquad
    \tilde{C}=TCS^{-1} = \begin{pmatrix} C_1 & 0 \\ 0 & I \end{pmatrix} 
  \]
  with $B_2$ invertible which is possible since $B$ and $CB$ are both
  left-invertible. Note that $\im E \subset \im B$ implies that:
  \[
    SE = \begin{pmatrix} 0 \\ E_2 \end{pmatrix}
  \]    
  Define:
  \begin{equation}\label{STT}
    \begin{pmatrix}
      x_{1i}\\
      x_{2i}
    \end{pmatrix}=Sx_i,\qquad
    \begin{pmatrix}
      y_{1i}\\
      y_{2i}
    \end{pmatrix}=Ty_i,\qquad
    \begin{pmatrix}
      \zeta_{1i}\\
      \zeta_{2i}
    \end{pmatrix}=T\zeta_i,
  \end{equation}
  such that the dynamics of $x_{1i}$ and $x_{2i}$ are given
  by
  \begin{equation}\label{eq6}
    \begin{system}{cl}
      \dot{x}_{1i} &= A_{11}x_{1i}+A_{12}x_{2i},\\
      \dot{x}_{2i} &= A_{21}x_{1i}+A_{22}x_{2i}+B_2u_i+E_2w_i,\\
      Ty_i&=\begin{pmatrix}
        y_{1i}\\y_{2i}
      \end{pmatrix}=\begin{pmatrix}
        C_1x_{1i}\\ x_{2i}
      \end{pmatrix}.
    \end{system}
  \end{equation}
  Since the system is minimum-phase it can be easily verified that we
  must have that $(C_1,A_{11})$ is detectable. Choose $H_1$ such that
  $A_{11}+H_1C_1$ is asymptotically stable.

  We denote:
  \[
    \tilde{A}=\begin{pmatrix} A_{11} & A_{12} \\ A_{21} &
      A_{22} \end{pmatrix},\qquad
    \tilde{E}_1=\begin{pmatrix} 0 \\ E_2 \end{pmatrix}.
  \]
  Since $(\tilde{A},\tilde{B})$ is stabilizable, there exists a
  matrix $P>0$ satisfying the following algebraic Riccati equation
  \begin{equation}\label{eq-Riccati}
    \tilde{A}\T P + P\tilde{A} -P\tilde{B}\tilde{B}\T P + I =0.
  \end{equation}

  Define $\bar{\delta}<\delta_1=\delta^2\lambda_{\min}(P)$ and choose $d$ such
  that
  \begin{equation}\label{dchoice}
    0<d< \| CS^{-1} \|^{-1} \bar{\delta}.
  \end{equation}
  
  Finally, for any parameter $d$ satisfying \eqref{dchoice}, we design the
  following adaptive protocol
  \begin{equation}\label{protocol}
    \begin{system*}{ccl}
      \dot{\hat{\xi}}_{1i} &=& A_{11} \hat{\xi}_{1i} + A_{12} \zeta_{2i}
      + H_1 (C_1\hat{\xi}_{1i} - \zeta_{1i}) \\
      \hat{\xi}_i &=& \begin{pmatrix} \hat{\xi}_{1i} \\ \zeta_{2i} \end{pmatrix} \\
      \dot{\rho}_i &=& \begin{cases} \hat{\xi}_i\T
        P\tilde{B}\tilde{B}\T P \hat{\xi}_i & 
        \text{ if } \hat{\xi}_i\T P \hat{\xi}_i \geq d, \\
        0 & \text{ if } \hat{\xi}_i\T P \hat{\xi}_i < d,
      \end{cases} \\
    u_i &=& -\rho_i \tilde{B}\T P \hat{\xi}_i.
  \end{system*}
\end{equation}
\end{tcolorbox}

We have the following theorem.

\begin{theorem}\label{theorem}
  Consider a MAS \eqref{agent-g-noise}, satisfying assumption
  \ref{ass}, with associated network communication \eqref{zetanoise}
  and a given parameter $\delta>0$. Then, the \textbf{scale-free
    {\boldmath $\delta$}-level-coherent output synchronization in the presence of
    bounded external disturbances} as stated in Problem
  \ref{prob_x} is solvable. In particular, protocol \eqref{protocol}
  with any $d$ satisfying \eqref{dchoice} solves
  $\delta$-level-coherent output synchronization in the presence of
  disturbances $w_i$, for an arbitrary number of agents $N$ and
  for any graph $\mathscr{G}\in\mathbb{G}^N$.
\end{theorem}

The proof of the above theorem relies on three lemmas which are
presented below.

\begin{lemma}\label{lem1}
  Consider a number of agents $N$ and a graph
  $\mathscr{G}\in\mathbb{G}^N$. Consider MAS \eqref{agent-g-noise} with
  associated network communication \eqref{zetanoise} and a given
  parameter $\delta>0$. Assume Assumption \ref{ass} is satisfied.
  Choose any $d$ satisfying \eqref{dchoice}. If all $\rho_i$ remain
  bounded then there exists a $T>0$ such that
  \begin{equation}\label{probdef}
    \hat{\xi}_i\T(t) P \hat{\xi}_i(t) \leq \bar{\delta},
  \end{equation}
  for all $t>T$ and for all $i=1,\ldots, N$.
\end{lemma}

\begin{proof}
  The proof in this lemma follows in exactly the same way as the proof
  of \cite[Lemma 1]{donya-liu-saberi-stoorvogel-arxiv-delta}. The fact
  that we have an additional unmatched disturbance $e$ does not matter
  because it converges to zero exponentially.
\end{proof}

\begin{lemma}\label{lem2a}
  Consider MAS \eqref{agent-g-noise} with associated network
  communication \eqref{zetanoise} and protocol \eqref{protocol}.
  Assume Assumption \ref{ass} is satisfied. Additionally, assume that
  either the $\rho_i$ associated to agents belonging to the basic
  bi-components are bounded or the graph is strongly
  connected. In that case, there exists $\eta,\mu >0$ such that
  \[
    \rho_i (t_2) -\rho_i(t_1) < \eta (t_2-t_1) + \mu
  \]
  for any $t_1,t_2$ with $t_2>t_1$ and for $i=1,\ldots,N$. 
\end{lemma}

\begin{proof}
  The proof in this lemma follows in exactly the same way as the proof
  of \cite[Lemma 2]{donya-liu-saberi-stoorvogel-arxiv-delta}. The fact
  that we have an additional unmatched disturbance $e$ does not matter
  because it converges to zero exponentially.
\end{proof}

\begin{lemma}\label{lem2}
  Consider MAS \eqref{agent-g-noise} with associated network
  communication \eqref{zetanoise} and the protocol \eqref{protocol}.
  Assume Assumption \ref{ass} is satisfied. Additionally, assume that
  either the $\rho_i$ associated to agents belonging to the basic
  bi-components are bounded or the graph is strongly connected. In that
  case, all $\rho_i$ remain bounded.
\end{lemma}

\begin{proof}
  The proof in this lemma follows the same arguments as the proof of
  \cite[Lemma 3]{donya-liu-saberi-stoorvogel-arxiv-delta}. However, we need
  one additional property. Since \eqref{lem2a} shows that the
  $\rho_i$ can grow at most linearly over time while we already knew
  that $e$ converges to zero exponentially, we find that:
  \[
    \rho_i^2 e \in L_2,
  \]
  for $i=1,\ldots ,N$. Using this property the arguments of
  \cite[Lemma 3]{donya-liu-saberi-stoorvogel-arxiv-delta} establish
  that all the $\rho_i$ remain bounded.
\end{proof}

\begin{proof}[Proof of Theorem \ref{theorem}]
  Consider the protocol \eqref{protocol}. Define
  \[
    \xi_i=\sum_{i=1}^N \ell_{ij} Sx_j,\qquad
    \xi_{1i}=\sum_{i=1}^N \ell_{ij} x_{1j},\qquad
    \xi_{2i}=\sum_{i=1}^N \ell_{ij} x_{2j}.
  \]
  Note that we have that $C\xi_i=\zeta_i$,\ $C_1\xi_{1i}=\zeta_{1i}$ and $\xi_{2i}=\zeta_{2i}$. 
  We obtain:
  \[
    \dot{e}_{1i} = (A_{11}+H_1C_1)e_{1i},
  \]
  where $e_{1i}= \hat{\xi}_{1i}-\xi_{1i}$. Clearly $e_{1i}$ converges to zero
  exponentially and its behavior is independent of our scheduling or
  our external disturbance $w$. We define:
  \[
    e_i = \begin{pmatrix} e_{1i} \\ 0 \end{pmatrix},\qquad
    \tilde{E}_2=\begin{pmatrix} H_1C_1 \\ -A_{21} \end{pmatrix}
  \]
  and
  \[
    \hat{\xi}=\begin{pmatrix} \hat{\xi}_1 \\ \vdots \\
      \hat{\xi}_N \end{pmatrix},\qquad
    e=\begin{pmatrix} e_1 \\ \vdots \\
      e_N \end{pmatrix},\qquad
    w=\begin{pmatrix} w_1 \\ \vdots \\
      w_N \end{pmatrix}.\qquad 
  \]
  We obtain:
  \begin{equation}\label{closeq}
    \dot{\hat{\xi}} = (I\otimes A)\hat{\xi} - (L\rho\otimes
    \tilde{B}\tilde{B}\T P)\hat{\xi} + (L\otimes \tilde{E}_1) w +
    (I\otimes \tilde{E}_2) e.
  \end{equation}
  Note that we obtain the same structure as in \cite[equation
  (12)]{donya-liu-saberi-stoorvogel-arxiv-delta} except for the extra
  term involving $e$. The scheduling in \eqref{protocol} has exactly
  the same structure as in \cite[equation
  (8)]{donya-liu-saberi-stoorvogel-arxiv-delta}.

  We can not directly apply the proofs of
  \cite{donya-liu-saberi-stoorvogel-arxiv-delta} because of the extra
  term involving $e$ in \eqref{closeq}. Clearly $e$ can be viewed as
  an additional disturbance which is bounded (even converging to zero
  exponentially). However it is an unmatched disturbance. We have:
  \[
    \im \tilde{E}_1 \subset \im \tilde{B},
  \]
  but the same property does not hold for $\tilde{E}_2$.

  The Laplacian matrix of the system in general has the form
  \eqref{Lstruc}. We note that if we look at the dynamics of the
  agents belonging to one of the basic bi-components then these
  dynamics are not influenced by the other agents and hence can be
  analyzed independent of the rest of network. The network within one
  of the basic bi-components is strongly connected and we can apply
  Lemmas \ref{lem2a} and \ref{lem2} to guarantee that the $\rho_i$
  associated to a basic bi-component are all bounded.

  Next, we look at the full network again. We have already established
  that the $\rho_i$ associated to all basic bi-components are all
  bounded. Then we can again apply Lemmas \ref{lem2a} and \ref{lem2}
  to conclude that the other $\rho_i$ not associated to basic
  bi-components are also bounded.

  After having established that all the $\rho_i$ are all bounded, we
  can then apply Lemma \ref{lem1} to conclude that \eqref{probdef}
  is satisfied. Since $\hat{\xi}_i-\xi_i=e_i$ converges to zero
  exponentially,
  we can conclude there exists $T>0$ such that:
  \[
    \xi_i\T(t) P \xi_i(t) \leq \delta_1,
  \]
  for all $t>T$ and for all $i=1,\ldots, N$. Since
  $\zeta_i=CS^{-1}\xi_i$, this immediately implies:
  \[
    \|\zeta_i(t)\| \le \delta,
  \]
  for all $t>T$. Therefore, we have established that we achieve
  scale-free $\delta$-level-coherent output synchronization.
\end{proof}

\section{Collaborative Protocol design}

In this section, we design an adaptive collaborative protocol to achieve the
objectives of Problem \ref{prob_xcol}. We make the following assumption.
\begin{assumption}\label{asscol}\mbox{}
\vspace{-3mm}
  \begin{enumerate}
  \item $(A,B)$ is stabilizable and $(C,A)$ is observable.
  \item The system $(A,B,C)$ is right-invertible and
    minimum-phase.
  \item The disturbances $w_i$ are bounded for
    $i\in \{1,2,\ldots, N\}$. In other words, we have that
    $\| w \|_\infty < \infty$ for $i\in \{1,2,\ldots, N\}$.
  \end{enumerate}
\end{assumption}

\begin{remark}
 Compared to the design of the noncollaborative protocol, strong assumptions such as
  $\im E \subset \im B$ and relative degree $1$ have been removed
, although we now assume that $(A,B,C)$ is right-invertible. If
  $\im E \subset \im B$ is satisfied then we can easily show that the
  assumption of right-invertibility is no longer needed.
\end{remark}

We design the nonlinear collaborative adaptive protocol as follows.

\begin{tcolorbox}[colback=white]
  According to \cite{zhou-duan-lin}, there exists $\eta>0$ and $Q>0$ such that:
  \begin{equation}\label{are0}
    AQ+QA\T-QC\T CQ+\eta Q = 0.     
  \end{equation}
  We will use the following observer:
  \begin{equation} \label{protocol2.1}
    \dot{\hat{x}}_i=A\hat{x}_i+Bu_i-\rho_i QC\T
    (C\tilde{\zeta}_i-\zeta_i),
  \end{equation}
  with $\zeta_i$ given by \eqref{zetanoise} and $\tilde{\zeta}_i$ given by \eqref{zetanoiseu}
  with $x_{i,c}=\hat{x}_i$. Define:
  \[
    e_i=\hat{x}_i-x_i,\qquad \xi_i = \sum_{j=1}^N \ell_{ij} e_j,\qquad
    \tilde{e}_i=\sum_{j=1}^N \ell_{ij} Ce_j.
  \]
  We have:
  \[
    \tilde{e}_i = C\tilde{\zeta}_i-\zeta_i = C\xi_i,
  \]
  we choose $d$ such
  that
  \begin{equation}\label{dchoice2}
    0<4d< \delta^2,
  \end{equation}
  and we choose the following adaptive gain:
  \begin{equation} \label{protocol2.2}
    \dot{\rho}_i = \begin{cases}
      \tilde{e}_i\T \tilde{e}_i & \text{ if } \tilde{e}_i\T \tilde{e}_i \geq d, \\ 
      0 & \text{ otherwise}.
    \end{cases}
  \end{equation}
  Let $P_{\alpha_i}$ be defined by:
  \begin{equation}\label{are}
    A\T P_{\alpha_i} + P_{\alpha_i} A - \alpha_i P_{\alpha_i} BB\T
    P_{\alpha_i} +2\eps P_{\alpha_i} +C\T C = 0,
  \end{equation}
  where $\eps>0$ is small enough such that $(A+\eps I,B,C)$ is
  right-invertible and minimum-phase. Note that $P_{\alpha}\rightarrow
  0$ as $\alpha\rightarrow \infty$. Finally, we choose a second adaptive gain:
  \begin{equation} \label{protocol2.3}
    \dot{\alpha}_i = \begin{cases}
      1 & \text{ if } \tilde{\zeta}_i\T C\T C \tilde{\zeta}_i \geq 1, \\
      \tilde{\zeta}_i\T C\T C \tilde{\zeta}_i & \text{ if } 1 >
      \tilde{\zeta}_i\T C\T C  \tilde{\zeta}_i \geq d, \\  
      0 & \text{ otherwise},
    \end{cases}
  \end{equation}
  and the associated feedback:
  \begin{equation} \label{protocol2.4}
    u_i = -\alpha_i B\T P_{\alpha_i} (\hat{x}_i+\tilde{\zeta}_i).
  \end{equation}  
\end{tcolorbox}

We have the following theorem.

\begin{theorem}\label{theorem2}
  Consider a MAS \eqref{agent-g-noise}, satisfying assumption
  \ref{asscol}, with associated network communication
  \eqref{zetanoise} and a given parameter $\delta>0$. If the agent
  model has uniform rank, then the
  \textbf{scale-free {\boldmath $\delta$}-level-coherent output
    synchronization in the presence of bounded external
    disturbances} as stated in Problem \ref{prob_xcol} is
  solvable. In particular, the protocol given by \eqref{protocol2.1},
  \eqref{protocol2.2} and \eqref{protocol2.3} with any $d$ satisfying
  \eqref{dchoice2} solves $\delta$-level-coherent output
  synchronization in the presence of disturbances $w_i$, for an
  arbitrary number of agents $N$ and for any graph
  $\mathscr{G}\in\mathbb{G}^N$.
\end{theorem}

Note that in the above theorem we impose an additional restriction on
the agents that the agent model has uniform rank. It is our conjecture
that the protocol that we presented above still solves the scale-free 
  $\delta$-level-coherent output synchronization in the presence of
bounded external disturbances even if the system is not uniform
rank. However, as of this moment, this is still an open research
question. 

However, we can design a modified protocol that  solves the scale-free 
  $\delta$-level-coherent output synchronization in the presence of
bounded external disturbances by using the protocol above combined
with a precompensator which makes the system uniform rank. We have the following theorem.

\begin{theorem}\label{theorem2a}
  Consider a MAS \eqref{agent-g-noise}, satisfying assumption
  \ref{asscol}, with associated network communication
  \eqref{zetanoise} and a given parameter $\delta>0$. In that case, the
  \textbf{scale-free {\boldmath $\delta$}-level-coherent output
    synchronization in the presence of bounded external
    disturbances} as stated in Problem \ref{prob_xcol} is
  solvable. 
\end{theorem}

\begin{proof}[Proof of Theorem \ref{theorem2a}]
  Using the results from \cite{sannuti-saberi-zhang-auto}, we
  can design a preliminary feedback:
  \[
    \begin{system*}{ccl}
      \dot{x}_{pi} &=& A_px_{pi}+B_pv_i\\
      u_i &=& C_p x_{pi}
    \end{system*}
  \]
  such that the agent with input $v_i$ and output $y_i$ has uniform
  rank and still preserves the other conditions of Assumption
  \ref{asscol}, \emph{i.e.} stabilizability, detectability, minimum-phase and
  right-invertibility. Then, we can use the protocol used in Theorem
  \ref{theorem2} but designed for the new agent model:
  \[
    \tilde{A}=\begin{pmatrix} A & BC_p \\ 0 & A_p \end{pmatrix},\quad
    \tilde{B}=\begin{pmatrix} 0 \\ B_p \end{pmatrix},\quad 
    \tilde{E}=\begin{pmatrix} E \\ 0  \end{pmatrix},\quad 
    \tilde{C}=\begin{pmatrix} C & 0  \end{pmatrix},\quad
  \]
  to achieve that any $d$ satisfying \eqref{dchoice2} solves
  $\delta$-level-coherent output synchronization in the presence of
  disturbances $w_i$, for an arbitrary number of agents $N$ and for
  any graph $\mathscr{G}\in\mathbb{G}^N$. In other words, 
  scale-free $\delta$-level-coherent output
  synchronization in the presence of bounded external disturbances is achieved.
\end{proof}

The proof of the result of Theorem \ref{theorem2} is split in a number
of lemmas to improve the presentation. The first three lemma establish
the behavior of the observer \eqref{protocol2.1} with its adaptive
gain given by \eqref{protocol2.2} after which we have some lemmas to
analyze the behavior of the state feedback \eqref{protocol2.4} and its
associated adaptive gain given by \eqref{protocol2.3}.

From the observer, we obtain:
\begin{equation} \label{syscomp}
  \dot{\xi}_i=A\xi_i-(L_i\rho \otimes QC\T C)\xi -(L_i\otimes E)w
\end{equation}
where $L_i$ is the $i$'th row of $L$ for $i=1,\ldots, N$ and
\[
  \xi=\begin{pmatrix} \xi_1 \\ \xi_2 \\ \vdots \\ \xi_N 
  \end{pmatrix},\qquad
  \rho= \begin{pmatrix}
    \rho_1  & 0          & \cdots & 0      \\
    0         & \rho_2 & \ddots & \vdots \\
    \vdots    & \ddots     & \ddots   & 0 \\
    0         & \cdots     & 0        & \rho_N
  \end{pmatrix}.
\]
together with the adaptive gain \eqref{protocol2.2}. We would like to
stress the crucial feature that this dynamics is completely
independent of the state feedback \eqref{protocol2.4} and its
associated adaptive gain given by \eqref{protocol2.3}. Therefore the
behavior of \eqref{syscomp} can be analyzed as a closed unit
independent of the rest of the dynamics.
  
\begin{lemma}\label{lem1col}
  Consider a number of agents $N$ and a graph
  $\mathscr{G}\in\mathbb{G}^N$. Consider the MAS \eqref{agent-g-noise} with
  associated network communication \eqref{zetanoise} and a given
  parameter $\delta>0$. Assume Assumption \ref{asscol} is satisfied.
  Choose any $d$ satisfying \eqref{dchoice2}. If all $\rho_i$ remain
  bounded then there exists a $T>0$ such that
  \begin{equation}\label{probdef2}
    \| \tilde{e}_i(t) \|  \leq \tfrac{\delta}{2},
  \end{equation}
  for all $t>T$ and for all $i=1,\ldots, N$.
\end{lemma}

\begin{proof}[Proof of Lemma \ref{lem1col}]
  For the system \eqref{syscomp}, we define
  $V_i = \xi_i\T Q^{-1} \xi_i$ and we obtain:
  \[
    \dot{V}_i = \xi_i\T(-\eta Q^{-1}+C\T C)\xi_i - 2\xi_i\T (L_i\rho
    \otimes C\T C)\xi - 2 \xi_i\T (L_i\otimes Q^{-1}E) w 
  \]
  and hence:
  \[
    \dot{V}_i \leq - \tfrac{\eta}{2} V_i + \tilde{e}_i\T \tilde{e}_i - 2\tilde{e}_i\T
    (L_i\rho \otimes I) \tilde{e} + \tfrac{2}{\eta}  w\T (L_i\T L_i\otimes E\T
    Q^{-1} E) w
  \]
  with $\tilde{e}=(I\otimes C)\xi$. Convergence of $\rho_i$ for
  $i=1,\ldots, N$ implies that for any $\eps>0$ there exists $T>T_0$
  such that
  \begin{equation}\label{rhoeps}
    \rho_i(t_2)-\rho_i(t_1) < \eps,
  \end{equation}
  for all $t_2>t_1>T$, and for all $i=1,\ldots, N$. This implies
  \[
    \tilde{e}_i = s_i+v_i
  \]
  with
  \begin{equation}\label{epsM}
    s_i\T(t) s_i(t) \leq d,\qquad \int_T^\infty v_i(\tau)\T
    v_i(\tau)\, \textrm{d}\tau <\eps, 
  \end{equation}
  for $t>T$ where $s_i(t)=\tilde{e}_i(t)$ and $v_i(t)=0$ if
  $\tilde{e}_i\T(t) \tilde{e}_i(t)<d$ and $s_i(t)=0$ and
  $v_i(t)=\tilde{e}_i(t)$ otherwise. The bounds in \eqref{epsM} then
  follow from \eqref{rhoeps} and our adaptation \eqref{protocol2.2}.

  We find:
  \[
    \dot{V}_i \leq - \tfrac{\eta}{2} V_i + m_1 w\T w + \sum_{j=1}^N
    m_2 s_j\T s_j + m_3 v_j\T v_j
  \]
  for suitable constants $m_1,m_2,m_3>0$. This implies that $V_i$ is
  bounded which yields that $\xi_i$ is bounded. This implies that the
  derivative of $\xi_i$ is also bounded which in turn yields that the
  derivative of $\tilde{e}_i$ is bounded. Together with $v_i\in L_2$
  this yields that $v_i(t)\rightarrow 0$ as $t\rightarrow \infty$. We
  obtain \eqref{probdef2} for $t$ large since
  $s_i\T(t) s_i(t)\leq d$.
\end{proof}

\begin{lemma}\label{lem2acol}
  Consider a number of agents $N$ and a graph
  $\mathscr{G}\in\mathbb{G}^N$. Consider the MAS \eqref{agent-g-noise}
  with associated network communication \eqref{zetanoise} and a given
  parameter $\delta>0$. Assume Assumption \ref{asscol} is satisfied.
  Additionally, assume that either the $\rho_i$ associated to agents
  belonging to the basic bi-components are bounded or the graph is
  strongly connected. Choose any $d$ satisfying \eqref{dchoice2}. In
  that case, there exists $\nu,\mu >0$ such that
  \[
    \rho_i (t_2) -\rho_i(t_1) < \nu (t_2-t_1) + \mu
  \]
  for any $t_1,t_2$ with $t_2>t_1$ and for $i=1,\ldots,N$. 
\end{lemma}

\begin{proof}
  Using the notation of Lemma \ref{lem1col} we obtain \eqref{syscomp}
  for $i=1,\ldots,N$ or equivalently:
  \begin{equation}\label{syscomp2}
    \dot{e} =(I\otimes A)e -(\rho L \otimes QC\T C)e -(I\otimes E)w.
  \end{equation}
  where
  \[
    e=\begin{pmatrix} e_1 \\ \vdots \\ e_N \end{pmatrix},\quad
    w=\begin{pmatrix} w_1 \\ \vdots \\ w_N \end{pmatrix}.
  \]
  If all $\rho_i$ are bounded the result of the lemma is
  trivial. Assume that $k\leq N$ of the $\rho_i$ are unbounded. Without loss
  of generality, we assume that the $\rho_i$ are unbounded for $i\leq k$
  while the $\rho_i$ are bounded for $i>k$.

  We first consider the case that $k<N$. We have
  \begin{equation*}
    L = \begin{pmatrix}
      L_{11} & L_{12} \\ L_{21} & L_{22} 
    \end{pmatrix},\quad
    e^k =\begin{pmatrix} e_1 \\ \vdots \\ e_k \end{pmatrix}, \quad
    e^k_c =\begin{pmatrix} e_{k+1} \\ \vdots \\
      e_N \end{pmatrix},\quad
    \tilde{e}^k =\begin{pmatrix} \tilde{e}_{1} \\ \vdots \\
      \tilde{e}_{k} \end{pmatrix},\quad
    \tilde{e}^k_c =\begin{pmatrix} \tilde{e}_{k+1} \\ \vdots \\
      \tilde{e}_{N} \end{pmatrix},
  \end{equation*}
  with $L_{11}\in \R^{k\times k}$. If all the agents associated to
  basic bi-components have a bounded $\rho_i$ this implies that agents
  associated to $i=1,\ldots,k$ are not associated to basic
  bi-components which implies that $L_{11}$ is invertible. On the other
  hand, if the network is strongly connected we always have that $L_{11}$
  is invertible since $k<N$. There exist $s^k_c\in L_\infty$ and
  $v^k_c\in L_2$ with
  \begin{equation}\label{K1K2}
    \| s^k_c \|_\infty < K_1,\quad \| v^k_c \|_2 < K_2,
  \end{equation}
  for suitable chosen $K_1$ and $K_2$ such that
  \[
    \tilde{e}^k_c =
    \begin{pmatrix} v_{k+1} \\ \vdots \\ v_{N} \end{pmatrix} +
    \begin{pmatrix} s_{k+1} \\ \vdots \\ s_{N} \end{pmatrix} 
    = v^k_c+s^k_c.
  \]
  This is easily achieved by setting $v_i(t)=0$
  and $s_i(t)=\tilde{e}_i(t)$ if
  $\tilde{e}_i\T(t) \tilde{e}_i(t) <d$, while for
  $\tilde{e}_i\T(t) \tilde{e}_i(t) \geq d$ we set
  $v_i(t)=\tilde{e}_i(t)$ and $s_i(t)=0$.  It is obvious
  that this construction yields that $s_i\in L_\infty$ while the
  fact that the $\rho_i$ are bounded for $i=k+1,\ldots, N$ implies that
  $v_i\in L_2$ for $i=k+1,\ldots, N$ (note that $\dot{\rho}_i = v_i\T v_i$
  in this construction). We define
  \[
    \hat{e}^k= \begin{pmatrix} \hat{e}_1 \\ \vdots \\
      \hat{e}_k \end{pmatrix} = e^k+(L_{11}^{-1}L_{12}\otimes I)
    e^k_c,\qquad 
    \tilde{e}^k=(L_{11}\otimes C)\hat{e}^k.
  \] 
  Using \eqref{syscomp2}, we then obtain
  \begin{equation}\label{barxt2xxx}
    \dot{\hat{e}}^k = (I\otimes A)\hat{e}^k
    -[\rho^{k} L_{11} \otimes QC\T C] \hat{e}^k
    -\left[ L_{11}^{-1}L_{12}\rho^k_c
      \otimes QC\T \right](s^k_c+v^k_c)-
    [{\setlength{\arraycolsep}{2mm}(\begin{matrix} I 
      & L_{11}^{-1}L_{12} \end{matrix})} \otimes E]w,
  \end{equation}
  where we used that
  \[
    \rho^k=\begin{pmatrix}
      \rho_1 & 0      & \cdots & 0 \\
      0      & \rho_2 & \ddots & \vdots \\
      \vdots & \ddots & \ddots & 0 \\
      0      & \cdots & 0      & \rho_k
    \end{pmatrix},\quad
    \rho^k_c=\begin{pmatrix}
      \rho_{k+1} & 0      & \cdots & 0 \\
      0      & \rho_{k+2} & \ddots & \vdots \\
      \vdots & \ddots & \ddots & 0 \\
      0      & \cdots & 0      & \rho_N
    \end{pmatrix}.
  \]
  Define
  \begin{align*}
    \hat{s}^k &=-(L_{11}^{-1}L_{12}\rho^k_c \otimes QC\T)s^k_c
    -\left[{\setlength{\arraycolsep}{2mm}\begin{pmatrix} I &
        L_{11}^{-1}L_{12} \end{pmatrix}} \otimes E\right]w,\\ 
    \hat{v}^k &=-(L_{11}^{-1}L_{12}\rho^k_c \otimes QC\T)v^k_c , 
  \end{align*}
  then \eqref{K1K2} in combination with the boundedness of
  $\rho^k_c$ implies that there exists $K_3$ and
  $K_4$ such that
  \begin{equation}\label{K3K4}
    \| \hat{s}^k \|_\infty < K_3,\quad \| \hat{v}^k \|_2 < K_4.
  \end{equation}
  We obtain
  \begin{equation}\label{barxt2xxxy}
    \dot{\hat{e}}^k = (I\otimes A)\hat{e}^k
    -[\rho^k L_{11} \otimes QC\T C ]\hat{e}^k+\hat{s}^k+\hat{v}^k,
  \end{equation}
  and we define
  \begin{equation}\label{Vj}
    V_k= (\hat{e}^k)\T (\rho^{-k} H^k \otimes Q^{-1}) \hat{e}^k,
  \end{equation}
  with $\rho^{-k}=(\rho^k)^{-1}$
  while
  \begin{equation}\label{HN}
    H^k=\begin{pmatrix}
      h_1 & 0      & \cdots & 0 \\
      0      & h_2 & \ddots & \vdots \\
      \vdots & \ddots & \ddots & 0 \\
      0      & \cdots & 0      & h_k
    \end{pmatrix}.
  \end{equation}
  Using \cite[Theorem 4.25]{qu-book-2009}, we choose
  $h_1,\ldots,h_k>0$ such that $H^k L_{11} +L_{11}\T H^k > 0$. It is
  easily seen that this implies that there exists a $\gamma$ such that
  \begin{equation}\label{HkL11}
    H^kL_{11}+L_{11}\T H^k > 2\gamma L_{11}\T L_{11}.
  \end{equation}
  We get from \eqref{barxt2xxxy} that
  \[
    \dot{V}_k \leq (\hat{e}^k)\T \left[ \rho^{-k} H^k \otimes (-\eta
      Q^{-1}+C\T C) \right] 
    \hat{e}^k-(\hat{e}^k)\T \left[
      (H^kL_{11}+L_{11}\T H^k)
      \otimes C\T C \right] \hat{e}^k \\
    + 2(\hat{e}^k)\T \left[ \rho^{-k} H^k \otimes Q^{-1}\right] (\hat{s}^k+\hat{v}^k),
  \]
  where we used that $V_k$ is decreasing
  in $\rho_i$ for $i=1,\ldots, k$. The above yields for $t>T$ that
  \begin{equation*}
    \dot{V}_k \leq -\eta V_k 
    - \gamma (\hat{e}^k)\T \left[ L_{11}\T L_{11} \otimes C\T C \right] \hat{e}^k \\
    + 2(\hat{e}^k)\T \left[\rho^{-k} H^k \otimes Q^{-1} \right] (\hat{s}^k+\hat{v}^k),
  \end{equation*}
  provided $T$ is such that
  \[
    \rho^{-k} H^k  < \gamma L_{11}\T L_{11},
  \]
  for $t>T$ which is possible since we have
  $\rho_i\rightarrow \infty$ for $i=1,\ldots, k$.  Note that given \eqref{K3K4}
  there exists some fixed $\beta$ such that
  \begin{equation}\label{lastone4xx}
    \sup_{t\in[T,\infty)} \| \check{s}^k(t) \| \leq \beta,\qquad
    \int_{T}^\infty\, \| \check{v}^k(t) \|^2\, \textrm{d}t   \leq \beta,
  \end{equation}
  with
  \begin{align*}
    \check{s}^k &= \left[ \tfrac{2}{\eta}    
      \rho^{-k}H^k \otimes Q^{-1}\right]^{1/2} \hat{s}^k,\\ 
    \check{v}^k &= \left[ \tfrac{2}{\eta}
       \rho^{-k}H^k \otimes Q^{-1}\right]^{1/2} \hat{v}^k. 
  \end{align*}
  We get
  \begin{equation}\label{lastone2xx}
    \dot{V}_k \leq - \tfrac{\eta}{2} V_k 
    - \gamma (\hat{e}^k)\T \left[ L_{11}\T L_{11} \otimes C\T C \right] \hat{e}^k 
    + (\check{s}^k)\T \check{s}^k + (\check{v}^k)\T\check{v}^k,
  \end{equation}
  for $t>T$. Moreover, 
  \begin{equation}\label{rfgt}
     (\hat{e}^k)\T \left[ L_{11}\T L_{11} \otimes C\T C \right]
     \hat{e}^k \geq  \sum_{i=1}^k \dot{\rho}_i, 
  \end{equation}
  since $(L_{11}\otimes C)\hat{e}^k=\tilde{e}^k$. Hence \eqref{lastone2xx} implies
  \begin{equation}\label{ttg1}
    \dot{V}_k \leq -\tfrac{\eta}{2} V_k -\gamma
    \sum_{i=1}^k \dot{\rho}_{i} + (\check{s}^k)\T \check{s}^k +
    (\check{v}^k)\T\check{v}^k.  
  \end{equation} 
  Note that the bounds in \eqref{lastone4xx} combined with the
  inequality \eqref{ttg1} for $t>T$ implies that there exists some
  $\nu,\mu>0$ such that
  \begin{equation}\label{ttg4}
    \tilde{\rho}_i(t_2)-\tilde{\rho}_i(t_1) <\nu (t_2-t_1) + \mu,
  \end{equation}
  for $i=1,\ldots, k$ and all $t_2,t_1>T$. Clearly, since
  $\rho_{k+1},\ldots,\rho_N$ are all bounded we trivially obtain \eqref{ttg4} for
  $i=k+1,\ldots,N$. 

  If \eqref{ttg4} is satisfied for $t_1,t_2>T$, then it is easily to
  obtain the lemma for $t_1,t_2>0$. After all, for $t<T$ all signals are bounded it
  is clear that for $t<T$ the $\rho_i$ can grow at most linearly and
  hence we can obtain the result for all $t_1,t_2>0$.

  Next, we consider the case that all $\rho_i$ are unbounded. In this
  case, we assumed the graph is strongly connected and hence by Lemma
  \ref{2.8} presented in the appendix there exists
  $\alpha_1,\ldots,\alpha_N>0$ such that 
  \eqref{Hlyap} is satisfied with $H^N$ given by \eqref{HN} for $k=N$.
  We define
  \begin{equation}\label{VN}
    V_N= e\T \left[ Q_{\rho} \otimes Q^{-1} 
    \right] e,
  \end{equation}
  where
  \begin{equation}\label{Qrho}
    Q_{\rho} = \rho^{-1} \left( H^N\rho - \mu_N
      \textbf{h}_N\textbf{h}_N\T \right) \rho^{-1}
  \end{equation}
  while
  \[
    \mu_N=\frac{1}{\sum_{i=1}^N h_i\rho_i^{-1}},\qquad
    \textbf{h}_N = \begin{pmatrix} h_1 \\ \vdots \\ h_N \end{pmatrix}.
  \]
  From Lemma \ref{2.9} in the appendix, we know that $Q_{\rho}$ is
  decreasing in $\rho_i$ for $i=1,\ldots N$.  Note that
  $Q_\rho \rho L=H^N L$.  We get from \eqref{syscomp2} that
  \begin{equation}\label{eqref}
    \dot{V}_N \leq e\T \left[ Q_\rho \otimes (-\eta Q^{-1}+C\T C) \right] e
    -e\T \left[ (H^NL+L\T H^N) \otimes C\T C \right] e 
    - 2e\T \left[ Q_{\rho} \otimes Q^{-1}E\right] w.
  \end{equation}
  It is easily verified that $Q_{\rho}\textbf{1}=0$. Moreover
  $\Ker L = \Span \{ \textbf{1} \}$ since the network is strongly
  connected. Therefore
  \[
    \ker Q_{\rho} \subset \Ker L\T L.
  \]
  Together with the fact that $\rho_j\rightarrow \infty$ for
  $j=1,\ldots, N$ and therefore $Q_{\rho}\rightarrow 0$ this implies
  that there exists $T$ such that
  \begin{equation}\label{Qrhobound}
    Q_{\rho} < \gamma L\T L,
  \end{equation}
  is satisfied for $t>T$.
 
  The above together with \eqref{Hlyap} yields for $t>T$ that
  \begin{equation*}
    \dot{V}_N \leq -\eta V_N
    - \gamma e\T \left[ L\T L \otimes C\T C \right] e \\
    - 2 e\T \left[ Q_{\rho} \otimes Q^{-1}E \right] w.
  \end{equation*}
  We obtain
  \[
    2 e\T \left[ Q_{\rho} \otimes Q^{-1}E\right] w
    \leq \tfrac{\eta}{2} e\T \left[ Q_{\rho} \otimes Q^{-1} \right] e
    + (\check{v}^N)\T \check{v}^N = \tfrac{\eta}{2} V_N
    + (\check{v}^N)\T \check{v}^N
  \]
  with
  \begin{equation}\label{checkv}
    \check{v}^N = \left( \tfrac{2}{\eta} Q_\rho \otimes
      Q^{-1}\right)^{1/2} (I\otimes E)w.
  \end{equation}
  Note that since $w$ is bounded, there exists some fixed $\alpha$ such
  that
  \begin{equation}\label{lastone4xxN}
    \sup_{t\in[T,\infty)} \| \check{v}^N(t) \| \leq \alpha.
  \end{equation}
  We get
  \begin{equation}\label{lastone2xxN}
    \dot{V}_N \leq - \tfrac{\eta}{2} V_N
    - \gamma e\T \left[ L\T L \otimes C\T C \right] e +  (\check{v}^N)\T\check{v}^N,
  \end{equation}
  for $t>T$. Moreover,
  \[
    e\T \left[ L\T L \otimes C\T C \right] e  \geq  \sum_{i=1}^N \dot{\rho}_i.
  \]
  The above implies
  \begin{equation}\label{ttg1N}
    \dot{V}_N \leq - \tfrac{\eta}{2} V_N -\gamma
    \sum_{i=1}^N \dot{\rho}_{i} + 
    (\check{v}^N)\T\check{v}^N. 
  \end{equation} 
  Note that the bound in \eqref{lastone4xxN} combined with the
  inequality \eqref{ttg1N} for $t>T$ implies that there exists some
  $\nu>0$, and $\mu$ such that
  \begin{equation}\label{ttg4N}
    \tilde{\rho}_i(t_2)-\tilde{\rho}_i(t_1) <\nu(t_2-t_1)+\mu,
  \end{equation}
  for $i=1,\ldots, N$ and $t_2,t_1>T$. 

  If \eqref{ttg4N} is satisfied for $t_1,t_2>T$ then the lemma follows
  for $t_1,t_2>0$. After all, for $t<T$ all signals are bounded and
  therefore it is clear that for $t<T$ the $\rho_i$ can grow at most
  linearly and hence we can obtain the result for all $t_1,t_2>0$.
\end{proof}

\begin{lemma}\label{lem2col}
  Consider a number of agents $N$ and a graph
  $\mathscr{G}\in\mathbb{G}^N$. Consider the MAS \eqref{agent-g-noise}
  with associated network communication \eqref{zetanoise} and a given
  parameter $\delta>0$. Assume Assumption \ref{asscol} is satisfied.
  Choose any $d$ satisfying \eqref{dchoice2}.  In that case, all
  $\rho_i$ remain bounded.
\end{lemma}

\begin{proof}[Proof of Lemma \ref{lem2col}]
  We prove this result by contradiction. Let $k$ be such that
  $\rho_i$ is unbounded for $i\leq k$ while $\rho_i$ is bounded
  for $i>k$. Clearly, if all $\rho_i$ are unbounded then we have
  $k=N$.

  For each $\nu\in \N$, we define a time-dependent permutation $p$ of
  $\{ 1,\ldots,N\}$ such that
  \begin{equation}\label{perm}
    \rho_{p_s(1)}(\nu) \geq \rho_{p_s(2)}(\nu) \geq \rho_{p_s(3)}(\nu) \geq
    \cdots \geq \rho_{p_s(k)}(\nu),\quad p_s(i)=i \text{ for } i>k
  \end{equation}
  and we choose $p_t=p_s$ for $t\in[s,s+1)$. Note that Lemma \ref{lem2acol} implies
  there exists $T>0$ such that $\rho_{s,k}(t)\leq 2 \rho_{s,i}(t)$ for
  $i=1,\ldots,k-1$ and $t>T$ (recall that by construction $\rho_{s,i}(t)\rightarrow \infty$
  for $t\rightarrow \infty$). We define
  \[
    \hat{e}_{s,i}(t) = \hat{e}_{p_t(i)}(t),\qquad
    \tilde{e}_{s,i}(t) = \tilde{e}_{p_t(i)}(t),\qquad
    \rho_{s,i}(t) = \rho_{p_t(i)}(t),\qquad
    h_{s,i}(t) = h_{p_t(i)}.
  \]
  For $k<N$, we set
  \begin{equation}\label{tildeVk}
    V_{s,k}= (e_s^k)\T (\rho_s^{-k} H_s^k \otimes
    Q^{-1}) e_s^k, 
  \end{equation}
  while for $k=N$ we have
  \begin{equation}\label{tildeVN}
    V_{s,N}= (e_s^N)\T \left[ Q_{s,\rho} \otimes Q^{-1}
    \right] e_s^N,
  \end{equation}
  where
  \begin{equation}\label{tildeQrho}
    Q_{s,\rho} = \rho_s^{-N} \left(
      H_s^N\rho_s^N - \mu_{N}
      \textbf{h}_{s,N}\textbf{h}_{s,N}\T \right)
    \rho_s^{-N}. 
  \end{equation}  
  Here $H_s^k, \textbf{h}_{s,N}$ and $\rho_s^k$ are
  obtained from $H^k, \textbf{h}_N$ and $\rho^k$ using our
  permutation.
  
  If we assume $k<N$ then using the arguments of Lemma
  \ref{lem2acol}, we obtain \eqref{barxt2xxxy} and the bound
  \eqref{K3K4}. We also obtained in Lemma
  \ref{lem2acol} the bound \eqref{lastone2xx}.
  Using the permutation we
  introduced this immediately yields:
  \begin{equation}\label{lastone2}
    \dot{V}_{s,k} \leq - \tfrac{\eta}{2} V_{s,k} 
    - \gamma (\hat{e}_s^k)\T \left[ L_{s,11}\T L_{s,11}
      \otimes C\T C \right] \hat{e}_s^k    
    + (\check{s}_s^k)\T \check{s}_s^k +
    (\check{v}_s^k)\T\check{v}_s^k, 
  \end{equation}
  where $\check{s}_s^k, \check{v}_s^k$ and $L_{s,11}$ are
  obtained from $\check{s}^k, \check{v}^k$ and $L_{11}$ by applying
  the  permutation introduced above.
    
  Note that there exists some fixed $\tilde{\alpha}$ such that
  \begin{equation}\label{lastone4}
    \| \rho_{s,k}^{1/2}  \check{s}_s^k \|^2_2  \leq
    \tilde{\alpha}^2,\qquad 
    \| \rho_{s,k}^{1/2} \check{v}_s^k \|^2_\infty \leq \tilde{\alpha}^2,
  \end{equation}
  where we exploited that $\rho_{s,k} \leq 2\rho_{s,i}$ for $i=1,\ldots,
  k$ implies that
  \[
    \rho_{s,k} H_s^k \rho_s^{-k}
  \]
  is bounded.

  On the other hand, for $k=N$, then using the arguments of Lemma
  \ref{lem2acol}, we obtain \eqref{lastone2xxN}. Set $\check{s}^N=0$ and
  let $\check{v}^N$ be defined by \eqref{checkv}.  Using the
  permutation we introduced this immediately yields \eqref{lastone2}
  where $\check{s}, \check{v}_s,L_{s,11}$ are
  obtained from $\check{s}, \check{v}$ and $L_{11}=L$ by applying the
  permutation introduced above.
  
  We continue with the case $k=N$. Note that
  $\rho_{s,N} \leq 2\rho_{s,i}$ for $i=1,\ldots, N$ implies that
  \[
    \rho_{s,N} Q_{s,\rho}
  \]
  is bounded which yields that there exists some fixed
  $\tilde{\alpha}$ such that \eqref{lastone4} is satisfied.
  
  We have established \eqref{lastone2} and \eqref{lastone4} for both
  $k<N$ and $k=N$. Equation \eqref{lastone2} implies
  \begin{equation}\label{lastone}
    \dot{V}_{s,k} \leq - \tfrac{\eta}{2} V_{s,k} + (\check{s}_s^k)\T
    \check{s}_s^k + (\check{v}_s^k)\T\check{v}_s^k. 
  \end{equation}
  This clearly yields that $V_{s,k}$ is bounded given our bounds
  \eqref{lastone4}. Using inequality \eqref{ttg1} together with our
  permutation we obtain that, for $t>T_1$, we have
  \begin{align}
    \left[ \rho_{s,k} V_{s,k} \right]' &\leq
    \dot{\rho}_{s,k} V_{s,k} -
    \tfrac{\eta}{2} \rho_{s,k} V_{s,k} -\gamma 
    \rho_{s,k} \dot{\rho}_{s,k} + \rho_{s,k} [(\check{s}_s^k)\T
    \check{s}_s^k + (\check{v}_s^k)\T\check{v}_s^k] \nonumber\\
    &\leq  - \tfrac{\eta}{2} \rho_{s,k} \tilde{V}_k + \rho_{s,k} [(\check{s}_s^k)\T
    \check{s}_s^k + (\check{v}_s^k)\T\check{v}_s^k],  \label{lastone3}
  \end{align}
  where we choose $T_1>T$ such that
  $V_{s,k} \leq \gamma \rho_{s,k}$ for $t>T_1$ which is
  obviously possible since $\rho_{s,k}$ increases to infinity
  while $V_{s,k}$, as argued before, is bounded. Then, using
  \eqref{lastone4} and \eqref{lastone3} we find
  \[
    \rho_{s,k}(\nu+\sigma) V_{s,k}(\nu+\sigma) < e^{-\sigma\eta/2}
    \rho_{s,k}(\nu) V_{s,k}(\nu)  + 2\tilde{\alpha}^2, 
  \]
  for all $\sigma\in (0,1]$ and any $\nu\in \N$ with $\nu>T_1$. This by
  itself does not yield that $\rho_{s,k} V_{s,k}$ is
  bounded because we have potential discontinuities for $\nu\in \N$
  because of the reordering process we introduced. Note that
  $V_{s,k}$ is not affected by the reordering but $\rho_{s,k}$
  can have discontinuities  for $\nu\in \N$. Hence
  \[
    \rho_{s,k}(\nu^+)  \text{ and }
    \rho_{s,k}(\nu^-) 
  \]
  might be different and, strictly speaking, we have obtained
  \begin{equation}\label{yhju}
    \rho_{s,k}(\nu+\sigma) V_{s,k}(\nu+\sigma) < e^{-\sigma\eta/2}
    \rho_{s,k}(\nu^+) V_{s,k}(\nu)  + 2\tilde{\alpha}^2,
  \end{equation}  
  for all $\sigma\in (0,1]$ and any $\nu\in \N$ with $\nu>T_1$. 

  Given the bounds on the growth of $\rho_i$ obtained in Lemma
  \ref{lem2acol}, it is easy to see that there exists $A_0>0$ such that a
  discontinuity can only occur when
  \[
    \rho_{s,k-1}(\nu)-\rho_{s,k}(\nu)<A_0,
  \]
  for $\nu$ sufficiently large. Using this bound, together with the fact
  that $\rho_{s,k}$ converges to infinity, we find that for any
  $\eps$ there exists $T_2>T_1$ such that
  \begin{equation}\label{bound2}
    \rho_{s,k}(\nu^+)  \leq (1+\eps)
    \rho_{s,k}(\nu^-),
  \end{equation}
  for $\nu>T_2$.  Therefore, we find that
  \begin{equation}\label{ttg3}
    \rho_{s,k}(t) V_{s,k}(t) 
  \end{equation}
  is bounded by combining \eqref{bound2} with \eqref{yhju} provided
  that we choose $(1+\eps)e^{-\eta/2} < 1$.
  
  In addition, \eqref{lastone2}, combined with  \eqref{lastone4},
  implies that
  \begin{equation*}
    \gamma \int_{\nu}^{\nu+1} \rho_{s,k} (\hat{e}_s^k)\T \left[
      L_{s,11}\T L_{s,11}\Tcc \otimes C\T C \right] \hat{e}_s^k\,
    \textrm{d}t \leq\\
     2\rho_{s,k}(\nu) V_{s,k}(\nu) + 
    \int_{\nu}^{\nu+1} \rho_{s,k}
    [(\check{s}_s^k)\T\check{s}_s^k+(\check{v}_s^k)\T\check{v}_s^k]
    \textrm{d}t\leq 2\rho_{s,k}(\nu) V_{s,k}(\nu) + 2\alpha_s^2
  \end{equation*}
  for all $\nu\in \N$ since $\tilde{\rho}_k(t)\leq 2\tilde{\rho}_k(\nu)$ for
  $t\in [\nu,\nu+1]$. Boundedness of \eqref{ttg3} then implies that
  \[
    \int_{s}^{s+1} \rho_{s,k} (\tilde{e}_s^k)\T \tilde{e}_s^k\, \textrm{d}t
  \]
  is bounded. 

  We have established that $\tilde{\rho}_k V_{s,k}$ is bounded
  but this does not establish that the $\hat{e}_s^k$ are bounded since
  the $\rho_{s,i}$ for $i=1,\ldots,k-1$ might be much larger than
  $\rho_{s,k}$. We need to do some extra work.
  We have that \eqref{barxt2xxxy} implies
  \begin{equation}\label{barxt2xxxy2}
    \dot{\hat{e}}_s^k = (I\otimes A)\hat{e}_s^k
    -[\rho_s^{k}
    L_{s,11} \otimes QC\T  C]\hat{e}_s^k
    +\hat{s}_s^k+\hat{v}_s^k,
  \end{equation}
  where $\hat{s}_s^k$ and $\hat{v}_s^k$ are
  obtained from $\hat{s}^k$ and $\hat{v}^k$ respectively by
  applying the permutation introduced above. A permutation clearly
  does not affect the bound we obtained in 
  \eqref{K3K4} and we obtain
  \begin{equation}\label{K3K42}
    \| \hat{s}_s^k \|_\infty < K_3,\qquad \| \hat{v}_s^k \|_2 < K_4.
  \end{equation}
  For any $j<k$, we decompose
  \begin{equation}\label{tildeL11}
    \hat{e}^j_{s,I} = \begin{pmatrix}
      \hat{e}_{s,1} \\ \vdots \\ \hat{e}_{s,j} \end{pmatrix},\quad
    \hat{e}^j_{s,II} = \begin{pmatrix}
      \hat{e}_{s,j+1} \\ \vdots \\ \hat{e}_{s,k} \end{pmatrix},\quad
    L_{s,11} =\begin{pmatrix} L^j_{s,11} & L^j_{s,12} \\
      L^j_{s,21} & L^j_{s,22}
    \end{pmatrix},
  \end{equation}
  with $\tilde{L}^j_{11}\in \R^{j\times j}$ while
  \begin{equation*}
    \hat{s}_s^k = \begin{pmatrix} \hat{s}_s^j \\
      \hat{s}^j_{s,c} \end{pmatrix},\quad
    \hat{v}_s^k = \begin{pmatrix} \hat{v}_s^j \\
      \hat{v}^j_{s,c} \end{pmatrix},
  \end{equation*}
  with 
  $\hat{s}_s^j\in \R^{nj}$, $\hat{v}_s^j\in \R^{nj}$
  and
  \begin{equation}\label{checkx}
    \check{e}_s^j = \hat{e}^j_{s,I}  + (L_{s,11}^j)^{-1}
    L^j_{s,12} \hat{e}^j_{s,II},
  \end{equation}
  for $j<k$ while $\check{e}_s^k=\hat{e}_s^k$.  We will show that 
  \begin{equation}\label{upbound}
    \rho_{s,j} V_{s,j}\qquad\text{ and }\qquad  \int_{\nu}^{\nu+1}
    \rho_{s,j} (\tilde{e}_s^j)\T \tilde{e}_s^j\, \textrm{d}t 
  \end{equation}
  are bounded for $j=1,\ldots, k$  where
  \begin{equation}\label{tildeVj}
    V_{s,j}=(\check{e}_s^j)\T \left[ H_s^j\rho_s^{-j} \otimes Q^{-1}
  \right] \check{e}_s^j, 
  \end{equation}
  while
  \[
    \rho_s^j=\begin{pmatrix}
      \rho_{s,1} & 0      & \cdots & 0 \\
      0      & \rho_{s,2} & \ddots & \vdots \\
      \vdots & \ddots & \ddots & 0 \\
      0      & \cdots & 0      & \rho_{s,j}
    \end{pmatrix},\quad
    \rho_{s,c}^j=\begin{pmatrix}
      \rho_{s,j+1} & 0      & \cdots & 0 \\
      0      & \rho_{s,j+2} & \ddots & \vdots \\
      \vdots & \ddots & \ddots & 0 \\
      0      & \cdots & 0      & \rho_{s,k}
    \end{pmatrix},\quad
    H_s^j=\begin{pmatrix}
      h_{s,1} & 0      & \cdots & 0 \\
      0      & h_{s,2} & \ddots & \vdots \\
      \vdots & \ddots & \ddots & 0 \\
      0      & \cdots & 0      & h_{s,j}
    \end{pmatrix}.
  \]
  Note that if $k=N$ then we have \eqref{tildeVN} instead of \eqref{tildeVj}.
  We also define
  \[
    \tilde{e}_s^j=\begin{pmatrix} \tilde{e}_{s,1} \\ \vdots \\
      \tilde{e}_{s,j} \end{pmatrix},\qquad
    \tilde{e}^j_{s,c}=\begin{pmatrix} \hat{e}_{s,j+1} \\ \vdots \\
      \tilde{e}_{s,k} \end{pmatrix}.
  \]
  It is not hard to verify that if $\rho_{s,j}V_{s,j}$ is bounded for
  $j=1,\ldots, k$ then we have that  $\hat{e}_j$ is bounded for
  $j=1,\ldots, k$.
  
  We will establish boundedness of $\rho_{s,j}V_{s,j}$ via
  recursion. Assume for $i=j$ we have either
  \begin{equation}\label{ttg3b}
    \rho_{s,i} V_{s,i}
  \end{equation}
  is unbounded or
  \begin{equation}\label{ttg3c}
    \int_{\nu}^{\nu+1} \rho_{s,i} (\tilde{e}_s^i)\T \tilde{e}_s^i\, \textrm{d}t
  \end{equation}
  is unbounded while, for $j<i\leq k$, both \eqref{ttg3b} and
  \eqref{ttg3c} are bounded. We will show this yields a
  contradiction. Note that in the above we already established
  \eqref{ttg3b} and \eqref{ttg3c} are bounded for $i=k$.
  
  Using \eqref{barxt2xxxy2} and \eqref{checkx}, we obtain
  \begin{equation}\label{checkxj}
    \dot{\check{e}}_s^j = (I\otimes A) \check{e}_s^j -
    \left[\rho_s^jL_{s,11}^j\otimes
      QC\T C  \right] \check{e}_s^j+ \check{s}^j_s+\check{v}^j_s 
    + \rho_{s,j}^{-1/2} \left[ \rho_s^j(H_s^j)^{-1} (L_{s,11}^j)\T
     \otimes QC\T) \right]\check{w}_s^j 
  \end{equation}
  where
  \begin{align*}
    \check{s}_s^j & =\hat{s}_s^j+(L_{s,11}^j)^{-1}L_{s,12}^j
    \hat{s}_{s,c}^j,\\
    \check{v}_s^j & = \hat{v}_s^j+(L_{s,11}^j)^{-1}\tilde{L}_{s,12}^j
    \hat{v}_{s,c}^j  \\
    \check{w}_s^j &= \rho_{s,j}^{1/2} \left[ 
      (L_{s,11}^j)^{-T}  H_s^j (\rho_s^j)^{-1}(L^j_{s,11})^{-1} L^j_{s,12}
      \rho^j_{s,c} \otimes I  \right] \tilde{e}_{s,c}^j. 
  \end{align*}
  Using \eqref{K3K42}, we find that there exist constant $\tilde{K}_3$ such that
  \begin{equation}\label{tildeK3K4}
    \| \check{s}_s^j \|_\infty < \tilde{K}_3,\quad
    \int_\nu^{\nu+1} (\check{v}_s^j)\T\check{v}_s^j \textrm{d}t < \tilde{K}_3.
  \end{equation}
  Finally,
  \eqref{ttg3c} bounded  for $i=j+1,\ldots ,k$ implies there exists a
  constant $\tilde{K}_4$ such that
  \begin{equation}\label{tildeK3K4b2}
    \int_{\nu}^{\nu+1} (\check{w}_s^j)\T \check{w}_s^j\,
    \textrm{d}t < \tilde{K}_4
  \end{equation}
  is bounded.   We obtain
  \begin{multline*}
    \dot{V}_{s,j} =(\check{e}_s^j)\T \left[
      H_s^j\rho_s^{-j} \otimes (-\eta Q^{-1} + C\T C
    \right] \check{e}_s^j  
    -(\check{e}_s^j)\T \left[ (H_s^jL_{s,11}^j+(L_{s,11}^j)\T
      H_s^j)\otimes C\T C \right]\check{e}_s^j \\
    + 2(\check{e}_s^j)\T \left[ H_s^j\rho_s^{-j} \otimes Q^{-1} \right]
    (\check{s}_s^j+\check{v}_s^j) + 2\rho_{s,j}^{-1/2}(\check{e}_s^j)\T \left[
      (L_{s,11}^j)\T \otimes C\T \right]
    \check{w}_s^j
  \end{multline*}
  using \eqref{tildeVj} and \eqref{checkxj}. Note that \eqref{HkL11} implies
  \[
    H_s^jL_{s,11}^j+(L_{s,11}^j)\T
    H_s^j > 2\gamma (L_{s,11}^j)\T L_{s,11}^j.
  \]
  Moreover, there exists $T>0$ such that
  \[
    H_s^j \rho_s^{-j}  \leq \gamma,
  \]
  for $t>T$ since $\rho_s^j \rightarrow \infty$. This yields that 
   \begin{equation}\label{gettingthere}
    \dot{V}_{s,j} \leq - \tfrac{\eta}{2} V_{s,j} -\tfrac{\gamma}{2}
    (\check{e}^j)\T \left[ (L_{s,11}^j)\T L_{s,11}^j \otimes C\T C
    \right]\check{e}^j+ \tfrac{2}{\eta} (\check{s}^j)\T \left[
      H_s^j\rho_s^{-j} \otimes Q^{-1} \right] \check{s}^j+
    \tfrac{2}{\eta} (\check{v}^j)\T \left[ H_s^j\rho_s^{-j} \otimes
      Q^{-1} \right] \check{v}^j + \rho_{s,j}^{-1} \tfrac{2}{\gamma}
    (\check{w}_s^j)\T \check{w}_s^j.
  \end{equation}
  We find $V_{s,j}$ is bounded and hence we can choose $T_1$ such that
  $V_{s,j} \leq \gamma \rho_{s,j}$ for $t>T_1$. Similar as before,
  \eqref{tildeK3K4}, \eqref{tildeK3K4b2} and \eqref{gettingthere}
  imply there exists some $\tilde{M}>0$ such that
  \begin{equation}\label{llast}
    \rho_{s,j}(\nu+\sigma) V_{s,j}(\nu+\sigma) < e^{-\sigma\eta/2}
    \rho_{s,j}(\nu^+) V_{s,j}(\nu^{+})  + \tilde{M},
  \end{equation}
  for all $\sigma\in (0,1]$ and any $\nu\in \N$ with $\nu>T_1$. Again,
  by itself, this does not imply that $\rho_{s,j}(\nu)V_{s,j}(\nu)$ is
  bounded because at time $\nu\in \N$ there might be a discontinuity
  due to the reordering we performed.

  If a new permutation has the same $j$ agents with the largest
  $\rho_i$ then $V_{s,j}(\nu^+) = V_{s,j}(\nu^-)$ and we obtain,
  as before, that there exists some $\eps$ satisfying
  $(1+\eps)e^{-\eta/2}<1$ and a $T_2>T_1$ such that
  \begin{equation}\label{llast23}
    \rho_{s,j}(\nu^+) V_{s,j}(\nu^+) \leq (1+\eps)
    \rho_{s,j}(\nu^-) V_{s,j}(\nu^-),
  \end{equation}
  for $\nu>T_2$ similarly as we did in the derivation of equation \eqref{bound2}.

  If a new permutation changes the set of $j$
  agents with the largest $\rho_i$ then it is easy to see that
  there exists $A_0>0$ such that a 
  discontinuity can only occur when
  \[
    \rho_{s,j}(\nu)-\rho_{s,j+1}(\nu)<A_0
  \]
  for $\nu$ sufficiently large using Lemma \ref{lem2acol}. This implies
  that there exists some $A>0$ such that
  \begin{equation*}
    \rho_{s,j}(\nu)V_{s,j}(\nu) <
    (\rho_{s,j+1}(\nu)+A_0) V_{s,j}(\nu)\\ 
    < (\rho_{s,j+1}(\nu)+A_0)  M V_{s,j+1}(\nu)\\
    < \frac{\rho_{s,j+1}(\nu)+A_0}{\rho_{s,j+1}(\nu)}
    M \rho_{s,j+1}(\nu) V_{s,j+1}(\nu) < A,
  \end{equation*}
  for large $\nu$ using Lemma \ref{2.10}. The last inequality follows
  from the fact that we already established that
  $\rho_{s,j+1} V_{s,j+1}(\nu)$ is bounded while
  $\rho_{s,j+1}$ is increasing to infinity. Combined with
  \eqref{llast} this shows
  \[
    \rho_{s,j}(\nu+\sigma)V_{s,j}(\nu+\sigma)
  \]
  is bounded as well. If $\rho_{s,j}(\nu) V_{s,j}(\nu)$ is
  larger than $A$ then we know discontinuities of $V_{s,j}(\nu)$ do
  not arise and hence \eqref{llast} and \eqref{llast23} show that
  $\rho_{s,j} V_{s,j}$ remains bounded.

  It remains to show that \eqref{ttg3c} is bounded. This follows
  immediately from \eqref{gettingthere} in combination with
  \eqref{tildeK3K4} and the boundedness of $\rho_{s,j} V_{s,j}$.

  In this way, we recursively established that \eqref{ttg3b} is
  bounded for $i=1,\ldots,k$. It is not hard to show that this implies
  that $\hat{\xi}^k=(L_{11}\otimes I)\hat{e}^k$ is bounded for $k<N$
  while we obtain that $\xi=(L\otimes I)e$ is bounded for $k=N$.
  
  Using \eqref{barxt2xxx} we obtain:
  \begin{equation}\label{barxt2xxx2}
    \dot{\hat{\xi}}^k = (I\otimes A)\hat{\xi}^k
    -[L_{11} \rho^{k} \otimes QC\T C] \hat{\xi}^k
    + \bar{v}^k + \bar{s}^k
  \end{equation}
  where
  \begin{align*}
    \bar{v}^k = -\left[ L_{12}\rho^k_c
      \otimes QC\T \right]v^k_c
    \bar{s}^k = -\left[ L_{12}\rho^k_c
      \otimes QC\T \right]s^k_c-
    [{\setlength{\arraycolsep}{2mm}(\begin{matrix} L_{11}
        & L_{12} \end{matrix})} \otimes E]w 
  \end{align*}
  for $k<N$. Note that $\hat{\xi}^k$ and $\bar{s}^k$ are bounded while $v_c^k$ has
  bounded energy (recall that $\rho_c^k$, by construction, is bounded).

  For $k=N$ we obtain:
  \begin{equation}\label{barxt2xxx2b}
    \dot{\xi} = (I\otimes A)\hat{\xi}
    -[L \rho \otimes QC\T C] \xi
    - (L \otimes E) w 
  \end{equation}
  with $\xi$ and $w$ bounded. We choose:
  \[
    \tilde{V}_k = \tfrac{1}{\rho_{s,k}} (\hat{\xi}^k)\T (
    \rho^k H^k \otimes C\T C) \hat{\xi}^k = \tfrac{1}{\rho_{s,k}} (\tilde{e}^k)\T (
    \rho^k H^k \otimes I) \tilde{e}^k
  \]
  For $k<N$ we get
  \[
    \dot{\tilde{V}}_k
    \leq M  - \tfrac{2\gamma }{\rho_{s,k}} (\bar{\xi}^k)\T (L_{11}\T
    L_{11} \otimes Q )\bar{\xi}^k
    + \tfrac{2}{\rho_{s,k}}  (\bar{\xi}^k)\T (H^k\otimes C\T C) \left[
      (I\otimes A)
      \hat{\xi}^k + \bar{s}^k+\bar{v}^k \right],  
  \]
  where we used \eqref{HkL11} and
  \[
    \bar{\xi}^k=( \rho^k\otimes C\T C ) \hat{\xi}^k,
  \]
  while $M$ is such that:
  \begin{equation} \label{mmmm}
    \tfrac{1}{\rho_{s,k}} (\hat{\xi}^k)\T ( \dot{\rho}^k H^k \otimes C\T C)
    \hat{\xi}^k < M,
  \end{equation}
  which is possible since $\hat{\xi}^k$ bounded guarantees that
  $\dot{\rho}^k$ is bounded as well. Since $Q>0$ and
  $L_{11}\T L_{11}>0$, we find there exists $m_1,m_2$ and $m_3$ such
  that:
  \begin{equation}\label{boundidk}
    \dot{\tilde{V}}_k
    \leq M  + m_1 (\hat{\xi}^k)\T \hat{\xi}^k + m_2 (\bar{s}^k)\T
    \bar{s}^k  + m_3 (\bar{v}^k)\T \bar{v}^k.
  \end{equation}
  For $k=N$ we get:
  \[
    \dot{\tilde{V}}_N
    \leq M  - \tfrac{2\gamma }{\rho_{s,k}} (\bar{\xi}^N)\T (I  \otimes Q )\bar{\xi}^N
    + \tfrac{2}{\rho_{s,N}}  (\bar{\xi}^N)\T (H^N\otimes C\T C) \left[
      (L^{\dagger} \otimes A)
      \xi + w\right],
  \]
  where we used \eqref{Hlyap} and
  \[
    \bar{\xi}^k=( \rho^kL \otimes C\T C ) \xi,
  \]
  while $M$ is given by \eqref{mmmm}. Moreover $L^{\dagger}$ is a
  generalized inverse of $L$ such
  that $LL^{\dagger}L = L$ and hence $\xi=(L\otimes
  I)e=(LL^{\dagger}\otimes I) \xi$. Since $Q>0$ and
  we find there exists $m_1$ and $m_2$ such
  that:
  \begin{equation}\label{boundidN}
    \dot{\tilde{V}}_N
    \leq M  + m_1 \xi\T \xi + m_2 w\T w.
  \end{equation}
  Note that if $\xi_{s,k}\T \xi_{s,k}\Tcc(t)\geq d$ for some $t_1>0$ then
  $V_k(t_1)\geq d$. Given that the growth of $V_k$ is limited by either
  \eqref{boundidk} or \eqref{boundidN} there exists an $\eps$
  independent of $t_1$ such that
  \[
    V_{k}(t)>\tfrac{d}{2}
  \]
  for all $t\in [t_1-\eps,t_1]$. But this implies:
  \[
    \int_{t_1-\eps}^{t_1} V_k(t) \,\textrm{d} t \geq \tfrac{1}{2} d\eps.
  \]
  However, for large $t_1$, this contradicts with the fact that
  \eqref{ttg3c} is bounded while $\rho_{s,k}\rightarrow \infty$. On
  the other hand, the fact that $\xi_{s,k}\T \xi_{s,k}\Tcc(t)\geq d$
  cannot happen for large $t$ contradicts with the underlying
  construction where $\rho_{s,k}\rightarrow \infty$. Hence, we obtain
  a contradiction with our underlying premise that some of the
  $\rho_i$ are unbounded. Hence all $\rho_i$ are bounded which the
  proof of the lemma.
\end{proof}

\begin{proof}[Proof of Theorem \ref{theorem2}]
  We first look at the observer dynamics given by
  \eqref{syscomp}. Note that the dynamics are completely independent
  of the adaptive parameters $\alpha_i$.

  The Laplacian matrix of the system in general has the form
  \eqref{Lstruc}. We note that if we look at the dynamics of the
  agents belonging to one of the basic bi-components, then these
  dynamics are not influenced by the other agents and hence can be
  analyzed independent of the rest of network. The network within one
  of the basic bi-components is strongly connected and we can apply
  Lemmas \ref{lem2acol} and \ref{lem2col} to guarantee that the $\rho_i$
  associated to a basic bi-components are all bounded.

  Next, we look at the full network again. We have already established
  that the $\rho_i$ associated to all basic bi-components are all
  bounded. Then we can again apply Lemmas \ref{lem2acol} and \ref{lem2col}
  to conclude that the other $\rho_i$ not associated to basic
  bi-components are also bounded.

  After having established that all the $\rho_i$ are all bounded, we
  can then apply Lemma \ref{lem1col} to conclude that \eqref{probdef2}
  is satisfied. Next, we consider the state feedback and note that
  \eqref{protocol2.1} yields:
  \begin{equation} \label{protocol2.1rev}
    \dot{\hat{x}}_i=A\hat{x}_i-\alpha_i BB\T P_{\alpha_i}
    \left[\hat{x}_i +\sum_{j=1}^N
    \ell_{ij}\hat{x}_j \right] + \tilde{E}\tilde{w}_i 
  \end{equation}
  together with
  \begin{equation} \label{protocol2.3bc}
    \dot{\alpha}_i = \begin{cases}
      1 & \text{ if } \tilde{\zeta}_i\T \tilde{\zeta}_i \geq 1 \\
      \tilde{\zeta}_i\T \tilde{\zeta}_i & \text{ if } 1 >
      \tilde{\zeta}_i\T \tilde{\zeta}_i \geq d \\  
      0 & \text{ otherwise}
    \end{cases}
  \end{equation}
  where $\tilde{E}=-QC\T$,
  \begin{align*}
    \tilde{\zeta}_i = \sum_{j=1}^N \ell_{ij} \hat{x}_j,\qquad\text{
      and }\qquad
    \tilde{w} = \rho_i \tilde{e}_i.
  \end{align*}
  Since we already know that the $\rho_i$ remain bounded while
  $\tilde{e}_i$ satisfies \eqref{probdef2}, we know that $\tilde{w}_i$
  is bounded.

  We obtain:
  \begin{equation}\label{fgvvb}
    \dot{\tilde{\zeta}} = (I\otimes A)\tilde{\zeta}-(\tilde{L}\alpha\otimes
    BB\T)P_{\alpha} \tilde{\zeta}+(L\otimes \tilde{E})\tilde{w},
  \end{equation}
  where $\tilde{L}=I+L$ and
  \[
    \alpha= \begin{pmatrix}
      \alpha_1  & 0          & \cdots & 0      \\
      0         & \alpha_2 & \ddots & \vdots \\
      \vdots    & \ddots     & \ddots   & 0 \\
      0         & \cdots     & 0        & \alpha_N
    \end{pmatrix},\qquad
    P_{\alpha}= \begin{pmatrix}
      P_{\alpha_1}  & 0          & \cdots & 0      \\
      0         & P_{\alpha_2} & \ddots & \vdots \\
      \vdots    & \ddots     & \ddots   & 0 \\
      0         & \cdots     & 0        & P_{\alpha_N}
    \end{pmatrix},\qquad
    \tilde{\zeta}=\begin{pmatrix} \tilde{\zeta}_1 \\ \vdots \\
      \tilde{\zeta}_N \end{pmatrix},\qquad
    \tilde{w}=\begin{pmatrix} \tilde{w}_1 \\ \vdots \\
      \tilde{w}_N \end{pmatrix}.
  \]
  Note that \eqref{are} implies that $P_{\alpha_i}\rightarrow 0$ as
  $\alpha_i\rightarrow \infty$.  Based on Lemma \ref{2.8} it is easy
  to see that we can find
  \[
    \tilde{H}^N = 
    \begin{pmatrix}
      \tilde{h}_1 & 0      & \cdots & 0 \\
      0      & \tilde{h}_2 & \ddots & \vdots \\
      \vdots & \ddots & \ddots & 0 \\
      0      & \cdots & 0      & \tilde{h}_N
    \end{pmatrix},
  \]
  such that
  \[
    \tilde{H}^N\tilde{L}+\tilde{L}\T \tilde{H}^N \geq 2\tilde{H}^N.
  \]
  We define:
  \[
    V= \tilde{\zeta}\T (\tilde{H}^N\alpha
    \otimes I) P_{\alpha} \tilde{\zeta},
  \]
  and we obtain:
  \begin{multline*}
    \dot{V} \leq
    \tilde{\zeta}\T P_{\alpha} (\dot{\alpha}\tilde{H}^N \otimes I)
    \tilde{\zeta}
    +\tilde{\zeta}\T \left[ P_{\alpha}(\tilde{H}^N\alpha \otimes
      A)+(\tilde{H}^N\alpha \otimes A\T)P_{\alpha} \right]\tilde{\zeta}
    -\tilde{\zeta}\T P_{\alpha} \left[ \alpha(\tilde{H}^N
      \tilde{L}+\tilde{L}\T\tilde{H}^N)\alpha\otimes
      BB\T\right] P_{\alpha} \tilde{\zeta} \\
    +2\tilde{\zeta}\T P_{\alpha} (\alpha\tilde{H}^N L \otimes \tilde{E}) w
  \end{multline*}
  where we used that $P_{\alpha}$ is decreasing in $\alpha$. Next, we
  use \eqref{are}  to get:
  \begin{multline} \label{bnound}
    \dot{V} \leq 
    \tilde{\zeta}\T P_{\alpha} (\dot{\alpha} \tilde{H}^N \otimes I)
    \tilde{\zeta}
    +\tilde{\zeta}\T P_{\alpha} \left[ \alpha^2\tilde{H}^N \otimes BB\T
    \right]P_{\alpha} \tilde{\zeta}  
    -\tilde{\zeta}\T  \left[ \alpha \tilde{H}^N \otimes C\T C
    \right] \tilde{\zeta} \\
    -2\eps \tilde{\zeta}\T P_{\alpha} \left[  \alpha\tilde{H}^N \otimes I
    \right]\tilde{\zeta}
    -2\tilde{\zeta}\T P_{\alpha} \left[ \alpha^2\tilde{H}^N \otimes BB\T
    \right]P_{\alpha} \tilde{\zeta}
    +2\tilde{\zeta}\T P_{\alpha} (\alpha\tilde{H}^N L \otimes \tilde{E}) w.
  \end{multline}
  Define
  \[
    V_i= \alpha_i \tilde{h}_i (\tilde{\zeta}_i)\T P_{\alpha_i} \tilde{\zeta}_i.
  \]
  For ease of presentation we sort the agents such that for
  $i=1,\ldots, \rho$ we have $\alpha_i=0$ for all $t$ while for
  $i=\rho+1,\ldots, \nu$ we have $\alpha_i$ bounded and converges to a
  value unequal to zero as $t\rightarrow \infty$. Finally, for
  $i=\mu+1,\ldots, N$ we have $\alpha_i\rightarrow \infty$ as
  $t\rightarrow \infty$. Let $t_0$ be such that for $t>t_0$ we have
  that $\alpha_i(t)\neq 0$ for $i>\rho$ and $\eps\alpha_i(t)>2$ for $i>\nu$.
  In that case,
  \begin{equation}\label{tty3}
    \tilde{\zeta}\T P_{\alpha} (\dot{\alpha}\tilde{H}^N \otimes I)
    \tilde{\zeta} - \eps \tilde{\zeta}\T P_{\alpha} \left[  \alpha\tilde{H}^N \otimes I
    \right]\tilde{\zeta}
    \leq \sum_{i=\rho+1}^{\nu} \tfrac{\dot{\alpha}_i}{\alpha_i} V_i 
    \leq \dot{\beta} V,
  \end{equation}
  where
  \[
    \beta(t) =\sum_{i=\rho+1}^{\nu} \ln \alpha_i(t)-\ln \alpha_i(t_0).
  \]
  Note that there exists $M_\beta > m_\beta >0$ such that $m<\beta(t)<M$ for $t\geq
  t_0$.
  
  Note that:
  \begin{equation}\label{tty4}
    -q\T \left[ P_{\alpha} (\alpha\tilde{H}^N \otimes I) \right] q + 2 q\T
    \left[ P_{\alpha} (\alpha\tilde{H}^N \otimes I) \right] r \leq r\T
    \left[ P_{\alpha} (\alpha\tilde{H}^N \otimes I) \right] r
  \end{equation}
  for
  \begin{align*}
    q &=\sqrt{\eps} \tilde{\zeta},
    \\
    r &=\tfrac{1}{\sqrt{\eps}} \left[ L\otimes \tilde{E} \right] \tilde{w}.
  \end{align*}
  Using \eqref{tty3} and \eqref{tty4} in \eqref{bnound}, we get:
  \begin{equation}\label{REGGRE}
    \dot{V} \leq \dot{\beta} V - \tfrac{1}{2} \eps V
    -\tilde{\zeta}\T  \left[ \alpha \tilde{H}^N \otimes C\T C
    \right] \tilde{\zeta} 
    -2\tilde{\zeta}\T P_{\alpha} \left[ \alpha^2\tilde{H}^N \otimes BB\T
    \right]P_{\alpha} \tilde{\zeta}
    +\tfrac{1}{\eps} \tilde{w}\T (L\otimes \tilde{E})\T \left[ P_{\alpha}
      (\alpha\tilde{H}^N \otimes I) \right](L\otimes \tilde{E}) \tilde{w},
  \end{equation}
  which implies:
  \begin{align}
    \dot{V} & \leq \dot{\beta} V - \tfrac{1}{2}  \eps V 
    +\tfrac{1}{\eps}  \tilde{w}\T (L\otimes \tilde{E})\T \left[ \alpha
      P_{\alpha} (\tilde{H}^N \otimes I) \right](L\otimes \tilde{E})
    \tilde{w} \nonumber \\
    & \leq \leq \dot{\beta} V - \tfrac{1}{2}  \eps V 
    +\tfrac{1}{\eps}  \tilde{w}\T (L\otimes \tilde{E})\T \left[  (\tilde{H}^N \otimes \alpha_{\max}
      P_{\alpha_{\max}}) \right](L\otimes \tilde{E})
    \tilde{w} \label{dotVbo},
  \end{align}
  where we define
  \[
    \alpha_{\max} = \max_{i} \alpha_i,
  \]
  and used that $\alpha P_{\alpha}$ is increasing in $\alpha$. Assume
  at least one of the $\alpha_i$ is unbounded. In that case
  $\alpha_{\max}\rightarrow\infty$. In that case there exists
  $\bar{w}$ such that
  \[
    \frac{1}{\eps}  \tilde{w}\T (L\otimes \tilde{E})\T \left[  (\tilde{H}^N \otimes \alpha_{\max}
      P_{\alpha_{\max}}) \right](L\otimes \tilde{E})
    \tilde{w} =\alpha_{\max}  \beta_m^2 M_w,
  \]
  for some $M_w>0$ since Lemma
  \ref{asympric} implies that $\| P_{\alpha_{\max}} \| <M \beta_m^2$
  for some $M>0$ with $\beta_i$ as defined in \eqref{betam} for
  $\alpha=\alpha_{\max}$.  We obtain:
  \[
    \dot{V} \leq
    \dot{\beta} V - \tfrac{1}{2}  \eps
    V  + \alpha_{\max}  \beta_m^2 \bar{w}.
  \]
  We find
  \[
    \tfrac{\textrm{d}}{\textrm{d}t} \bar{V} \leq
    -\tfrac{1}{2} \eps \bar{V} +  e^{M_\beta} M_w,
  \]
  for
  \[
    \bar{V}=\tfrac{e^{\beta}}{\alpha_{\max}\beta_m^2} V.
  \]
  This clearly implies that $\bar{V}$ is bounded and hence 
  \[
    \frac{V}{\alpha_{\max}\beta_m^2}
  \]
  is bounded. Consider the dynamics for $\zeta_i$ which can be written
  as:
  \[
    \begin{system*}{ccl}
      \dot{\zeta}_i &=& A\zeta_i + B \check{u}_i + \check{w}_i\\
      \check{y}_i &=& C\zeta_i
    \end{system*}
  \]
  for appropriately chosen $\check{u}_i$, $\check{w}_i$ and
  $\check{y}_i$ with $\check{w}_i$ bounded.
    Using \eqref{bnound}, we find that for an agent for
  which $\alpha_i=\alpha_{\max}$ we have for any $\delta$ that:
  \[
    \int_{t_0}^{t_0+1}  \alpha^{-1} \check{u}_i\T \check{u}_i\, \textrm{d}t
    <\delta\qquad \text{and}\qquad
    \int_{t_0}^{t_0+1}  \check{y}_i\T \check{y}_i \,\textrm{d}t
    <\delta
  \]
  provided $t_0$ is large enough. Using the structure obtained in the
  appendix, in particular Lemma \ref{asympric} and \ref{tty5}, we obtain that for
  any $\eps>0$ we have that
  \[
    \int_{t_0}^{t_0+1} \frac{V}{\alpha_{\max}\beta_m^2} \textrm{d}t <
    \eps
  \]
  provided $t_0$ is large enough. However, because the derivative of 
  \begin{equation}\label{V/ab}
    \frac{V}{\alpha_{\max}\beta_m^2}
  \end{equation}
  is upper bounded this implies that \eqref{V/ab} converges to zero as
  $t\rightarrow \infty$. But this implies that $C\zeta_i$ converges to
  zero which implies that $\dot{\alpha}_{\max}=0$ for $t$ large. This
  contradicts our assumption that $\alpha_{\max}$ converges to infinity.

  On the other hand if all $\alpha_i$ are bounded then we find that
  \[
    \tilde{\zeta}_i(t)\T C\T C \tilde{\zeta}_i(t) = p_i(t)+q_i(t)
  \]
  with $p_i(t)\leq d$ while $q_i(t)\in L_2$. Since $\tilde{\zeta}_i$ is
  bounded we find that $\dot{q}_i(t)$ is bounded which yields together
  with $q_i(t)\in L_2$ that $q_i(t)\rightarrow 0$ as $t\rightarrow
  \infty$. But this implies that
  \[
    \| C \tilde{\zeta}_i \| \leq \tfrac{\delta}{2}
  \]
  given \eqref{dchoice2}. Together with \eqref{probdef2}, obtained
  before, this completes the proof.
\end{proof}

\section{\textbf{Numerical examples}}

In this section, we will show the feasibility of the proposed
noncollaborative and collaborative controllers, respectively. We study
the effectiveness of our proposed protocols as they are applied to
systems with different sizes, different communication graphs,
different noise
patterns, and different $\delta$ values.

In all examples of the paper, the weight of edges of the communication
graphs is considered to be equal~$1$.

\subsection{Noncollaborative Protocols}\label{xmpl1}
We consider agent models
\begin{equation}\label{agent-g-noise-Example}
  \begin{system*}{cl}
    \dot{x}_i&=\begin{pmatrix}0&1&1&0\\-1&0&1&0\\0&0&0&1\\0&0&0&-2\end{pmatrix}x_i+\begin{pmatrix}
      0\\1\\0\\1
    \end{pmatrix}u_i+\begin{pmatrix}
      0\\1\\0\\1
    \end{pmatrix}w_i,\\
    y_i&=\begin{pmatrix}1&0&0&0\\0&1&0&0\end{pmatrix}x_i,
  \end{system*} 
\end{equation}
for $i=1,\hdots, N$, satisfying Assumption \ref{ass}. With the choice of
\[
  S=\begin{pmatrix}
    1&0&0&0\\
    0&1&0&-1\\
    0&0&1&0\\
    0&1&0&0\\
  \end{pmatrix}, \text{ and } T=I,
\]

we get the matrices $\tilde{A}, \tilde{B}, \tilde{C}$, and $\tilde{E}$ as:
\[
  \tilde{A}=\begin{pmatrix}
    0&0&1&1\\
    -1&-2&1&2\\
    0&-1&0&1\\
    -1&0&1&0\\
  \end{pmatrix}, \text{ }   
  \tilde{B}=\begin{pmatrix}
    0\\0\\0\\1\end{pmatrix}, \text{ }
      \tilde{C}=\begin{pmatrix}
    1&0&0&0\\
    0&0&0&1
  \end{pmatrix}, \text{and  }   
   \tilde{E}=\begin{pmatrix}
    0\\0\\0\\1\end{pmatrix}.
\]
 Then, we choose $H_1$ such that $A_{11}+H_1C_1$ is asymptotically stable:
  \[
  H_1=\begin{pmatrix}
    -1\\0\\-1
  \end{pmatrix}.
\]

Solving \eqref{eq-Riccati}, we obtain $P$ as
\[
  P=\begin{pmatrix}
    3.0498&   -0.7942&    2.0169&    0.8943\\
   -0.7942&    0.9875&   -1.7544&   -0.7475\\
    2.0169&   -1.7544&    4.8899&    2.5890\\
    0.8943&   -0.7475&    2.5890&    2.2308\\   
  \end{pmatrix}.
\]

With the design parameters above, we obtain the adaptive noncollaborative protocol \eqref{protocol} as following
\begin{equation}\label{protocol-example}
    \begin{system*}{ccl}
      \dot{\hat{\xi}}_{1i} &=& \begin{pmatrix}
          0& 0& 1\\ -1&-2&1\\ 0 &-1& 0
      \end{pmatrix} \hat{\xi}_{1i} + \begin{pmatrix}
          1\\ 2\\ 1
      \end{pmatrix} \zeta_{2i}
      + \begin{pmatrix}
    -1\\0\\-1
  \end{pmatrix}(\begin{pmatrix}
    1&0&0
  \end{pmatrix}\hat{\xi}_{1i} - \zeta_{1i}) \\
      \hat{\xi}_i &=& \begin{pmatrix} \hat{\xi}_{1i} \\ \zeta_{2i} \end{pmatrix} \\
      \dot{\rho}_i &=& \begin{cases} \hat{\xi}_i\T
        \begin{pmatrix}
                0.7998 &  -0.6685  &  2.3154 &   1.9950\\
   -0.6685  &  0.5588  & -1.9353  & -1.6676\\
    2.3154  & -1.9353  &  6.7030  &  5.7756\\
    1.9950  & -1.6676   & 5.7756  &  4.9766
        \end{pmatrix}\hat{\xi}_i & 
        \text{ if } \hat{\xi}_i\T P \hat{\xi}_i \geq d, \\
        0 & \text{ if } \hat{\xi}_i\T P \hat{\xi}_i < d,
      \end{cases} \\
    u_i &=& -\rho_i \begin{pmatrix}
        0.8943 &  -0.7475  &  2.5890  &  2.2308
    \end{pmatrix} \hat{\xi}_i.
  \end{system*}
\end{equation}
where, considering $T=I$ and \eqref{STT}, $\zeta_{2i}$ is equal to the
second element of $\zeta_{i}$ in \eqref{zetanoise}.

\subsubsection{\textbf{Scalability -- independence to the size of communication networks}}\label{Size}

We consider MAS with agent models
\eqref{agent-g-noise-Example} and disturbances
\begin{equation}\label{w_i}
  w_i(t)=\sin(0.1it+0.01t^2), \quad i=1,\hdots, N.
\end{equation}
To illustrate the scalability of proposed noncollaborative protocols,
we study three MAS with $5$, $25$, and $121$ agents communicating over
directed Vicsek fractal graphs shown in
Figure~\ref{vicsek-fractal-graph}. In this example, we consider
$d=0.5$. 

The simulation results presented in
Figures~\ref{NonCol_Vicsek_N5}-\ref{NonCol_Vicsek_N121} show the
scalability of our one-shot-designed noncollaborative protocol. In other words,
we achieve scale-free $\delta$-level-coherent
output synchronization regardless of the size of the network.

\begin{figure}[htp!]
  \centering
    \includegraphics[width=0.75\textwidth]{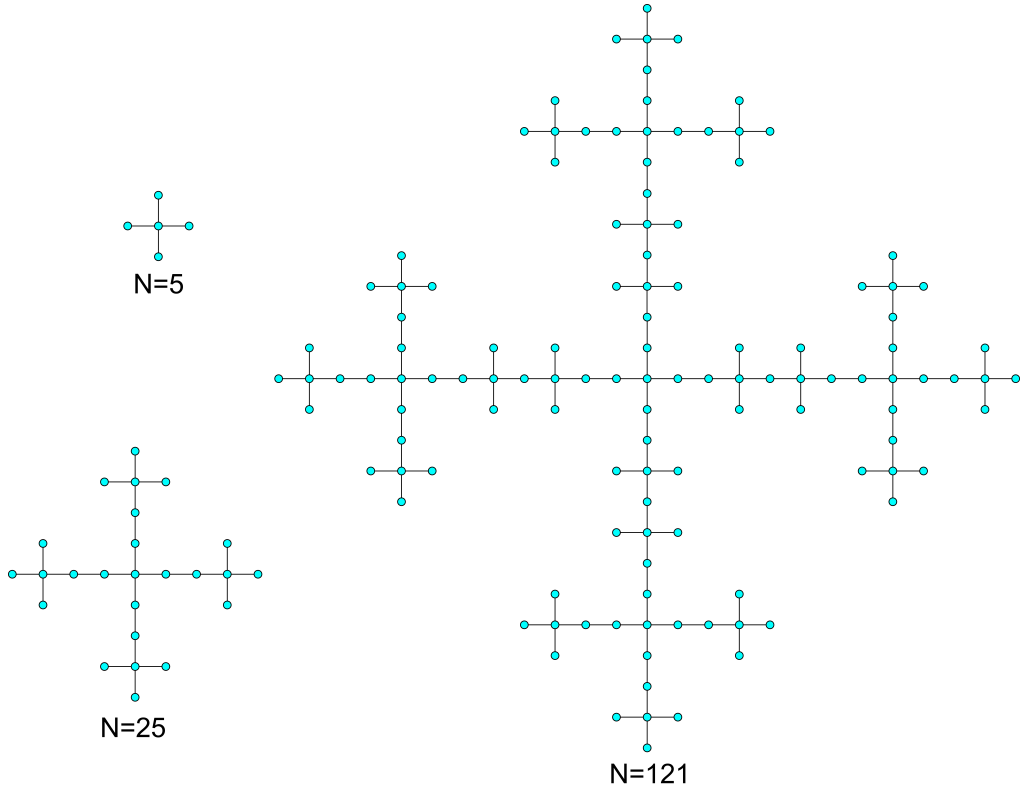}
  \caption{Vicsek fractal graphs}\label{vicsek-fractal-graph}
\end{figure}

\begin{figure}[htp!]
  \centerline{\includegraphics[width=0.1\textwidth]{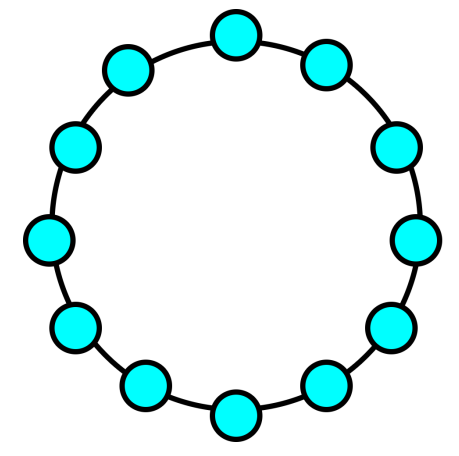}}
  \caption{Circulant graph}\label{Circulant_Graph_UD}
\end{figure}

\begin{figure}[htp!]
 \centering{\includegraphics[width=0.4\textwidth]{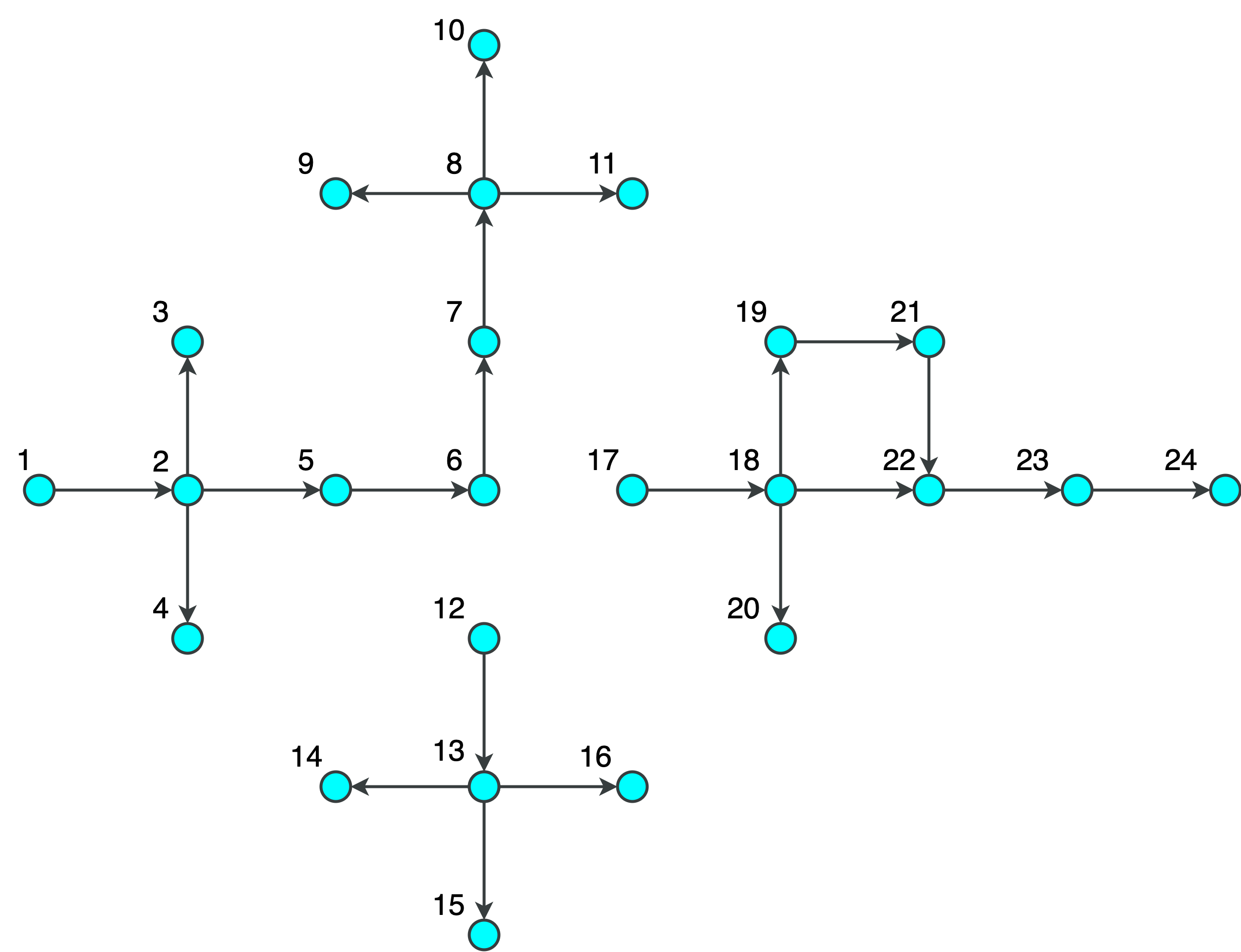}}
  \caption{Disconnected directed graph}\label{dis-directed-net}
\end{figure}

\begin{figure}[p!]
  \centering \includegraphics[width=0.75\textwidth]{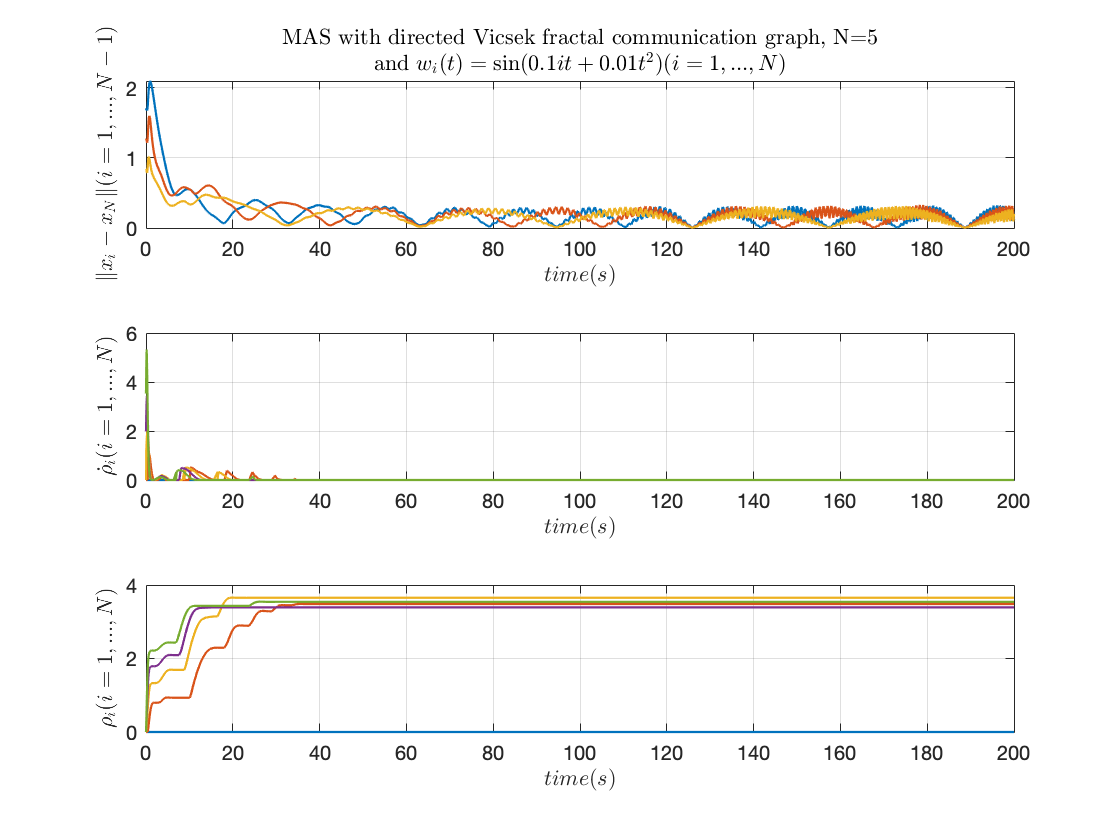}
  \vspace{-3mm}
  \caption[]{{Scale-free $\delta$-level-coherent output
      synchronization of MAS \eqref{agent-g-noise-Example} with $N=5$,
      communicating over directed Vicsek fractal communication graphs
      in the presence of disturbances \eqref{w_i}, via
      noncollaborative protocol \eqref{protocol-example} with
      $d=0.5$}} \label{NonCol_Vicsek_N5}
\end{figure}

\begin{figure}[p!]
	\centering
	\includegraphics[width=0.75\textwidth]{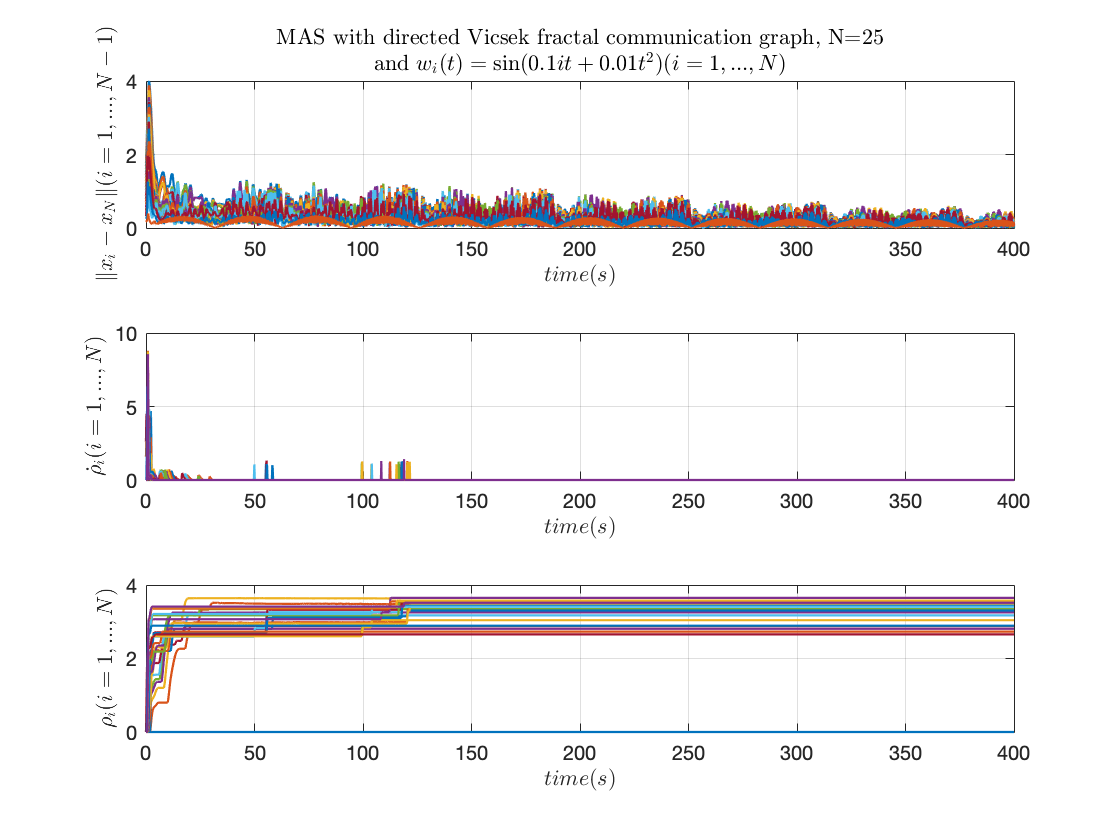}
  \vspace{-3mm}
  \caption[]{{Scale-free $\delta$-level-coherent output synchronization of MAS \eqref{agent-g-noise-Example}
      with $N=25$, communicating over directed Vicsek fractal communication graphs in
      the presence of disturbances \eqref{w_i}, via noncollaborative protocol
      \eqref{protocol-example} with $d=0.5$}} \label{NonCol_Vicsek_N25}
\end{figure}
\begin{figure}[p!]
  \centering \includegraphics[width=0.75\textwidth]{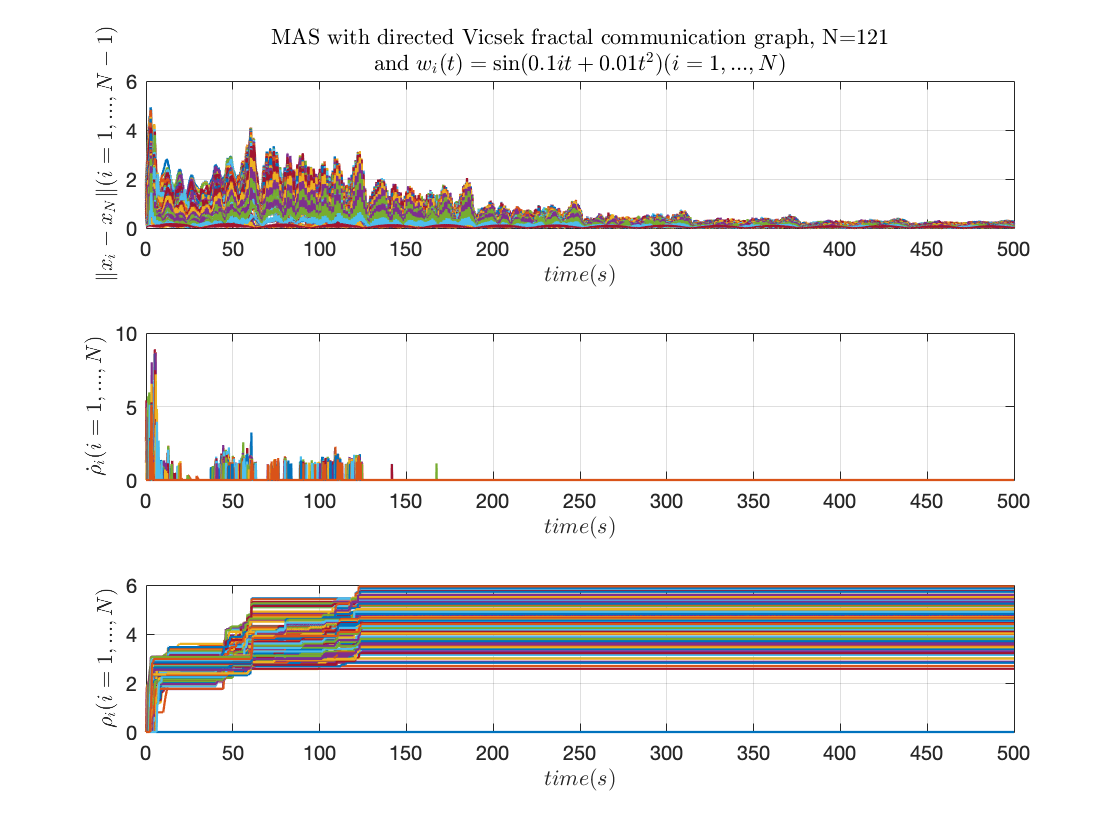}
  \vspace{-3mm}
  \caption[]{{Scale-free $\delta$-level-coherent output
      synchronization of MAS \eqref{agent-g-noise-Example} with
      $N=121$, communicating over directed Vicsek fractal
      communication graphs in the presence of disturbances
      \eqref{w_i}, via noncollaborative protocol
      \eqref{protocol-example} with
      $d=0.5$}} \label{NonCol_Vicsek_N121}
\end{figure}

\begin{figure}[p!]
  \centering
  \includegraphics[width=0.75\textwidth]{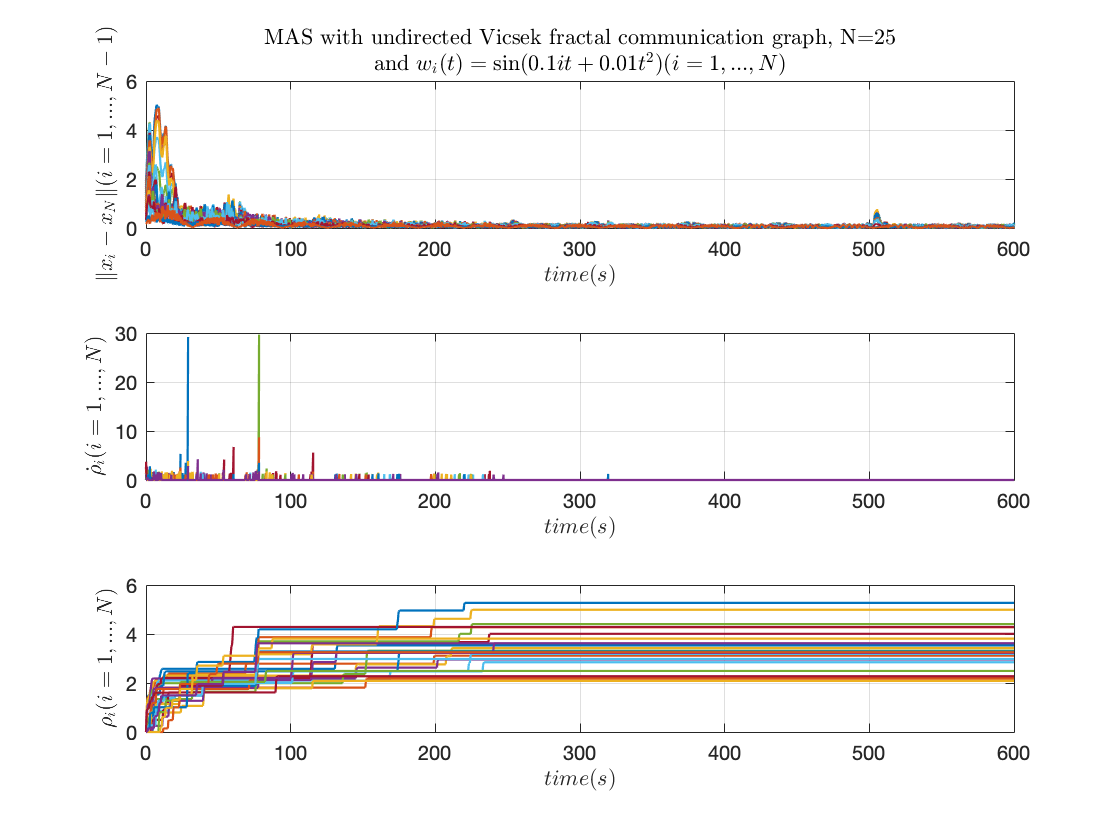}
  \vspace{-3mm}
  \caption[]{{Scale-free $\delta$-level-coherent output
      synchronization of MAS \eqref{agent-g-noise-Example} with
      $N=25$, communicating over undirected Vicsek fractal
      communication graphs in the presence of disturbances
      \eqref{w_i}, via noncollaborative protocol
      \eqref{protocol-example} with
      $d=0.5$}} \label{NonCol_Vicsek_UN_N25}
\end{figure}

\begin{figure}[th!]
  \centering \includegraphics[width=0.75\textwidth]{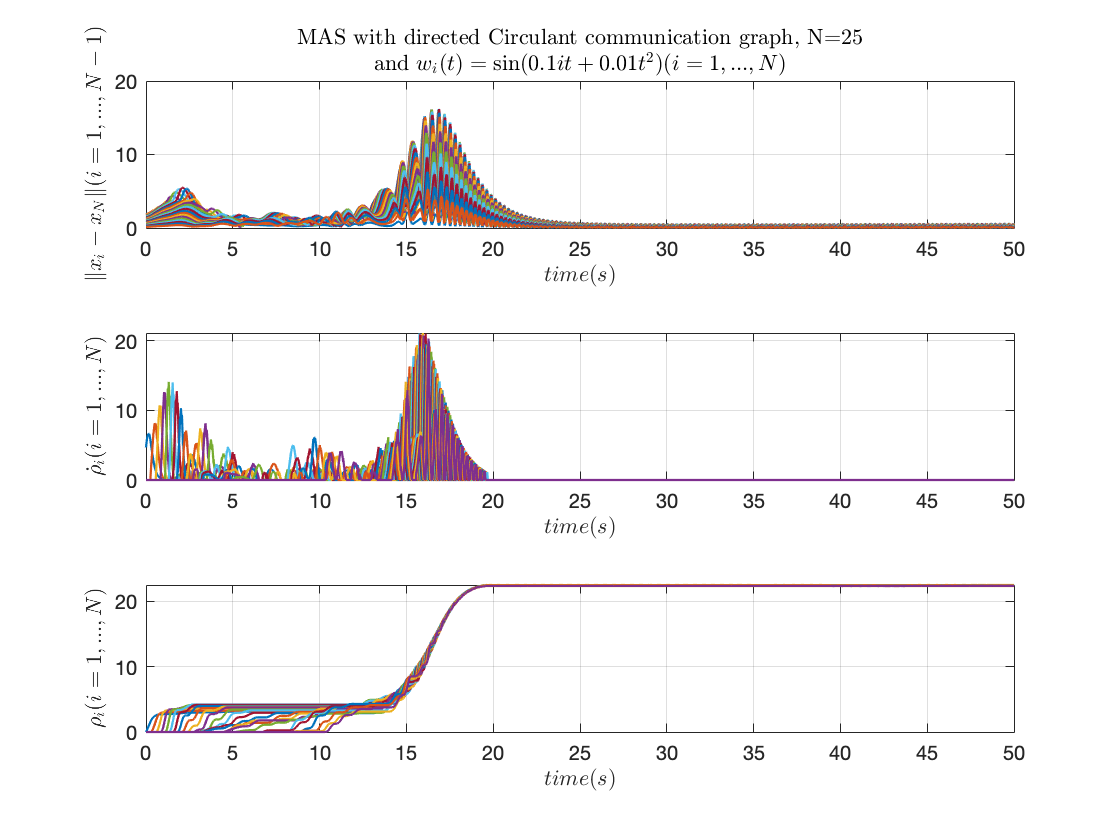}
  \vspace{-3mm}
  \caption[]{{Scale-free $\delta$-level-coherent output
      synchronization of MAS \eqref{agent-g-noise-Example} with
      $N=25$, communicating over directed Circulant communication
      graphs in the presence of disturbances \eqref{w_i}, via
      noncollaborative protocol \eqref{protocol-example} with
      $d=0.5$}} \label{NonCol_Cir_N25}
\end{figure}

\begin{figure}[p!]
  \centering \includegraphics[width=0.75\textwidth]{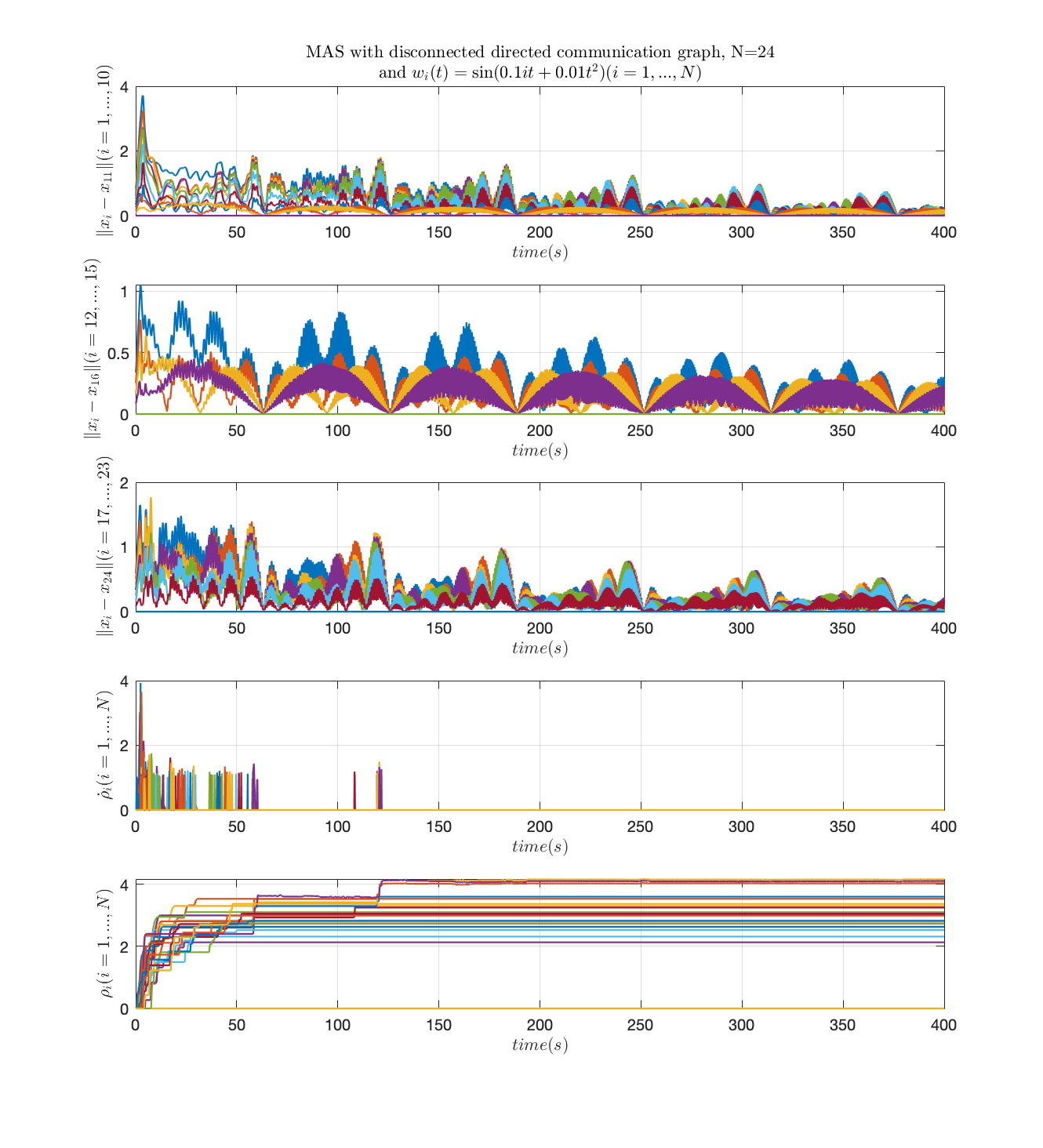}
  \vspace{-3mm}
  \caption[]{{Scale-free $\delta$-level-coherent output
      synchronization of MAS \eqref{agent-g-noise-Example} with
      $N=24$, communicating over disconnected directed communication
      graphs in Figure \ref{dis-directed-net} and in the presence of disturbances \eqref{w_i}, via
      noncollaborative protocol \eqref{protocol-example} with
      $d=0.5$}} \label{NonCol_dis_N24}
\end{figure}

% \begin{figure}[p!]
%   \centering \includegraphics[width=0.75\textwidth]{NonCol_dis_UN_N24.png}
%   \vspace{-3mm}
%   \caption[]{{$\delta$-level coherent state synchronization of MAS
%       with $N=24$ and undirected disconnected communication graphs in the
%       presence of disturbances via nonlinear protocol with
%       $d=0.5$}} \label{NonCol_dis_UN_N24}
% \end{figure}

\subsubsection{Independence to the directedness of the communication
  network}

In this example, we show the feasibility of our noncollaborative
protocols when the agents are communicating over undirected Vicsek
fractal graphs. We consider MAS \eqref{agent-g-noise-Example} with
$N=25$ agents where the agents are subject to noise \eqref{w_i}. We
consider $d=0.5$ in our protocol.  The simulation results presented in
Figure~\ref{NonCol_Vicsek_UN_N25} show that the one-shot designed
protocol \eqref{protocol-example}, achieves scale-free $\delta-$level
coherent output synchronization regardless of the directedness of the
communication graphs.

% \begin{table}[t!]
%   \caption{Algebraic connectivity of undirected Vicsek fractal graphs}
%   \centering
%   \begin{tabularx}{.3\columnwidth}{X X l}
%     \hline
%     $N$& g&$\re\{\lambda_2\}$\\
%     \hline
%     5   &  1&1\\
%     25   &  2&0.0692\\
%     121   &3&0.0053\\
%     \hline
%   \end{tabularx}
%   \label{table_vicsek}
% \end{table}
\subsubsection{\textbf{Effectiveness with different types of
    communication graphs}}\label{Graph}

Next, we illustrate that the designed noncollaborative protocol
achieves $\delta$-level coherent output synchronization for different
types of communication graphs.  We consider MAS
\eqref{agent-g-noise-Example} with $N=25$ where the agents are subject
to noise \eqref{w_i}. In this example, the agents are communicating
through directed Circulant graphs shown in
Figure~\ref{Circulant_Graph_UD}.  The effectiveness of our designed
noncollaborative protocol \eqref{protocol-example} for MAS with
different type of communication graphs, \emph{i.e.} Circulant graphs, is shown in Figure
\ref{NonCol_Cir_N25}.

\subsubsection{\textbf{Effectiveness with disconnected graphs}} \label{connectivity}

Now, we consider MAS
\eqref{agent-g-noise-Example} with $N=24$ in the presence of noise
\eqref{w_i} where the agents are communicating through disconnected
graph shown in Figure \ref{dis-directed-net} which consists of three
bi-components. We consider $d=0.5$ in our protocols.

The simulation results are presented in Figure
\ref{NonCol_dis_N24}. The simulation results show that our designed
protocol is agnostic to the communication network and achieves
$\delta$-level coherent output synchronization for the network
bi-components. We also show the convergence of $\rho_i(t)$ to
constants.
% The simulation results show that our proposed protocols effectively
% achieve $\delta$-level state synchronization where the agents
% communicate through disconnected directed graphs.

\subsubsection{\textbf{Robustness to different noise patterns}}\label{Noise}

In this example, we illustrate the robustness of our protocols to
different noise patterns. We consider MAS \eqref{agent-g-noise-Example} with $N=25$, communicating through a directed Vicsek fractal graphs as in section \ref{Size}. In this example, we assume that agents are subject to
\begin{equation}\label{w_i2}
  w_i(t)=it- round(it), \quad i=1,\ldots, N.
\end{equation}
Figure \ref{NonCol_Noise_Vicsek_N25} shows that our designed noncollaborative protocol is
robust even in the presence of noises with different patterns.

\begin{figure}[th!]
  \centering
  \includegraphics[width=0.75\textwidth]{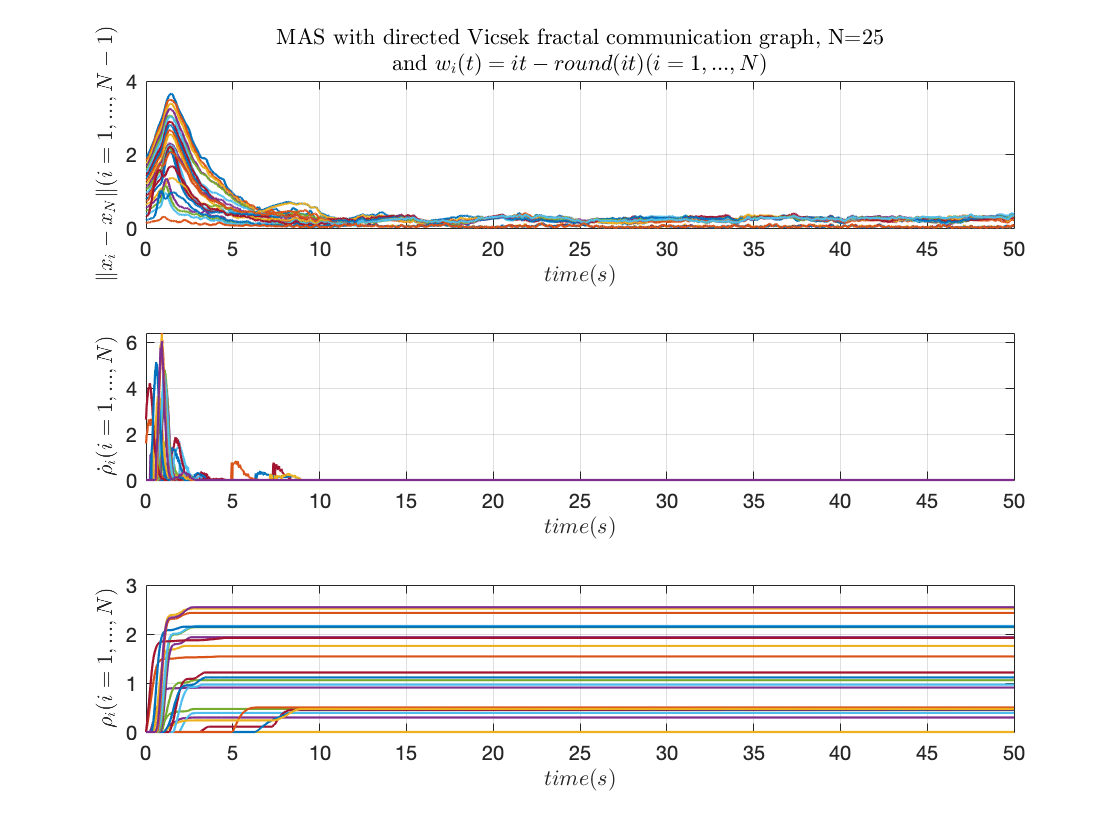}
  \vspace{-3mm}
  \caption[]{{Scale-free $\delta$-level-coherent output
      synchronization of MAS \eqref{agent-g-noise-Example} with
      $N=25$, communicating over directed Vicsek fractal communication
      graphs in the presence of disturbances \eqref{w_i2}, via
      noncollaborative protocol \eqref{protocol-example} with
      $d=0.5$}} \label{NonCol_Noise_Vicsek_N25}
\end{figure}
\begin{figure}[th!]
  \centering
  \includegraphics[width=0.75\textwidth]{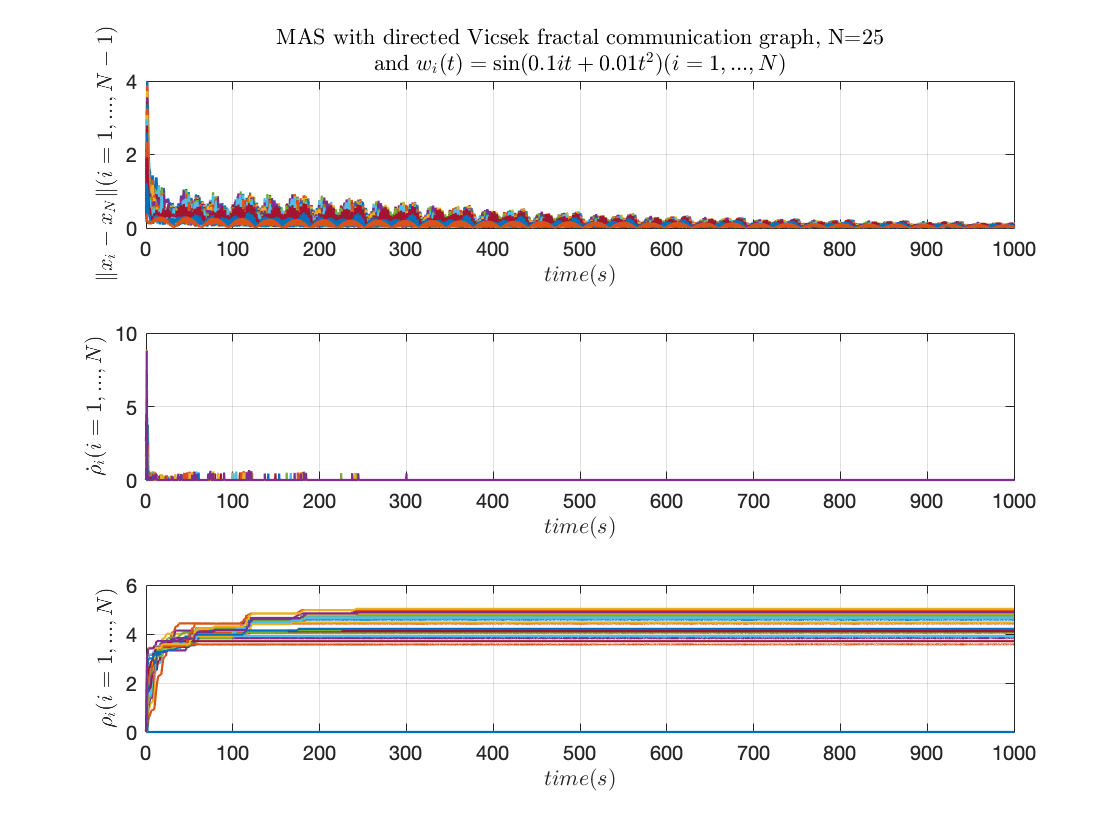}
  \vspace{-3mm}
  \caption[]{{Scale-free $\delta$-level-coherent output
      synchronization of MAS \eqref{agent-g-noise-Example} with
      $N=25$, communicating over directed Vicsek fractal communication
      graphs in the presence of disturbances \eqref{w_i}, via
      noncollaborative protocol \eqref{protocol-example} with
      $d=0.2$}} \label{NonCol_Vicsek_d02_N25}
\end{figure}

\subsubsection{\textbf{Effectiveness for different values of \texorpdfstring{$\delta$}{delta}}}

Finally, in this section, we show the effectiveness of the proposed
protocol for different values of $\delta$ (or, equivalently, different
values of $d$).  Similarly to the previous examples, we consider the MAS \eqref{agent-g-noise-Example} with $N=25$, communicating through directed
Vicsek fractal graphs as in section \ref{Size}, in the presence of noise
\eqref{w_i} where in this example, we choose $d=0.2$. The simulation
results presented in Figure \ref{NonCol_Vicsek_d02_N25} show the
effectiveness of our noncollaborative protocol independent of the value of
$d$.

\subsection{Collaborative Protocols}\label{xmpl2}

In this section, we consider agent models as
\begin{equation}\label{agent-g-noise-Example2}
  \begin{system*}{cl}
    \dot{x}_i&=\begin{pmatrix}-1&1&0\\0&-2&1\\0&0&-3
    \end{pmatrix}x_i+\begin{pmatrix}
      1\\1\\1
    \end{pmatrix}u_i+\begin{pmatrix}
      1\\0\\1
    \end{pmatrix}w_i,\\
    y_i&=\begin{pmatrix}1&0&1\end{pmatrix}x_i
  \end{system*} 
\end{equation}
for $i=1,\hdots, N$, satisfying Assumption \ref{asscol}.

With the choice of $\eta=1$, we obtain 
\[
Q=\begin{pmatrix}
    0.4117  &  0.1136  &  0.0086\\
    0.1136  &  0.2997  &  0.0553\\
    0.0086  &  0.0553  &  0.1792
\end{pmatrix},
\]
as the solution of \eqref{are0}. We obtain the collaborative adaptive protocol as follows.
\begin{equation}\label{protocol-example2}
    \begin{system*}{ccl}
      \dot{\hat{x}}_{i} &=& \begin{pmatrix}-1&1&0\\0&-2&1\\0&0&-3
    \end{pmatrix}\hat{x}_{i} -\alpha_i\begin{pmatrix}
     1   &  1  &   1\\
     1   &  1  &   1\\
     1   &  1  &   1
    \end{pmatrix} P_{\alpha_i}(\hat{x}_i+\tilde{\zeta}_i)-\rho_i\begin{pmatrix}
    0.4203\\
    0.1689\\
    0.1879
    \end{pmatrix}(\begin{pmatrix}
             1   &  0   &  1
    \end{pmatrix}\tilde{\zeta}_i+\zeta_i)\\
      \dot{\rho}_i &=& \begin{cases}
      \tilde{e}_i\T \tilde{e}_i & \text{ if } \tilde{e}_i\T \tilde{e}_i \geq d \\ 
      0 & \text{ otherwise}
    \end{cases}\\
        \dot{\alpha}_i &=& \begin{cases}
      1 & \text{ if } \tilde{\zeta}_i\T C\T C \tilde{\zeta}_i \geq 1 \\
      \tilde{\zeta}_i\T C\T C \tilde{\zeta}_i & \text{ if } 1 >
      \tilde{\zeta}_i\T C\T C  \tilde{\zeta}_i \geq d \\  
      0 & \text{ otherwise}
    \end{cases}\\
    u_i &=& -\alpha_i\begin{pmatrix}
     1   &  1  &   1
    \end{pmatrix} P_{\alpha_i}(\hat{x}_i+\tilde{\zeta}_i).
  \end{system*}
\end{equation}
where $
  C\T C=\begin{pmatrix}
     1   &  0  &   1\\
     0   &  0  &   0\\
     1   &  0  &   1
    \end{pmatrix}$ and $\tilde{e}_i
    = \begin{pmatrix}1&0&1\end{pmatrix}\tilde{\zeta}_i-\zeta_i$ and
    $P_{\alpha_i} \text{ for } i=1,..., N$, are the solutions of
    \eqref{are}. Note that $\zeta_{i}$ and $\tilde{\zeta}_i$ are the
    information exchange between agents as defined in
    \eqref{zetanoise} and \eqref{zetanoiseu}, respectively. 
    
    In this example similar to example \ref{xmpl1}, in sections
    \ref{Size2}- \ref{DValue2}, we will show the effectiveness of our
    one-shot-designed collaborative protocols independent of the
    number of agents, communication network, noise patterns, and value
    of $\delta$. 
    
\subsubsection{\textbf{Scalability -- independence to the size of communication networks}}\label{Size2}

In this example, we consider MAS with agent models
\eqref{agent-g-noise-Example2} and disturbances \eqref{w_i}. To
illustrate the scalability of the proposed collaborative protocols, we
study three MAS with $5$, $25$, and $121$ agents communicating through
the directed Vicsek fractal graphs shown in
Figure~\ref{vicsek-fractal-graph}. In this example, we consider
$d=0.5$. The simulation results presented in
Figures~\ref{Col_Vicsek_N5}-\ref{Col_Vicsek_N121} show the scalability
of our one-shot-designed collaborative protocols \eqref{protocol-example2} for networks with
$N=5$, $N=25$ and $N=121$ respectively. In other words, scalable
adaptive collaborative protocols achieve $\delta$-level coherent
output synchronization regardless of the size of the network. 

\begin{figure}[p!]
  \centering \includegraphics[width=0.75\textwidth]{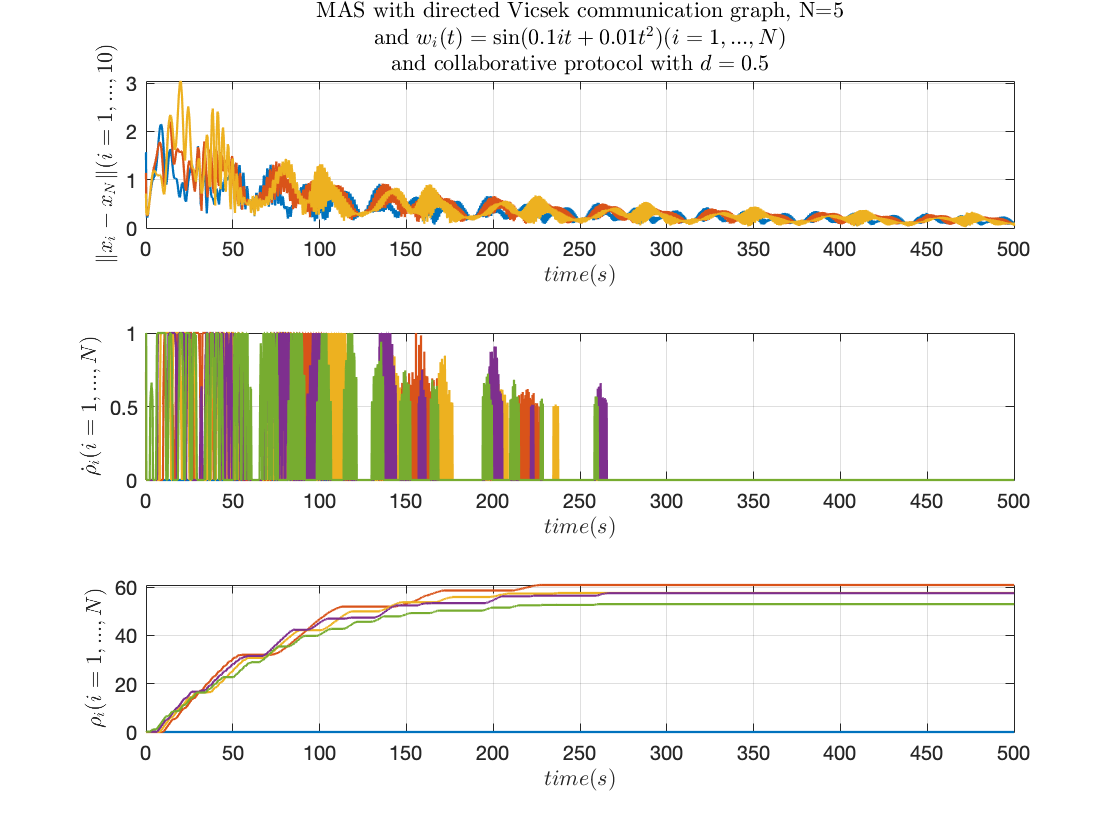}
  \vspace{-3mm}
  \caption[]{{Scale-free $\delta$-level-coherent output
      synchronization of MAS \eqref{agent-g-noise-Example2} with
      $N=5$, communicating over directed Vicsek fractal communication
      graphs in the presence of disturbances \eqref{w_i}, via
      collaborative protocol \eqref{protocol-example2} with
      $d=0.5$}} \label{Col_Vicsek_N5}
\end{figure}

\begin{figure}[p!]
	\centering
	\includegraphics[width=0.75\textwidth]{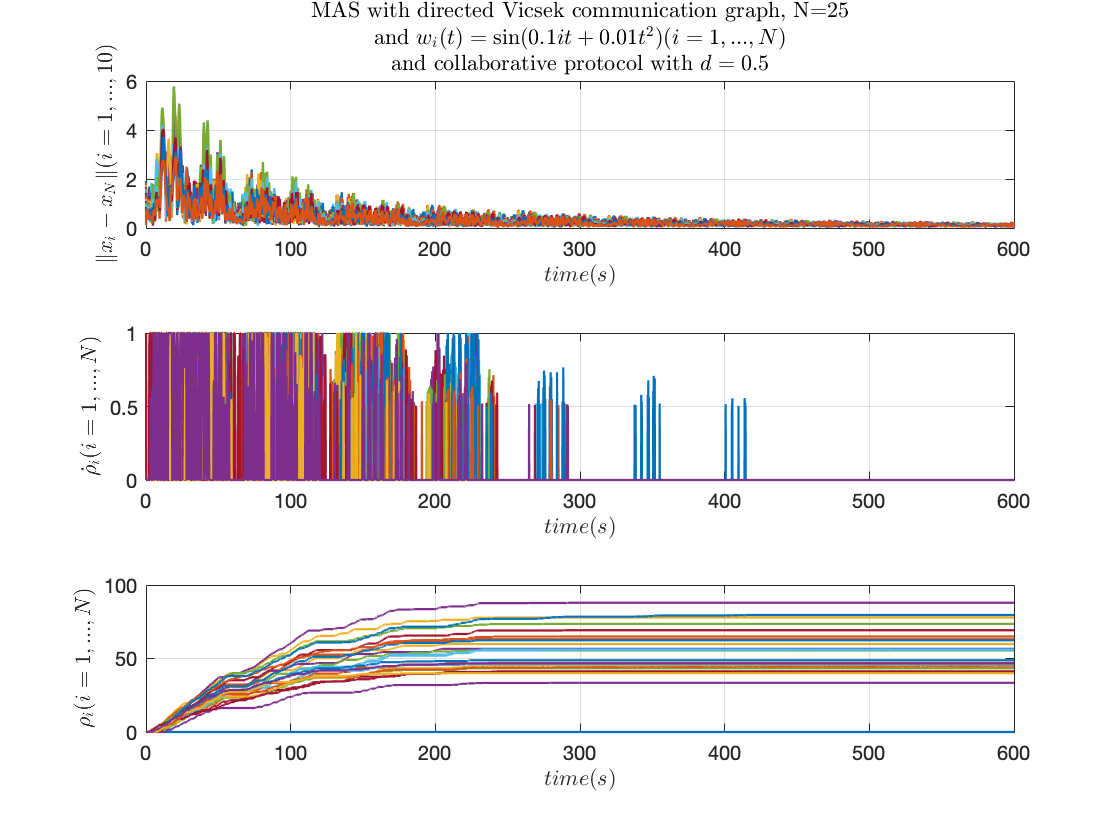}
  \vspace{-3mm}
  \caption[]{{Scale-free $\delta$-level-coherent output
      synchronization of MAS \eqref{agent-g-noise-Example2} with
      $N=25$, communicating over directed Vicsek fractal communication
      graphs in the presence of disturbances \eqref{w_i}, via
      collaborative protocol \eqref{protocol-example2} with
      $d=0.5$}} \label{Col_Vicsek_N25}
\end{figure}
\begin{figure}[p!]
  \centering \includegraphics[width=0.75\textwidth]{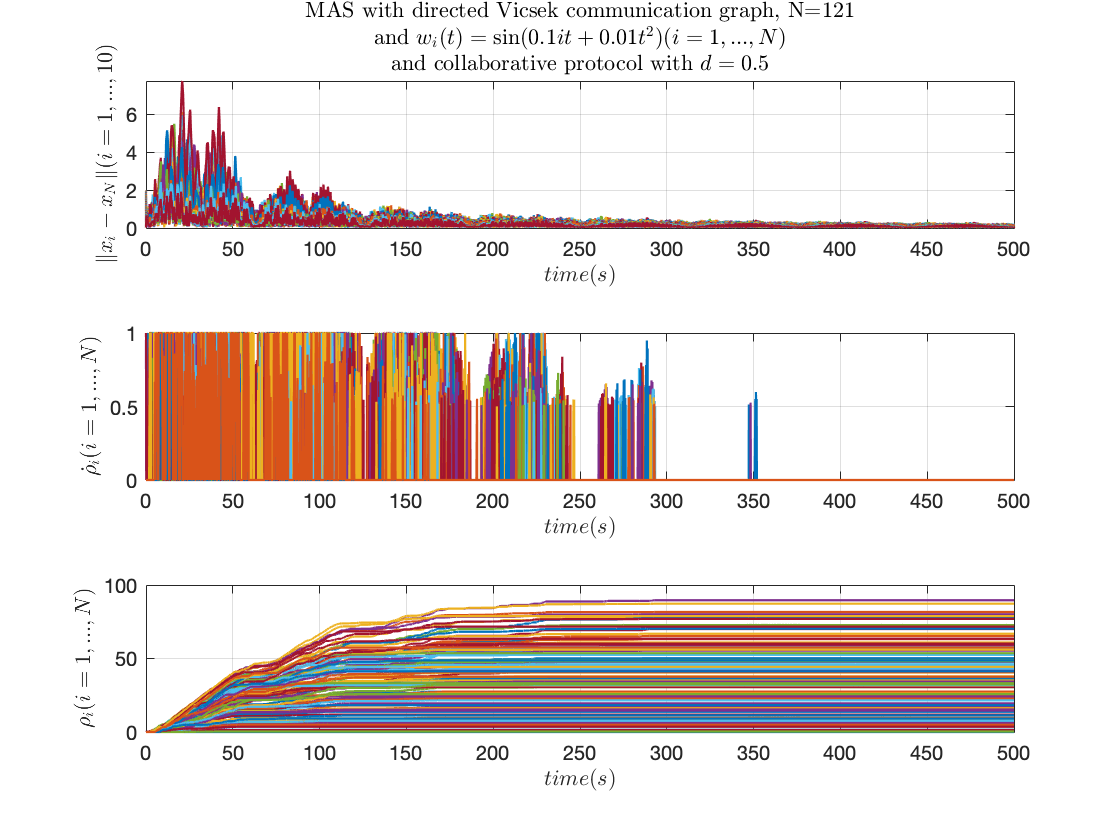}
  \vspace{-3mm}
  \caption[]{{Scale-free $\delta$-level-coherent output
      synchronization of MAS \eqref{agent-g-noise-Example2} with
      $N=121$, communicating over directed Vicsek fractal
      communication graphs in the presence of disturbances
      \eqref{w_i}, via collaborative protocol
      \eqref{protocol-example2} with $d=0.5$}} \label{Col_Vicsek_N121}
\end{figure}

\begin{figure}[p!]
  \centering
  \includegraphics[width=0.75\textwidth]{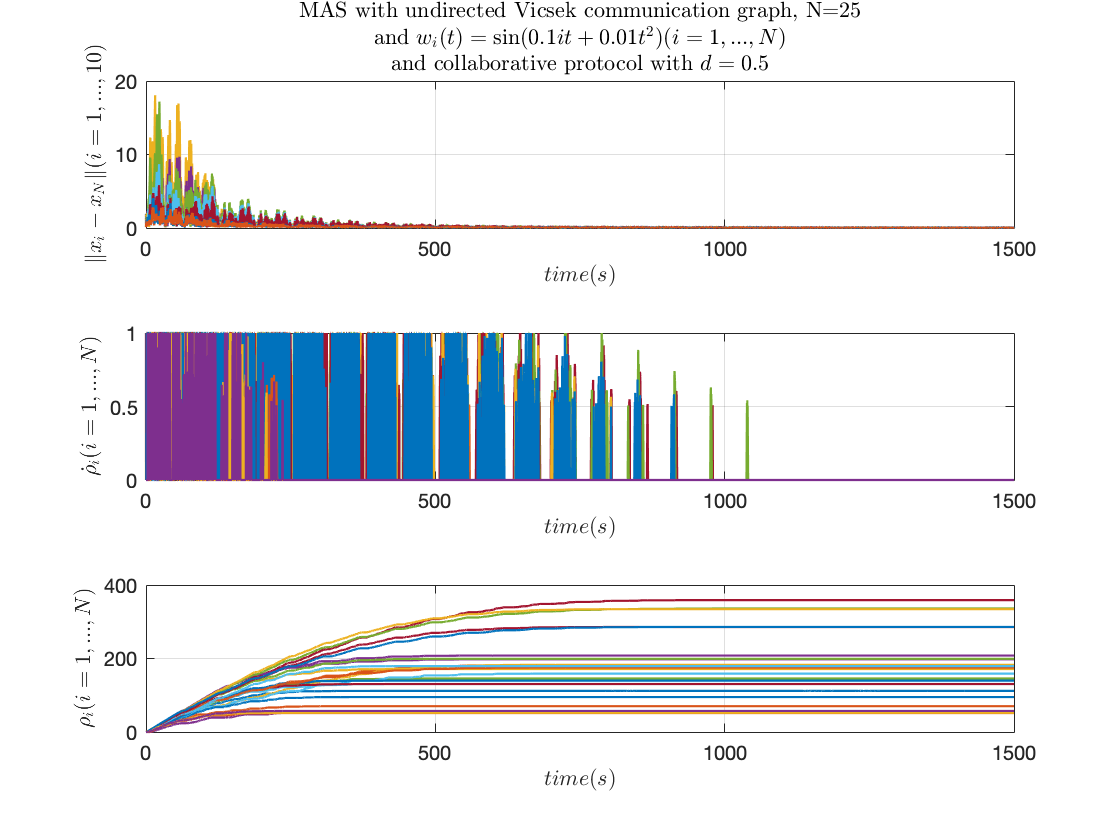}
  \vspace{-3mm}
  \caption[]{{Scale-free $\delta$-level-coherent output
      synchronization of MAS \eqref{agent-g-noise-Example2} with
      $N=25$, communicating over undirected Vicsek fractal
      communication graphs in the presence of disturbances
      \eqref{w_i}, via collaborative protocol
      \eqref{protocol-example2} with
      $d=0.5$}} \label{Col_Vicsek_UN_N25}
\end{figure}

\begin{figure}[th!]
  \centering \includegraphics[width=0.75\textwidth]{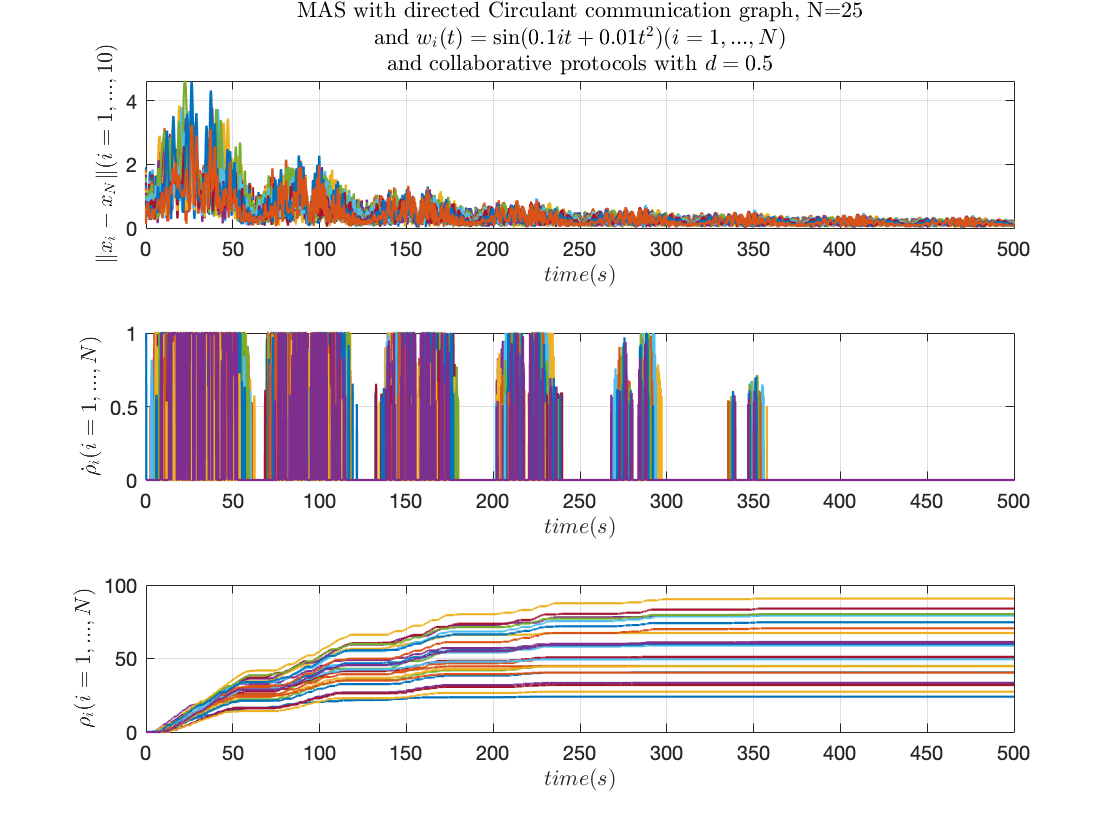}
  \vspace{-3mm}
  \caption[]{{Scale-free $\delta$-level-coherent output
      synchronization of MAS \eqref{agent-g-noise-Example2} with
      $N=25$, communicating over directed Circulant communication
      graphs in the presence of disturbances \eqref{w_i}, via
      collaborative protocol \eqref{protocol-example2} with
      $d=0.5$}} \label{Col_Cir_N25}
\end{figure}
\begin{figure}[p!]
  \centering \includegraphics[width=0.75\textwidth]{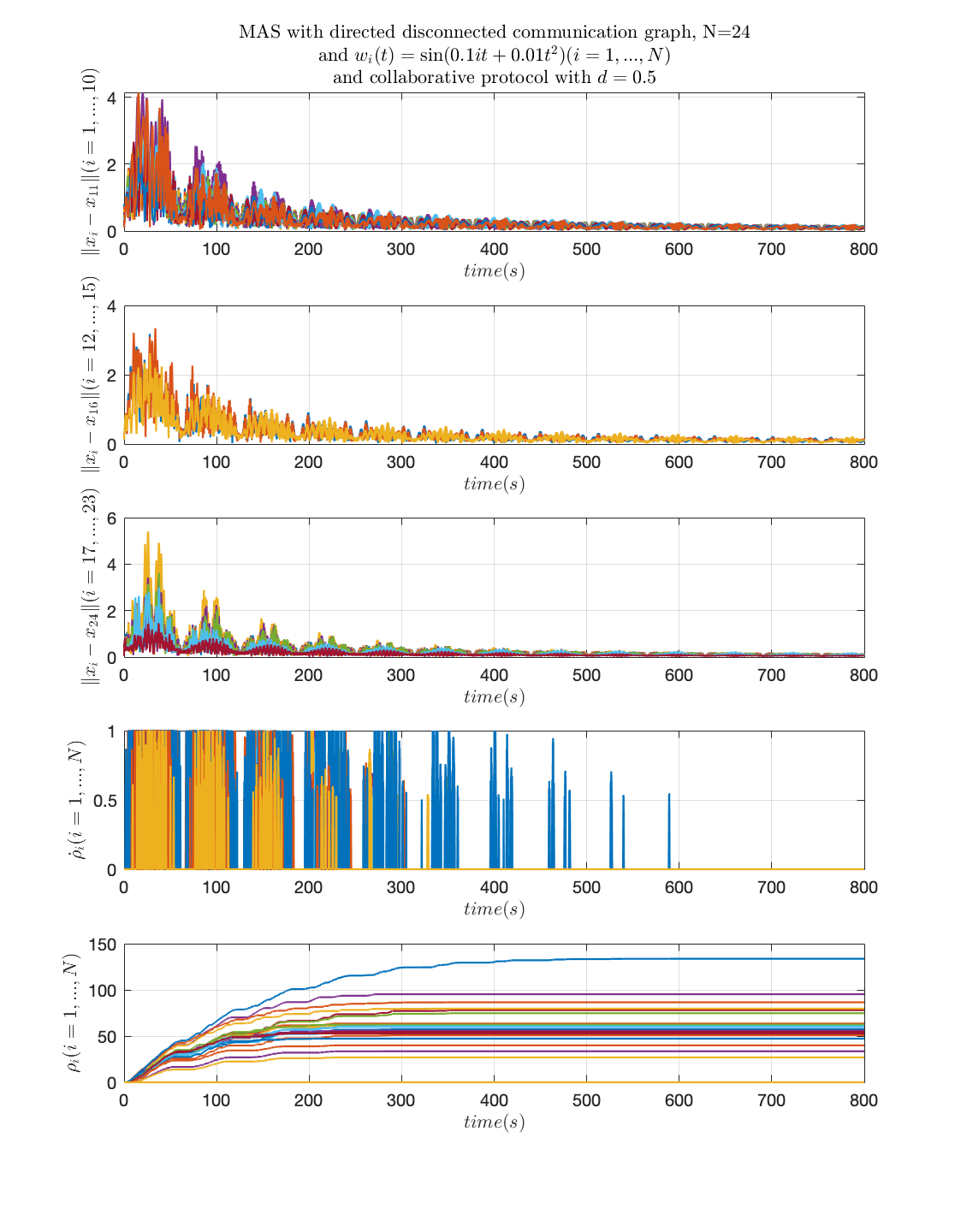}
  \vspace{-3mm}
  \caption[]{{Scale-free $\delta$-level-coherent output
      synchronization of MAS \eqref{agent-g-noise-Example2} with
      $N=24$, communicating over directed disconnected communication
      graphs in the presence of disturbances \eqref{w_i}, via
      collaborative protocol \eqref{protocol-example2} with
      $d=0.5$}} \label{Col_dis_N24}
\end{figure}

\subsubsection{Independence to the directedness of the communication
  network}

In this example, we show the feasibility of our protocols when the
agents are communicating over undirected Vicsek fractal graphs. We
consider MAS with agent models
\eqref{agent-g-noise-Example2} and $N=25$ agents where
the agents are subject to noise \eqref{w_i}. We consider $d=0.5$ in
our protocols.  The simulation results presented in
Figure~\ref{Col_Vicsek_UN_N25} show that the one-shot designed
protocol \eqref{protocol-example2}, achieves $\delta-$level coherent
output synchronization regardless of the directedness of the
communication graphs.

\subsubsection{\textbf{Effectiveness with different types of
    communication graphs}}\label{Graph2}

In this example, we illustrate that the collaborative protocol that we designed also achieves synchronization for different
types of communication graphs.  We consider MAS
\eqref{agent-g-noise-Example2} with $N=25$ where the agents are
subject to noise \eqref{w_i}. In this example, the agents are
communicating through directed Circulant graphs shown in
Figure~\ref{Circulant_Graph_UD}.  The effectiveness of our designed collaborative
protocol \eqref{protocol-example2} for MAS with directed
Circulant communication graphs is shown in Figure \ref{Col_Cir_N25}.

\subsubsection{\textbf{Effectiveness with disconnected
    graphs}} \label{connectivity2}

In this section, we consider MAS communicating through
directed disconnected graphs shown in Figure
\ref{dis-directed-net}. We consider MAS \eqref{agent-g-noise-Example2}
with $N=24$ agents where the agents are
subject to noise \eqref{w_i}. We consider $d=0.5$ in our protocols.\\
The simulation results are presented in Figure
\ref{Col_dis_N24}. In the first three sub-figures, we showed
$\delta$-level-coherent output synchronization for the bi-components
of the disconnected graph. We also showed the convergence of
$\rho_i(t)$ to constants.

\subsubsection{\textbf{Robustness to different noise patterns}}\label{Noise2}

In this example, we analyze the robustness of our collaborative
protocols to different noise patterns. We consider MAS
\eqref{agent-g-noise-Example2} with $N=25$ communicating through a
directed Vicsek fractal graph. In this example, we assume that agents
are subject to \eqref{w_i2}. Figure \ref{Col_Noise_Vicsek_N25} shows
that our designed collaborative protocol is robust even in the presence of noises
with different patterns.

\begin{figure}[th!]
  \centering
  \includegraphics[width=0.75\textwidth]{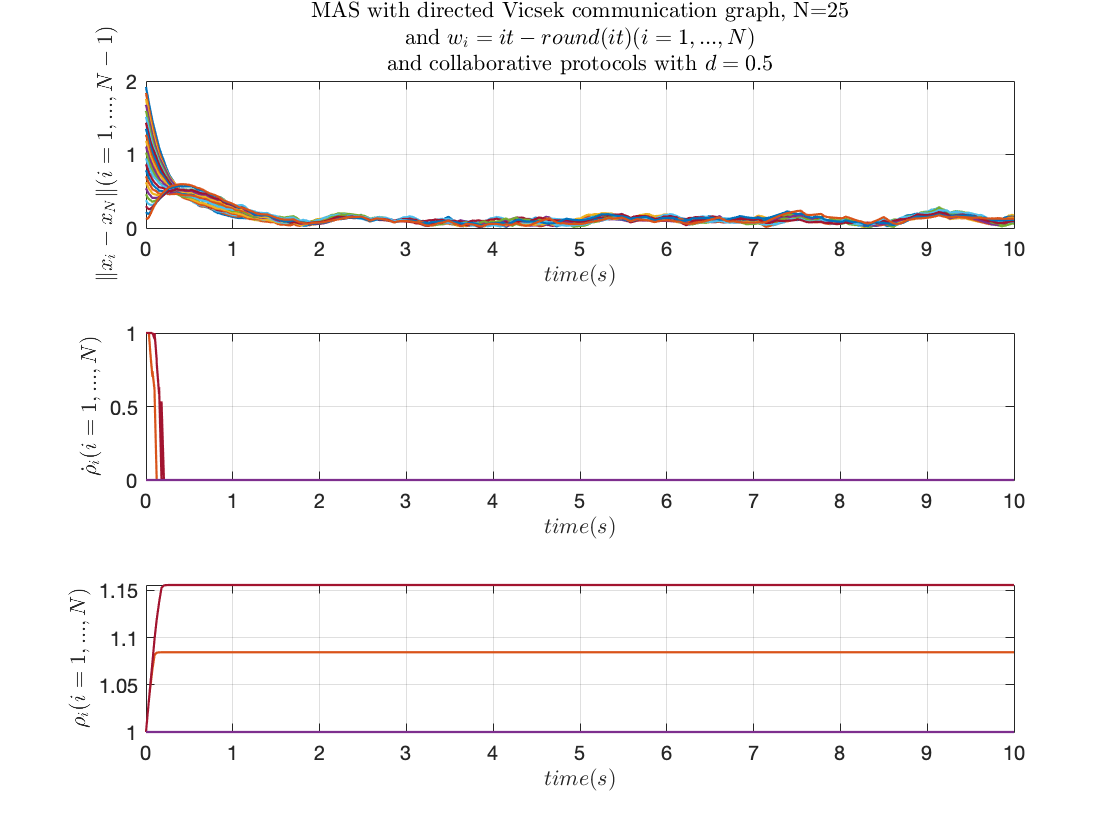}
  \vspace{-3mm}
  \caption[]{{Scale-free $\delta$-level-coherent output
      synchronization of MAS \eqref{agent-g-noise-Example2} with
      $N=25$, communicating over directed Vicsek fractal communication
      graphs in the presence of disturbances \eqref{w_i2}, via
      collaborative protocol \eqref{protocol-example2} with
      $d=0.5$}} \label{Col_Noise_Vicsek_N25}
\end{figure}
\begin{figure}[th!]
  \centering
  \includegraphics[width=0.75\textwidth]{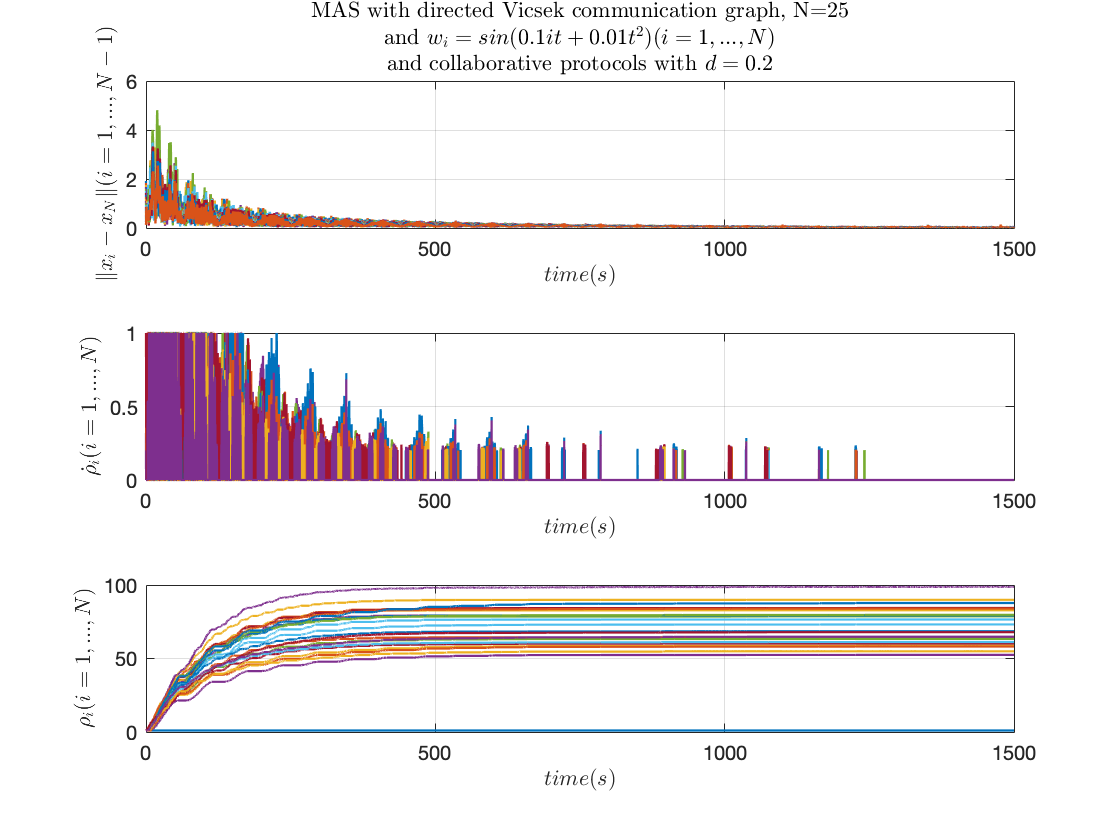}
  \vspace{-3mm}
  \caption[]{{Scale-free $\delta$-level-coherent output
      synchronization of MAS \eqref{agent-g-noise-Example2} with
      $N=25$, communicating over directed Vicsek fractal communication
      graphs in the presence of disturbances \eqref{w_i}, via
      collaborative protocol \eqref{protocol-example2} with
      $d=0.2$}} \label{Col_Vicsek_d02_N25}
\end{figure}

\subsubsection{\textbf{Effectiveness for different values of
    \texorpdfstring{$\delta$}{delta}}}\label{DValue2} 

Finally, in this section, we show the effectiveness of the proposed
collaborative protocol for different values of $\delta$ (or, equivalently, different
values of $d$). Similarly to the previous examples, we consider the MAS \eqref{agent-g-noise-Example2} with $N=25$, communicating through directed
Vicsek fractal graphs as in section \ref{Size2}, where the agents are subject to noise
\eqref{w_i}. In this example, we choose $d=0.2$. The simulation
results presented in Figure \ref{Col_Vicsek_d02_N25} show the
effectiveness of our collaborative protocol independent of the value of
$d$.
\FloatBarrier

\newcounter{equation2}
\setcounter{equation2}{\value{equation}}
\appendix
\setcounter{equation}{\value{equation2}}%

We recall the following lemmas from
\cite{donya-liu-saberi-stoorvogel-arxiv-delta} (with minor modifications):

\begin{lemma}\label{2.8}
  Consider a directed graph with Laplacian matrix $L$ which is strongly
  connected. Then, there exists $h_1,\ldots,h_N>0$ such that
  \begin{equation}\label{Hlyap}
    H^NL +L\T H^N \geq 2\gamma L\T L, 
  \end{equation}
  for some $\gamma >0$ with $H^N$ given by \eqref{HN} for $k=N$.
\end{lemma}

\begin{lemma}\label{2.9}
  The quadratic form
  \[
    V=z\T Q_{\rho} z
  \]
  with $Q_{\rho}$ given by \eqref{Qrho} is decreasing in $\rho_i$ for
  $i=1,\ldots, N$.
\end{lemma}

\begin{lemma}\label{2.10}
  Consider $\tilde{V}_j$ as defined in \eqref{tildeVj} for $j<N$ and
  in \eqref{tildeVN} for $j=N$. There exists $M$ such that for
  $j=2,\ldots ,k$ we have
  \begin{equation}\label{Mbound}
    \tilde{V}_{j-1} \leq M \tilde{V}_j
  \end{equation}
  provided $\tilde{\rho}_j \leq 2\tilde{\rho}_i$ for $i=1,\ldots,j-1$.
\end{lemma}

We need in this paper, a detailled analysis of the parameter
dependency of algebraic Riccati equations. 

Based on the special coordinate basis introduced in
\cite{sannuti-saberi} we can obtain the following decomposition of a
system $(A,B,C,0)$ which is uniform rank, minimum-phase and
right-invertible.

\begin{lemma}\label{scb}
  Consider a linear system $(A,B,C,0)$ which is uniform rank, minimum-phase and
  right-invertible. In that case there exists a basis transformation
  $\Gamma_s$, unitary input and output transformation
  $\Gamma_i$ and $\Gamma_o$, and state feedback $\tilde{F}$ and output
  injection $\tilde{K}$ such that: 
  \[
    \Gamma_s^{-1}A\Gamma_s=\tilde{A}-\tilde{K}\tilde{C}-\tilde{B}\tilde{F},\quad
    \tilde{A}=\begin{pmatrix} A_0 & 0 \\ 0 & A_1
    \end{pmatrix},\quad \tilde{B}=\Gamma_s^{-1}B\Gamma_i = \begin{pmatrix} B_0 & 0 \\
      0 & B_1 \end{pmatrix},\quad \tilde{C}=\Gamma_o C\Gamma_s = \begin{pmatrix} 0 &
      C_1 \end{pmatrix} 
  \]
  where the subsystem $(A_1,B_1,C_1,0)$ based on the SCB has the
  following structure:
  \[
    A_1= \begin{pmatrix} 
      0          & I_{m}      & 0      & \cdots & 0      \\
      \vdots     & \ddots & \ddots & \ddots & \vdots \\
      \vdots     &        & \ddots & I_{m}     & 0      \\
      0          & \cdots & \cdots & 0      & I_{m}     \\
      0          & \cdots & \cdots & 0      & 0
    \end{pmatrix}.
  \]
  Moreover,
  \[
    B_1= \begin{pmatrix} 
      0 \\ \vdots \\ 0 \\ 0 \\ I_{m} 
    \end{pmatrix}
  \]
  and
  \[
    C_1  =\begin{pmatrix} 
      I_{m} & 0 & \cdots & \cdots & 0 & 0
    \end{pmatrix}.
  \]
  Note that $\tilde{F}$ and $\tilde{K}$ have the following structure:
  \[
    \tilde{F}=\begin{pmatrix} 0 & 0 \\ \tilde{F}_{21} & \tilde{F}_{22}
    \end{pmatrix},\qquad \tilde{K}= \begin{pmatrix} \tilde{K}_1 \\
      \tilde{K}_2
    \end{pmatrix}
  \]
  and the unobservable dynamics of $(C,A)$ is equal to the
  unobservable dynamics of $(\tilde{F}_{21},A_0)$. Finally,
  $(A_0,B_0)$ is stabilizable.
\end{lemma}

The above allows us a detailled analysis of the $\alpha$ dependency
of the algebraic Riccati equation:
\begin{equation}\label{are2}
  A\T P_{\alpha} + P_{\alpha} A - \alpha P_{\alpha} BB\T
  P_{\alpha} + 2\eps P_{\alpha} +C\T C = 0
\end{equation}
as presented in the following lemma:

\begin{lemma}\label{asympric}
  Consider a linear system $(A,B,C,0)$ which is uniform rank,
  minimum-phase and right-invertible. Consider the structure as
  presented in Lemma \ref{scb} with $A_1\in \R^{n_1m\times n_1m}$. Define
  \begin{equation}\label{betam}
    \beta_0=\alpha^{-1/2},\qquad \beta_1=\alpha^{-1/(4n_1)}.
  \end{equation}
  and
  \[
    D_{\alpha}=\begin{pmatrix} 
      D_0 & 0      \\
      0   & D_1
    \end{pmatrix},
  \]
  where $D_0=\beta_0 I$ and
  \[
    D_1 = \begin{pmatrix} 
      \beta_1 I & 0      & \cdots & 0      \\
      0         & \beta_1^3 I & \ddots & \vdots \\
      \vdots    & \ddots & \ddots & 0      \\
      0         & \cdots & 0      & \beta_1^{2n_1-1} I
    \end{pmatrix}.
  \]
  Note that $\beta_1^{2n_1}=\alpha^{-1/2}$. In that case, there exist
  matrices $P_{\ell}\geq 0$ and $P_{u} \geq 0$ independent of $\alpha$
  such that
  \[
    \Gamma_s^{-T} D_{\alpha} P_{\ell} D_{\alpha}\Gamma_s^{-1} \leq 
    P_{\alpha} \leq \Gamma_s^{-T} D_{\alpha} P_{u} D_{\alpha}\Gamma_s^{-1} 
  \]
  for $\alpha>1$ with
  $\ker P_{\ell}\Gamma_s^{-1} =\ker{P}_{\alpha}=\ker P_{u}\Gamma_s^{-1}$.
\end{lemma}

\begin{proof}
  Choose $m_{\ell},m_u>0$ such that
  \[
    m_{\ell}I \leq \Gamma_o\T \Gamma_o \leq m_u I,\qquad
    m_{\ell}I \leq \Gamma_i\T \Gamma_i \leq m_u I.
  \]
  Let $P_{u,1}$ be the unique solution of the
  algebraic Riccati equation:
  \[
    A_1\T P_{u,1} + P_{u,1} A_1 - m_{\ell} P_{u,1} B_1B_1\T
    P_{u,1}  +2m_{\ell}^{-1} C_1\T C_1 + P_{u,1}+I =0.
  \]
  while $P_{u,0}$ satisfies
  \[
    A_0\T P_{u,0} + P_{u,0} A_0 - \tfrac{1}{2} m_{\ell} P_{u,0} B_0B_0\T
    P_{u,0} + \eps P_{u,0}+2m_{\ell}^{-1} \tilde{F}_{21}\T \tilde{F}_{21} =0.
  \]
  Note that the kernel of $P_{\alpha}$ is related to the unobservable
  dynamics of $(C,A)$. Note that the unobservable dynamics of $(C,A)$
  in the SCB is given by the unobservable dynamics of
  $(\tilde{F}_{21},A_0)$. To be more specific, the kernel of $P_{u,0}$
  equals the unobservable dynamics of $(\tilde{F}_{21},A_0)$ while
  $P_{u,1}$ is invertible. Choose:
  \[
    P_{u} =\begin{pmatrix} 
      P_{u,0}  & 0      \\
      0       & P_{u,1} 
    \end{pmatrix}.
  \]
  Next, using some straightforward algebra, we can show that 
  $\tilde{P}_{u,\alpha} = \Gamma_s^{-T} D_1 P_{u} D_1
  \Gamma_s^{-1}$ satisfies:
  \begin{equation}\label{are2u}
    A\T \tilde{P}_{u,\alpha} + \tilde{P}_{u,\alpha} A - \alpha \tilde{P}_{u,\alpha} BB\T
    \tilde{P}_{u,\alpha} + \eps \tilde{P}_{u,\alpha} +C\T C  \leq 0
  \end{equation}
  for $\alpha$ sufficiently large, where we used that:
  \[
    \tilde{A} D_{\alpha} =\tilde{D}_{\alpha}^2 D_{\alpha} \tilde{A},\qquad
    \tilde{D}_{\alpha} D_{\alpha} \tilde{B} = \alpha^{-1/2} \tilde{B},\qquad
    \tilde{C} D_{\alpha} = C \tilde{D}_{\alpha}
  \]
  with
  \[
    \tilde{D}_{\alpha}=\begin{pmatrix} 
      I         & 0          \\
      0         & \beta_1 I 
    \end{pmatrix}.
  \]
  Given that $P_{\alpha}$ satisfies \eqref{are2} we then immediately find
  that $P_{\alpha} \leq \tilde{P}_{u,\alpha}$ as required.

  Next, we focus on the lower bound. Note:
  \[
    x_0\T P_\alpha x_0 = \inf_u \left\{ \| y \|_{\eps}^2 + \alpha^{-1}
      \| u \|_{\eps}^2\, \mid x(0)=x_0 \, \right\}
  \]
  where
  \[
    \| u \|_{\eps}^2 = \int_0^{\infty} e^{\eps t} \| u(t) \|^2\,
    \textrm{d} t.
  \]
  We denote by $L_{2,\eps}$ the class of all $L_2$ signals for which
  the above norm is finite. Note, if we define:
  \[
    x_0\T P_{\ell,1} x_0 = \inf_u \left\{ \| y
    \|_{\eps}^2 + \| u \|_{\eps}^2\, \mid x(0)=x_0 \,
  \right\}
  \]
  then it is easy to see that $P_{\alpha} \geq \alpha^{-1} P_{\ell,
    1}$.

  Next, we convert the system into SCB according to Lemma
  \ref{scb} and obtain:
  \[
    \tilde{x}_0\T \tilde{P}_\alpha \tilde{x}_0 \geq \inf_{\tilde{u}}
    \left\{ m_u^{-1} \|
      C_1\tilde{x}_2 \|_{\eps}^2 + \alpha^{-1} m_{\ell} \|
      \tilde{u}_1
      \|_{\eps}^2 + \alpha^{-1} m_{\ell} \| \tilde{u}_2 \|_{\eps}^2 \, \mid
      \tilde{x}(0)=\tilde{x}_0 \, \right\} 
  \]
  where $P_\alpha=\Gamma_s\T \tilde{P}_\alpha \Gamma_s$ and
  \[
    \tilde{x}=\begin{pmatrix} \tilde{x}_0
      \\ \tilde{x}_1 \end{pmatrix} =\Gamma_s x,\qquad u=\Gamma_i \begin{pmatrix} \tilde{u}_1
      \\ \tilde{u}_2 \end{pmatrix},\qquad \tilde{y} = \Gamma_0 y
  \]
  with the decomposition of the input into $\tilde{u}_1$ and
  $\tilde{u}_2$ is compatible with the decomposition of $\tilde{B}$
  presented in Lemma \ref{scb}. Choose $\tilde{\eps}<\eps$ such that
  $A+\tfrac{1}{2} \tilde{\eps}$ has no eigenvalues on the imaginary
  axis. If we decompose $\tilde{P}_\alpha$:
  \[
    \tilde{P}_\alpha = \begin{pmatrix} \tilde{P}_{11,\alpha} &
      \tilde{P}_{12,\alpha} \\ \tilde{P}_{21,\alpha} &
      \tilde{P}_{22,\alpha} \end{pmatrix}
  \]
  then we obtain:
  \begin{equation}\label{costt}
    \tilde{x}_{1,0}\T \tilde{P}_{22,\alpha} \tilde{x}_{1,0} \geq \inf_{\tilde{u}}
    \left\{ m_u^{-1} \|
      C_1\tilde{x}_1 \|_{\tilde{\eps}}^2 + \alpha^{-1} m_{\ell} \|
      \tilde{u}_1
      \|_{\tilde{\eps}}^2 + \alpha^{-1} m_{\ell} \| \tilde{u}_2 \|_{\tilde{\eps}}^2 \, \mid
      \tilde{x}_0(0)=0, \tilde{x}_1(0)=\tilde{x}_{1,0} \, \right\} 
  \end{equation}
  Note that we have:
  \[
    \tilde{x}_0=A_0 \tilde{x}_1+B_0\tilde{u}_1-\tilde{K}_1 C_1
    \tilde{x}_1,\qquad \tilde{x}_0(0)=0.
  \]
  Choose $F_0$ such that $A+B_0F_0+\tfrac{1}{2}\tilde{\eps} I$ is
  asymptotically stable then we have
  $\tilde{x}_0=\tilde{x}_{0,y}+\tilde{x}_{0,v}$
  provided $\tilde{u}_0=F_0\tilde{x}_{0,y}+v_0$ where:
  \[
    \tilde{x}_{0,y}=(A_0+B_0F_0)\tilde{x}_{0,y}-\tilde{K}_1 C_1
    \tilde{x}_1,\qquad \tilde{x}_{0,y}(0)=0.
  \]
  and
  \begin{equation}\label{dyn00}
    \tilde{x}_{0,v}=A_0 \tilde{x}_{0,v}+B_0 v_0 ,\qquad
    \tilde{x}_{0,v}(0)=0.
  \end{equation}
  Note that there exists $\gamma$ such that
  \[
    \| \tilde{F}_{21} \tilde{x}_{0,y} \|_{\tilde{\eps}} \leq \gamma_1 \|
    \tilde{y} \|_{\tilde{\eps}},\qquad
    \| F_{0} \tilde{x}_{0,y} \|_{\tilde{\eps}} \leq \gamma_1 \|
    \tilde{y} \|_{\tilde{\eps}}
  \]
  Next, note that in order to obtain a finite cost in \eqref{costt} we need
  $\tilde{y}$, $\tilde{u}_1$ and $\tilde{u}_2$ in
  $L_{2,\tilde{\eps}}$. This guarantees that we need that $\tilde{x}_0$
  needs to be in $L_{2,\tilde{\eps}}$ as well since the system is
  detectable. Hence if we choose
  \[
    \tilde{u}_0=F_0\tilde{x}_{0,y}+v_0,\qquad
    \tilde{u}_1=-\tilde{F}_{21}(\tilde{x}_{0,y}+\tilde{x}_{0,v}) +v_1
  \]
  then $v_1$ must be such that $\tilde{x}_{0,v}\in L_{2,\tilde{\eps}}$
  But then \eqref{dyn00} implies there exists $\gamma_2>1$ such that
  \begin{equation}\label{boundd21}
    \| \tilde{F}_{21} \tilde{x}_{0,v} \|_{\tilde{\eps}} \leq \gamma_2 \|
    v_0 \|_{\tilde{\eps}}
  \end{equation}
  (recall that $A_0+\tfrac{1}{2}\tilde{\eps}$ does not have any
  imaginary axis eigenvalues) and hence this is a standard bound using
  $L_\infty$ instead of $H_\infty$ which yields:
  \[
    \tilde{x}_{1,0}\T \tilde{P}_{22,\alpha} \tilde{x}_{1,0} \geq \inf_{v_0,v_1}
    \left\{ (m_u^{-1}-\gamma_1\alpha^{-1}m_{\ell}-\gamma_1\alpha^{-1}m_{\ell}) \|
      C_1\tilde{x}_1 \|_{\tilde{\eps}}^2 + \alpha^{-1} m_{\ell} \|
      v_0 \|_{\tilde{\eps}}^2 + \alpha^{-1} m_{\ell} \|
      \tilde{F}_{21}\tilde{x}_{0,v}
      +v_2 \|_{\tilde{\eps}}^2 \, \mid
      \tilde{x}_0(0)=0, \tilde{x}_1(0)=\tilde{x}_{1,0} \, \right\} 
  \]
  This gives us:
  \begin{equation}\label{costt2}
    \tilde{x}_{1,0}\T \tilde{P}_{22,\alpha} \tilde{x}_{1,0} \geq \inf_{v_0,v_1}
    \left\{ \tfrac{1}{2} m_u^{-1} \|
      C_1\tilde{x}_1 \|_{\tilde{\eps}}^2 + \alpha^{-1} m_{\ell} \|
      v_0 \|_{\tilde{\eps}}^2 + \tfrac{1}{\gamma_2} \alpha^{-1} m_{\ell} \|
      \tilde{F}_{21}\tilde{x}_{0,v}
      +v_2 \|_{\tilde{\eps}}^2 \, \mid
      \tilde{x}_0(0)=0, \tilde{x}_1(0)=\tilde{x}_{1,0} \, \right\} 
  \end{equation}
  for $\alpha$ sufficiently large. After another bounding, we get
  \begin{equation}\label{boundd22}
    \tilde{x}_{1,0}\T \tilde{P}_{22,\alpha} \tilde{x}_{1,0} \geq \inf_{v_0,v_1}
    \left\{ \tfrac{1}{2} m_u^{-1} \|
      C_1\tilde{x}_1 \|_{\tilde{\eps}}^2  + \tfrac{1}{\gamma_2} \alpha^{-1} m_{\ell} \|
      v_2 \|_{\tilde{\eps}}^2 \, \mid
      \tilde{x}_0(0)=0, \tilde{x}_1(0)=\tilde{x}_{1,0} \, \right\}
  \end{equation}
  subject to the dynamics:
  \begin{equation}\label{dyn22}
    \dot{\tilde{x}}_1 = (A_1-B_1F_{22}-\tilde{K}_1)\tilde{x}_1 +
    B_1v_1,\qquad \tilde{y}_1 = C_1 \tilde{x}_1
  \end{equation}
  Note that the optimization in \eqref{boundd22} does not depend on
  $v_0$ or $\tilde{x}_0$. Hence if we define:
  \[
    \tilde{x}_{1,0}\T \tilde{P}_{\ell,\alpha} \tilde{x}_{1,0} \geq \inf_{v_1}
    \left\{ \tfrac{1}{2} m_u^{-1} \|
      C_1\tilde{x}_1 \|_{\tilde{\eps}}^2  + \tfrac{1}{\gamma_2} \alpha^{-1} m_{\ell} \|
      v_2 \|_{\tilde{\eps}}^2 \, \mid  \tilde{x}_1(0)=\tilde{x}_{1,0} \, \right\}
  \]
  subject to \eqref{dyn22} then we obtain:
  \[
    \tilde{P}_{22,\alpha} \geq \tilde{P}_{\ell,\alpha}
  \]
  We have:
  \[
    (A_1-\tilde{K}_2C_1-B_1\tilde{F}_{22})\T \tilde{P}_{\ell,\alpha}+
    \tilde{P}_{\ell,\alpha}
    (A_1-\tilde{K}_2C_1-B_1\tilde{F}_{22})-\alpha \tfrac{m_{\ell}}{\gamma_2}
    \tilde{P}_{\ell,\alpha}B_1 B_1\T \tilde{P}_{\ell,\alpha} +
    \tfrac{1}{2} C_1\T
    C_1 +\tilde{\eps} \tilde{P}_{\ell,\alpha} = 0
  \]
  Clearly
  \[
    \tilde{P}_{\ell,\alpha}\tilde{K}_2C_1 +
    C_1\T \tilde{K}_2\T \leq 8 m_u
    \tilde{P}_{\ell,\alpha}\tilde{K}_2\tilde{K}_2\T
    \tilde{P}_{\ell,\alpha} + \tfrac{1}{8} m_u^{-1} C_1\T C_1 \leq
    \tfrac{1}{4}\tilde{\eps} \tilde{P}_{\ell,\alpha} + \tfrac{1}{8} m_u^{-1} C_1\T C_1
  \]
  for large enough $\alpha$ since $ \tilde{P}_{\ell,\alpha}\rightarrow
  0$. Moreover.
  \[
    \tilde{P}_{\ell,\alpha}B_1\tilde{F}_{22} + \tilde{F}_{22}\T B_1\T \tilde{P}_{\ell,\alpha} \leq
    \tfrac{\alpha m_{\ell}}{2\gamma_2} \tilde{P}_{\ell,\alpha} B_1 B_1\T  \tilde{P}_{\ell,\alpha} 
    +  \tfrac{2\gamma_2}{\alpha m_{\ell}} \tilde{F}_{22}\T
    \tilde{F}_{22}
  \]
  We claim
  \[
    \tfrac{2\gamma_2}{\alpha m_{\ell}} \tilde{F}_{22}\T
    \tilde{F}_{22} \leq \tfrac{1}{4} \tilde{\eps}
    \tilde{P}_{\ell,\alpha}
  \]
  for $\alpha$ large enough. Assume not, then there exists for some
  $\mu>0$ a sequence $\alpha_n\rightarrow \infty$ such that
  \[
    \tilde{x}_{1,0}\T \tilde{P}_{\ell,\alpha_n}\tilde{x}_{1,0} \geq \mu \alpha_n^{-1}
  \]
  But given the definition of $\tilde{P}_{\ell,\alpha_n}$ this implies
  there exists a sequence $v_{1,n}\in L_{2,\tilde{\eps}}$ resulting in a state
  $\tilde{x}_{1,n}$ and output $\tilde{y}$ such that
  \[
    \| v_{1,n} \|_{\tilde{\eps}} \leq \mu, \qquad \| \tilde{y}_n
    \|_{\tilde{\eps}} \rightarrow 0 
  \]
  but then the sequence $\{v_{1,n}\}$ has a weakly convergent
  subsequence with limit $v_1^{*}$ with corresponding output
  $\tilde{y}^{*}=0$ and a nonzero initial condition. This yields a
  contradiction since by construction the system with state
  $\tilde{x}_1$ is invertible and has no zero-dynamics.

  Using the above bounds, we find that $\tilde{P}_{\ell,\alpha}$
  satisfies:
  \[
    A_1\T \tilde{P}_{\ell,\alpha}+
    \tilde{P}_{\ell,\alpha} A_1-\alpha \tfrac{m_{\ell}}{2\gamma_2}
    \tilde{P}_{\ell,\alpha}B_1 B_1\T \tilde{P}_{\ell,\alpha} +
    \tfrac{1}{4}m_u^{-1} C_1\T
    C_1 +\tfrac{1}{2} \tilde{\eps} \tilde{P}_{\ell,\alpha} \leq 0
  \]
  But then
  \[
    \tilde{P}^{*}_{\ell,\alpha} \leq \tilde{P}_{\ell,\alpha} \leq
    \tilde{P}_{22,\alpha}
  \]
  where $\tilde{P}^{*}_{\ell,\alpha}$
  satisfies:
  \[
    A_1\T \tilde{P}^{*}_{\ell,\alpha}+
    \tilde{P}^{*}_{\ell,\alpha} A_1-\alpha \tfrac{m_{\ell}}{2\gamma_2}
    \tilde{P}^{*}_{\ell,\alpha}B_1 B_1\T \tilde{P}^{*}_{\ell,\alpha} +
    \tfrac{1}{4} C_1\T
    C_1 +\tfrac{1}{2} \tilde{\eps} \tilde{P}^{*}_{\ell,\alpha} =  0
  \]
  It is now easy to check that $\tilde{P}^{*}_{\ell,\alpha}$ has a
  nice diagonal structure. Let $P_{\ell,1}$ satisfy
  \[
    A_{1}\T P_{\ell,22} + P_{u,1} A_{1} -
    \tfrac{m_{\ell}}{2\gamma_2} P_{\ell,22} B_1B_1\T
    P_{\ell,22}  + \tfrac{1}{2} m_{u}^{-1} C_1\T C_1 + \tfrac{1}{2}
    \eps P_{\ell,22} =0.    
  \]
  Next, using some straightforward algebra, we can show that 
  \[
    \tilde{P}^{*}_{\ell,\alpha} = D_1 P_{\ell,22}
    D_1.
  \]
  We have now shown that:
  \[
    \tilde{P}_{\alpha}  \geq \alpha^{-1} \tilde{P}_{\ell,1} \geq 0
  \]
  where $P_{\ell,1}=\Gamma_s\T P_{\ell,1} \Gamma_s$ and
  \begin{equation}\label{fgbound}
    \tilde{P}_{22,\alpha} \geq D_1 P_{\ell,22}
    D_1
  \end{equation}
  Let
  \[
    \tilde{P}_{\ell,1} = \begin{pmatrix}
      \tilde{P}_{11,\ell,1} & \tilde{P}_{12,\ell,1} \\
      \tilde{P}_{21,\ell,1} & \tilde{P}_{22,\ell,1} 
    \end{pmatrix} \geq 0
  \]
  Since $\dim \ker \tilde{P}_{\ell,1}= \dim \ker
  \tilde{P}_{11,\ell,1}$ we find that there exists $\nu$ such that
  \[
    \tilde{P}_{\alpha}  \geq  \alpha^{-1} \tilde{P}_{\ell,1}  \geq  \begin{pmatrix}
      \nu \alpha^{-1} \tilde{P}_{11,\ell,1} & 0 \\
      0                         & 0
    \end{pmatrix}
  \]
  Combining with \eqref{fgbound} this implies
  \[
    \tilde{P}_{\alpha}  \geq \begin{pmatrix}
      \nu\alpha^{-1} \tilde{P}_{11,\ell,1} &  0 \\ 0 &  D_1 P_{\ell,22}
      D_1
    \end{pmatrix}.
  \]
  Choosing
  \[
    P_\ell = \begin{pmatrix}
      \nu \tilde{P}_{11,\ell,1} &  0 \\ 0 &  P_{\ell,22}
    \end{pmatrix}  
  \]
  then gives the required lower bound in the lemma.
\end{proof}

\begin{lemma}\label{tty5}
  Consider a linear system
  \[\
    \begin{system*}{ccl}
      \dot{x}_1 &=& A_1 x_1 + B_1u_1 + E_1w_1 \\
      y_1 &=& C_{1} x_1
    \end{system*}
  \]
  with the structure outlined Lemma \ref{scb}. Assume $w_1$ is bounded
  and $\displaystyle\lim_{t\rightarrow \infty} \alpha = \infty$. Then
  for any $\eps>0$ there exists $T$ and $\delta>0$ such that for
  $t_0>T$,
  \[
    \int_{t_0}^{t_0+1}  \alpha^{-1} u_1\T u_1\, \textrm{d}t
    <\delta\qquad \text{and}\qquad
    \int_{t_0}^{t_0+1}  y_1\T y_1 \,\textrm{d}t
    <\delta
  \]
  that
  \begin{equation}\label{VVb}
    \int_{t_0}^{t_0+1} V_{\alpha}(t)\, \textrm{d}t \leq \eps
  \end{equation}
  where
  \begin{equation}\label{VCCV}
    V_{\alpha} = x_1\T \begin{pmatrix} 
      I & 0      & \cdots & 0      \\
      0         & \beta^4 I & \ddots & \vdots \\
      \vdots    & \ddots & \ddots & 0      \\
      0         & \cdots & 0      & \beta^{4n_1-4} I
    \end{pmatrix} x_1,
  \end{equation}
  with $\beta=\alpha^{-1/(4n_1)}$. 
\end{lemma}

\begin{proof}
  This is a consequence of the Gagliardo-Nirenberg inequality. Denote
  by $x_{1,0}$ the value of the state for given $u_1$ and
  given initial conditions if we set $w_1=0$. while $x_{1,w_1}$ denotes the value of the
  state for given $w_1$ if we set $u_1=0$ and the initial conditions
  to zero. Let $y_{1,0}=Cx_{1,0}$ and $y_{1,w_1}=Cx_{1,w_1}$. We then have:
  \[
    x_1 =x_{1,w_1}+x_{1,0}.
  \]
  Since $w_1$ is bounded we find that $x_{1,w_1}$ is bounded. This
  implies that $y_{1,0}$ is bounded given that $y_1$ is small and
  $y_{1,w_1}$ is bounded. The  Gagliardo-Nirenberg inequality (see
  \cite{nirenberg}) implies there exits a constants $C_1, C_2>0$ (independent of
  initial conditions, $w_i$ and $u_i$) such that 
  \[
    \| y^{(j)}_{1,0} \|_2 \leq C_1 \| y_{1,0}^{(m)} \|_2^{j/m} \| y_{1,0}
    \|_2^{1-j/m} + C_2 \| y_{1,0} \|_2,
  \]
  which yields:
  \[
    \beta_i^{2j} \| y^{(j)}_{1,0} \|_2 \leq C_1 \delta^{j/2m} \| y_{1,0}
    \|_2^{1-j/m} + C_2 \beta_i^{2j} \| y_{1,0} \|_2.
  \]  
  For $\alpha$ large enough (i.e.\ $T$ large enough) we find that
  $\beta$ is arbitrarily small and hence for $\delta$ small enough, we
  have that \eqref{VVb} is satisfied.
\end{proof}

\bibliographystyle{plainnat}

\end{document}